\documentclass[12pt]{article}
\usepackage{setspace}
\usepackage[utf8]{inputenc}
\usepackage{arxiv}
\usepackage[english]{babel}  
\usepackage[colorlinks=true, linkcolor=blue, citecolor = blue]{hyperref} \usepackage[ruled,vlined]{algorithm2e}    
\usepackage{url}  % simple URL typesetting
\usepackage{booktabs}       % professional-quality 
\usepackage{amsthm}

\newtheorem{theorem}{Theorem}
\usepackage{amsfonts}       % blackboard math symbols
\usepackage{nicefrac}       % compact symbols for 1/2, etc.
\usepackage{microtype}
\usepackage[top=1 in,bottom=1in, left=1 in, right=1 in]{geometry}% microtypography
\usepackage{graphicx}
\usepackage{subfig}
\usepackage{rotating}
\usepackage{float}
\usepackage{wrapfig}
\usepackage{amsmath}
\usepackage{amssymb}
\usepackage[round]{natbib}
\usepackage{tabularx}
\usepackage[export]{adjustbox}
\usepackage{multicol}
\usepackage{color}
\usepackage{multirow}
\usepackage[thinlines]{easytable}
\usepackage{multirow}
\usepackage[table,xcdraw]{xcolor}
\hypersetup{
    colorlinks=true,
    linkcolor=blue,
    filecolor=blue,      
    urlcolor=blue,
}
\usepackage{tcolorbox}
\usepackage{fancyhdr}
\newcommand\independent{\protect\mathpalette{\protect\independenT}{\perp}}
\def\independenT#1#2{\mathrel{\rlap{$#1#2$}\mkern2mu{#1#2}}}
\title{\textbf{A Framework for Understanding Selection Bias in Real-World  Healthcare Data}}
\title{\textbf{A Framework for Understanding Selection Bias in Real-World  Healthcare Data}}
\author{
  Ritoban Kundu \\
  PhD Student\\
  Department of Biostatistics \\ 
  University of Michigan, Ann Arbor, USA.\\
  \texttt{kundur@umich.edu} \\
  \And 
  \text{Xu Shi} \\
  Assistant Professor\\
  Department of Biostatistics \\ 
  University of Michigan, Ann Arbor, USA.\\
  \texttt{shixu@umich.edu} \\
  \And 
  Jean Morrison \\
  Assistant Professor\\
  Department of Biostatistics \\ 
  University of Michigan, Ann Arbor, USA.\\
  \texttt{jvmorr@umich.edu} \\
  \And
  Jessica Barrett\\
  MRC Investigator\\
  Biostatistics Unit, Medical Research Council\\
  University of Cambridge, UK.\\
  \texttt{jessica.barrett@mrc-bsu.cam.ac.uk}\\
 \And
  Bhramar Mukherjee \\
  Professor\\
  Department of Biostatistics and Epidemiology\\ 
  University of Michigan, Ann Arbor, USA.\\
  \texttt{bhramar@umich.edu} \\
}
\vspace{0.2in}

\begin{document}
\maketitle
%\vspace{0.2in}

%\noindent {\bf Short Running Head:} Selection Bias in Electronic Health Records Data
%\newpage
\centerline{\Large \bf Abstract}
\vspace{0.1in}
\noindent
Using administrative patient-care data such as Electronic Health Records (EHR) and medical/ pharmaceutical claims for population-based scientific research has become increasingly common. With vast sample sizes leading to very small standard errors, researchers need to pay more attention to potential biases in the estimates of association parameters of interest, specifically to biases that do not diminish with increasing sample size. Of these multiple sources of biases, in this paper, we focus on understanding selection bias. We present an analytic framework using directed acyclic graphs for guiding applied researchers to dissect how different sources of selection bias may affect  estimates of the association between a binary outcome and an exposure (continuous or categorical) of interest. We consider four easy-to-implement weighting approaches to reduce selection bias with accompanying variance formulae. We demonstrate through a simulation study when they can rescue us in practice with analysis of real world data. We compare these methods using a data example where our goal is to estimate the well-known association of cancer and biological sex, using EHR from a longitudinal biorepository at the University of Michigan Healthcare system. We provide annotated R codes to implement these weighted methods with associated inference.

\vspace{0.2in}


\section{Introduction} Massive amounts of data are routinely collected in health care clinics for administrative and billing purposes. Longitudinally  varying time-stamped observational patient care data such as electronic health records (EHRs) allow researchers from various disciplines to run agnostic queries \citep{denny2013systematic,hoffmann2017genome} or validate hypothesis driven questions in large databases \citep{roberts2022estimating,shen2022efficacy}.
%% gwas and hypothesis driven problem vaccine effectiveness, drug
However these observational studies pose several practical challenges for health research which can negatively impact internal validity and external generalizability of the results \citep{beesley2020analytic}. Without properly accounting for potential sources of biases and study design issues, association analysis using these data can  result in spurious findings \citep{madigan2014systematic} and misguided policies \citep{wang2020characterizing}. One major challenge in removing or reducing bias  in these studies lies in the fact that there can be several potential causes of bias that may be simultaneously at play in an analysis done with a given real world dataset. With larger datasets at researchers' fingertips, the  impact of bias relative to variance becomes even more pronounced. This phenomenon has recently been termed as the ``curse of large n'' \citep{kaplan2014big,bradley2021unrepresentative}.  The common sources of biases related to EHR studies do not disappear with increased sample size and thus with increased precision comes the increased possibility of achieving incorrect inference. This is also termed as the ``big data paradox'' \citep{meng2018genome}. In studies with large  n, bias often dominates the mean squared error of an estimator and thus we need to update our statistical thinking to focus on  strategies for reducing bias  as opposed to the classical thinking around reducing variance.\\

\noindent
Given a scientific question and access to a potentially large and messy database, we first need to define a target population of inference. A careful investigator then needs to think about the possible sources of bias that are most critical for the underlying question at hand. Selection bias, missing data, clinically informative patient encounter process, confounding, lack of consistent data harmonization across cohorts, true heterogeneity of the studied populations, registration of start time or  definition of time zero, and misclassification bias due to imperfect phenotyping are some of the most common sources of bias in EHR. An overview of the different kinds of biases mentioned above with relevant references are given in Table \ref{tab:Table 1}.\\

\noindent
In this paper, we focus on understanding and tackling one major source of bias, namely selection bias in administrative healthcare data. The selection mechanism underlying the question ``Who is in my study sample?'' may vary widely across the different sources of real-world data. For example, in using EHRs in the United States, where there is no universal healthcare or nationally integrated clinical data warehouse, one challenge is understanding factors that influence selection into a given study such as health care seeking behavior and insurance coverage \citep{haneuse2016general,rexhepi2021cancer,heintzman2015supporting,heart2017review}.
% %policy relevant reference
Population-based biobanks such as the UK Biobank that are based on invitation to volunteers can lead to specific types of biases such as healthy control bias \citep{fry2017comparison,van2022reweighting}. Nationally representative studies such as the NIH All of Us often have a purposeful sampling strategy that leads to, say, oversampling certain underrepresented subgroups \citep{all2019all}. In contrast, medical center and health system based studies attempt to recruit patients meeting specific criteria within the health system, often through multiple disease/treatment clinics. This leads to enrichment of certain diseases in the study sample \citep{zawistowski2021michigan,pendergrass2011visual}.  In addition, there is non-response and consenting bias among those who are approached to participate in the study.  Since the process of selection into  each study is unique and often unknown, conventional survey sampling techniques to handle probability samples with known sampling/survey weights are not generally applicable for such type of observational data which can be predominantly considered as non probability samples (samples where selection probabilities are unknown) \citep{beesley2022statistical,chen2020doubly}.\\

\noindent
If the issue of selection bias is ignored, it can negatively impact downstream inference \citep{kleinbaum1981selection,christensen1992selection}. Due to unknown selection weights, naive inference from these non probability samples is generally not directly transportable to the target population. On the other hand, it is important to know when the selection process can be ignored and we can proceed with straightforward naive analysis. A structural framework to study selection bias using Directed Acyclic Graphs (DAG) was introduced in \citet{hernan2004structural}. We use this approach to study some common scenarios of selection mechanisms and their effects on estimates of association between a binary disease outcome and an exposure of interest (after adjusting for a set of confounders/covariates) specifically for real world data. We consider a logistic regression model as the underlying disease outcome model.\\

\noindent
After dissecting the selection mechanism to the best of our ability, we need to think about methods that are available to address/account for selection bias.  Some of these methods rely on having individual level data from an external probability-sample. \citet{chen2020doubly}, \citet{wang2021information} adopted the method of pseudolikelihood based estimating equations to account for selection bias in estimating population mean of a response variable in non probability samples using individual level data from an external probability sample. On the other hand, beta regression generalized linear model (glm)  \citep{ferrari2004beta} was used to estimate selection probabilities in \citet{elliot2009combining} and \citet{beesley2022statistical}. When only summary level information are available on an external probability sample, some methods in survey sampling, such as post stratification, raking and calibration techniques as in \citet{kim2010calibration,deville1992calibration,montanari2005nonparametric} can be modified to reduce selection bias in non probability samples \citep{beesley2022statistical,beesley2022case}. We consider simulation settings reflecting common selection mechanisms represented by the DAGs and assess the bias-reduction properties of four of these weighting methods using the general framework of inverse-probability weighted (IPW) logistic regression. The methods differ in how the weights are constructed and what type of external data are required. We also present variance formulae associated with each weighting method.\\

\noindent Using EHR data from a longitudinal biorepository at the University of Michigan Healthcare system, the Michigan Genomics Initiative (MGI) and auxiliary data from a nationally representative probability sample study to define the selection weights, we illustrate how and when the weighted methods enable us to get closer to the truth compared to naive unweighted logistic regression.\\

\noindent
The rest of the paper is organized as follows. In section \ref{sec:DAGs}, we describe the study setting and four common types of selection DAGs. The expected extent of biases under different selection DAGs in a logistic regression outcome-exposure association model is studied using an analytical expression that relates the parameters of the true association model in overall population to the model restricted to the selected sub-population (without any adjustment for selection). Four variants of weighted logistic regression methods with individual or summary level external data (targeted to reduce selection bias in association parameters of interest) are described in sections \ref{sec:weightedlogistic}- \ref{sec:summethods}. We also present variance formulae for each method in section \ref{sec:avar}. In section \ref{sec:onesimu}, we conduct a simulation study comparing the four methods under different selection DAGs. In section \ref{sec:realex}, we estimate the association between cancer and biological sex in the Michigan Genomics Initiative Data using the four IPW methods discussed in the previous sections with associated confidence intervals. We conclude with a brief discussion in section \ref{sec:discussion}.

\section{Methods}\label{sec:methods}
\subsection{Notation}\label{sec:DAGs}

Our main focus is on the relationship between a binary disease indicator $D$ and a set of covariates $\boldsymbol Z$ in a target population from which the internal non-probability sample is drawn. Selection is denoted by a binary indicator $S=1$ which is assumed to be driven by a set of covariates and may also depend on $D$. Figure \ref{fig:models} summarizes the structures of the disease and selection models. $\boldsymbol Z_1$ is the subset of $\boldsymbol Z$ present only in the disease model,  $\boldsymbol Z_2$ influences the disease indicator $D$ and may influence the selection indicator $S$. While $\boldsymbol W$ denotes covariates  present only in the selection model. The primary disease model of interest is :
\begin{equation}
    \text{logit}(P(D=1|\boldsymbol Z_1,\boldsymbol Z_2))=\theta_0+\underset{}{\boldsymbol\theta_1} \boldsymbol Z_1 +\boldsymbol \theta_2 \boldsymbol Z_2\cdot \label{eq:eq1}
\end{equation}
Selection into the internal sample is driven by a probability mechanism $P(S=1|\boldsymbol Z_2,\boldsymbol W,D)$ which is allowed to be completely nonparametric.\\

\noindent
 Our desired target is $D|\boldsymbol Z_1,\boldsymbol Z_2$ as in equation \eqref{eq:eq1}. However, we can only fit $D|\boldsymbol Z_1,\boldsymbol Z_2,S=1$.  As derived in Supplementary section \ref{sec:pr1}, if equation \eqref{eq:eq1} holds one can relate the true model parameters and the ones for  naively fitted model (conditional on $S=1$) by the following key relationship: 
\begin{equation}
\text{logit}(P(D=1|\boldsymbol Z_1,\boldsymbol Z_2,S=1))=\theta_0+\underset{}{\boldsymbol\theta_1} \boldsymbol Z_1 +\boldsymbol\theta_2\boldsymbol Z_2 + \text{log}(r(\boldsymbol Z_1,\boldsymbol Z_2))\cdot  \label{eq:eq2}
\end{equation}
where $$r(\boldsymbol Z_1,\boldsymbol Z_2)=\frac{P(S=1|D=1,\boldsymbol Z_1,\boldsymbol Z_2)}{P(S=1|D=0,\boldsymbol Z_1,\boldsymbol Z_2)}\cdot $$ describes how disease predictors $\boldsymbol Z_1,\boldsymbol Z_2$ modify the selection mechanism.  Unless $r(\boldsymbol Z_1,\boldsymbol Z_2)$ is a constant function of $\boldsymbol Z_1,\boldsymbol Z_2$ (like in a population-based case-control study), estimates obtained from the naive unweighted logistic regression model of $D$ on $\boldsymbol Z_1$ and $\boldsymbol Z_2$ based on just the internal data lead to biased estimates of $\boldsymbol\theta_1$, $\boldsymbol\theta_2$, or both. A common example of such predictor outcome-dependent selection bias is case control studies where factors like education could influence the likelihood of volunteering as controls \citep{geneletti2009adjusting,kleinbaum1981selection}.

\subsection{Selection DAGs}\label{sec:DAGselect}
We study the extent of bias introduced by the additional term $r(\boldsymbol Z_1,\boldsymbol Z_2)$ in equation \eqref{eq:eq2} when we use naive logistic regression on the selected sample, namely $D|\boldsymbol Z_1,\boldsymbol Z_2,S=1$. Note that, 
\begin{align}
   r(\boldsymbol Z_1,\boldsymbol Z_2) &=\frac{P(S=1|D=1,\boldsymbol Z_1,\boldsymbol Z_2)}{P(S=1|D=0,\boldsymbol Z_1,\boldsymbol Z_2)}\nonumber\\
   & =\frac{\int P(S=1|D=1,\boldsymbol Z_1,\boldsymbol Z_2,\boldsymbol W)dP(\boldsymbol W|D=1,\boldsymbol Z_1,\boldsymbol Z_2)}{\int P(S=1|D=0,\boldsymbol Z_1,\boldsymbol Z_2,\boldsymbol W)dP(\boldsymbol W|D=0,\boldsymbol Z_1,\boldsymbol Z_2)}\cdot \label{eq:rexp}
\end{align}
We study the bias in the naive approach under some plausible DAGs with increasing levels of complexity in dependencies among $D$, $S$, $\boldsymbol Z_1$, $\boldsymbol Z_2$ and $\boldsymbol W$. We simplify the expression of $r(\boldsymbol Z_1,\boldsymbol Z_2)$ in equation \eqref{eq:rexp} under the different DAGs introduced in Figure \ref{fig:Fig2}.

\subsubsection*{Example DAG 1: unbiased case}\label{sec:setup1method}
  Under DAG 1 in Figure \ref{fig:Fig2} the arrows from ($\boldsymbol Z_1$, $\boldsymbol Z_2$) to $(\boldsymbol W,S)$  do not exist. In addition, $D$ does not directly affect $\boldsymbol W$. This implies, none of the disease model predictors $\boldsymbol Z_1$ and $\boldsymbol Z_2$ affect the selection mechanism; As shown in Supplementary Section \ref{sec:pr2.1}, the expression of $r(\boldsymbol Z_1,\boldsymbol Z_2)$ in equation \eqref{eq:rexp} simplifies to a constant denoted by $r$ :
\begin{align*}
   r(\boldsymbol Z_1,\boldsymbol Z_2) =r=\frac{P(S=1|D=1)}{P(S=1|D=0)}\cdot 
\end{align*}
 In this case, estimates obtained from an unweighted logistic regression of $D$ on $\boldsymbol Z_1$ and $\boldsymbol Z_2$ in the selected sample (conditional on $S=1$) are unbiased for $\boldsymbol\theta_1$ and $\boldsymbol\theta_2$ even without adjusting for any selection bias. However, the intercept term estimate is biased for $\theta_0$ with the bias being the offset term $\log(r)$. This is equivalent to the results that are well-known for a case-control study \citep{cornfield1959smoking}.

\subsubsection*{Example  DAG 2: $\boldsymbol Z_1\rightarrow \boldsymbol W$ arrow induced bias for coefficient of $\boldsymbol Z_1$ }\label{sec:setup2method}
Under DAG 2 in Figure \ref{fig:Fig2} we observe that there is an additional direct dependence from $\boldsymbol Z_1$ to $\boldsymbol W$ compared to DAG 1.  As shown in Supplementary Section \ref{sec:pr2.2}, the expression for $r(\boldsymbol Z_1,\boldsymbol Z_2)$ under this scenario reduces to, 
\begin{align*}
   r(\boldsymbol Z_1,\boldsymbol Z_2)=\frac{\int P(S=1|D=1,\boldsymbol W)dP(\boldsymbol W|\boldsymbol Z_1)}{\int P(S=1|D=0,\boldsymbol W)dP(\boldsymbol W|\boldsymbol Z_1)}\cdot 
\end{align*}
With the introduction of additional arrow between $\boldsymbol Z_1$ and $\boldsymbol W$, the function  $r(\boldsymbol Z_1,\boldsymbol Z_2)$ depends on $\boldsymbol Z_1$ through $P(\boldsymbol W|\boldsymbol Z_1)$ but does not depend on $\boldsymbol Z_2$. Consequently estimates from a naive logistic regression of  $D|\boldsymbol Z_1,\boldsymbol Z_2,S=1$ leads to biased estimates of $\boldsymbol \theta_1$ but not of $\boldsymbol \theta_2$. Similarly if $\boldsymbol Z_2 \rightarrow \boldsymbol W$ is present and $\boldsymbol Z_1 \rightarrow \boldsymbol W$ is absent, then using identical arguments, we obtain unbiased estimates for $\boldsymbol\theta_1$ and biased for $\boldsymbol\theta_2$.

\subsubsection*{Example DAG 3: $\boldsymbol Z_1\rightarrow \boldsymbol W$ and $\boldsymbol Z_2\rightarrow S$ induced bias for coefficients of $\boldsymbol Z_1$ and $\boldsymbol Z_2$ }\label{sec:setup3method}
Under DAG 3 in Figure \ref{fig:Fig2}, ${\boldsymbol Z_2}$ has a direct causal pathway to $S$ which leads to increase in the strength of dependence between the selection and disease models when compared to DAG 2. As shown in Supplementary section \ref{sec:pr2.3}, the expression for $r(\boldsymbol Z_1,\boldsymbol Z_2)$ in this case is,
\begin{align*}
   r(\boldsymbol Z_1,\boldsymbol Z_2)=\frac{\int P(S=1|D=1,\boldsymbol W,\boldsymbol Z_2)dP(\boldsymbol W|\boldsymbol Z_1)}{\int P(S=1|D=0,\boldsymbol W,\boldsymbol Z_2)dP(\boldsymbol W|\boldsymbol Z_1)}\cdot 
\end{align*}
Therefore, $r(\boldsymbol Z_1,\boldsymbol Z_2)$ is a function of both $\boldsymbol Z_1$ and $\boldsymbol Z_2$. The dependence on $\boldsymbol Z_1$ is through $P(\boldsymbol W|\boldsymbol Z_1)$. The dependence on $\boldsymbol Z_2$ is through $P(S|D,\boldsymbol W,\boldsymbol Z_2)$. The naive unweighted logistic regression method fails to provide unbiased estimates of both $\boldsymbol\theta_1$ and $\boldsymbol\theta_2$. In case where $\boldsymbol Z_2\rightarrow S$ exists  but there is no arrow $\boldsymbol Z_1\rightarrow \boldsymbol W$ then estimates for $\boldsymbol\theta_2$ will be biased and unbiased for $\boldsymbol\theta_1$.
\subsubsection*{Example DAG 4: strong dependence, increased bias for coefficients of $\boldsymbol Z_1$ and $\boldsymbol Z_2$ }\label{sec:setup4method}
DAG 4 in Figure \ref{fig:Fig2} corresponds to a situation where the dependence between the selection and disease model is the most complex among the four selection DAGs we considered. As shown in Supplementary section \ref{sec:pr2.4}, the expression for $r(\boldsymbol Z_1,\boldsymbol Z_2)$ in this case is given by,
\begin{align*}
   r(\boldsymbol Z_1,\boldsymbol Z_2)=\frac{\int P(S=1|D=1,\boldsymbol W,\boldsymbol Z_2)dP(\boldsymbol W|\boldsymbol Z_1,\boldsymbol Z_2,D=1)}{\int P(S=1|D=0,\boldsymbol W,\boldsymbol Z_2)dP(\boldsymbol W|\boldsymbol Z_1,\boldsymbol Z_2,D=0)}\cdot 
\end{align*}
Here, $r(\boldsymbol Z_1,\boldsymbol Z_2)$ depends on $\boldsymbol Z_2$ via $P(S|D,W,\boldsymbol Z_1,\boldsymbol Z_2)$ and $P(\boldsymbol W|\boldsymbol Z_1,\boldsymbol Z_2,D)$. Consequently, the estimated coefficient of $\boldsymbol Z_2$ from a naive unweighted logistic regression of $D|\boldsymbol Z_1,\boldsymbol Z_2,S=1$ potentially becomes more biased for $\boldsymbol\theta_2$ compared to the other DAGs conditioned on the fact that the strength of associations among the variables remain same across the different DAGs. However, the dependence of $r(\boldsymbol Z_1,\boldsymbol Z_2)$ on $\boldsymbol Z_1$ is only through $P(\boldsymbol W|\boldsymbol Z_1,\boldsymbol Z_2,D)$ potentially leading to less bias in estimate of $\boldsymbol \theta_1$ compared to estimate of $\boldsymbol\theta_2$.\\

\noindent
\textbf{Remark:} The issue of correcting for selection bias becomes more challenging in our setting due to the joint dependence of $S$ on both the disease indicator $D$ and other covariates $(\boldsymbol Z_2,\boldsymbol W)$. If in fact there was no arrow $D\rightarrow S$, then conditioned on $(\boldsymbol Z_1,\boldsymbol Z_2)$, all paths between $D$ and $S$ are blocked leading to d-separation of $D$ and $S$ in DAGs 1, 2 and 3, implying $D\independent S|(\boldsymbol Z_1,\boldsymbol Z_2)$. This also implies $P(D|\boldsymbol Z_1,\boldsymbol Z_2,S=1)=P(D|\boldsymbol Z_1,\boldsymbol Z_2)$. Thus, estimates from fitting the naive unweighted model in equation \eqref{eq:eq2} are consistent for the true parameters, $\boldsymbol \theta_1$ and $\boldsymbol \theta_2$. On the other hand for DAG 4, the estimates are biased since conditioned on $(\boldsymbol Z_1,\boldsymbol Z_2)$ the path $D\rightarrow\boldsymbol W\rightarrow \boldsymbol S$ is still unblocked.\\

\noindent
Now that we have established that fitting a model on the selected sample namely, $D|\boldsymbol Z_1,\boldsymbol Z_2,S=1$ can generally (for example in DAGs 2, 3 and 4) lead to biased estimates of the true model parameters $\boldsymbol\theta_1$ and $\boldsymbol\theta_2$  in the target population, we consider four easy-to-use weighted logistic regression methods that address selection bias. The methods differ in terms of their construction of weights and the type of external data required.
\subsection{Weighted Logistic Regression}
\label{sec:weightedlogistic}
In this section, we use the following notation. We assume that we have an internal nonprobability sample with selection indicator $S$ and an external probability sample with selection indicator $S_{\text{ext}}$ drawn from the same target population. Figure \ref{fig:population} is a schematic representation of the assumed scenario. The internal and external samples may or may not have overlap ($S_{\text{both}}=1$ or 0 respectively as in Figure \ref{fig:population}).\\

\noindent
Inverse probability weighted (IPW) regression is a potential remedy to adjust for selection bias and obtain less biased estimates of parameters in the disease model \citep{beesley2022statistical,haneuse2016general}. 
Let $N$ be the size of the target population.  Let $\boldsymbol X=(D,\boldsymbol Z_2,\boldsymbol W)$ denote the selection model covariates and $\pi(\boldsymbol X)=P(S=1|\boldsymbol X)$ denote probability of selection into internal sample. Therefore the size of the internal non-probability sample is given by $\sum_{i=1}^N S_i$. Let $\boldsymbol Z=(1,\boldsymbol Z_1,\boldsymbol Z_2)$ denotes the disease model covariates in equation \eqref{eq:eq1} with parameters denoted by $\boldsymbol\theta$. Thus $\text{dim}(\boldsymbol \theta)=\text{dim}(\boldsymbol Z_1)+\text{dim}(\boldsymbol Z_2)+1$.
In IPW logistic regression, the estimating equations are given by
\begin{equation}
     \frac{1}{N}\sum_{i=1}^{N}\frac{S_i}{\pi(\boldsymbol X_i)}\left\{D_i\boldsymbol Z_i-\frac{e^{\boldsymbol\theta'\boldsymbol Z_i}}{(1+e^{\boldsymbol\theta'\boldsymbol Z_i})}\cdot \boldsymbol Z_i\right\}=\mathbf{0}.\label{eq:eq6}
\end{equation}
where $i$ corresponds to the $i^{\text{th}}$ individual in the target population. The consistency of the estimate $\widehat{\boldsymbol \theta}$ obtained as a solution to equation \eqref{eq:eq6} is presented in Supplementary Section \ref{sec:ipw}.\\

\noindent
For non probability samples, the selection probability of individual $i$, given by $\pi(\boldsymbol X_i)$ is unknown. Since there is no information available on participants who are not selected into the internal study $(S=0)$, the estimation of $\pi(\boldsymbol X)$ requires some form of external information. Auxiliary external data are typically available in two forms, either individual level data or summary level statistics.  Moreover, since the external sample is a probability sample drawn from the target population, we assume that we have access to the known sampling probabilities, say $\pi_{\text{ext}}$.\\

\noindent
In Sections \ref{sec:indimethods} and \ref{sec:summethods}, we describe four methods to estimate the selection probabilities $\pi(\boldsymbol X)$  depending on the nature of available external information. All four methods adopt a two step process: the first step involves obtaining estimates of the selection probabilities, $\widehat{\pi}(\boldsymbol X)$; the second step is estimation of disease model parameters $\boldsymbol\theta$ using the weighted score equation \eqref{eq:eq6} with $\pi(\boldsymbol X)$ replaced by $\widehat{\pi}(\boldsymbol X)$. A summary of all the methods including the unweighted and the four weighted ones are given in Table \ref{tab:Table2}.

\subsection{Estimation of Weights Using Individual Level External Data}\label{sec:indimethods}
In this subsection, we consider two methods to account for selection bias in the internal sample using individual level data from an external probability sample. The first one is adaptation of a pseudolikelihood based estimating equation approach originally proposed in \citet{chen2020doubly} for estimation of population mean of a response variable. We modified this technique to our context. The second one is based on simplex regression method \citep{barndorff1991some}, as an improvement over beta regression that has been previously used in this problem \citep{elliot2009combining,beesley2022statistical}.

\subsubsection{Pseudolikelihood-based estimating equation (PL)}\label{sec:chen}
The selection indicator variable into the internal sample for the $i^{\text{th}}$ individual in the population, $S_i|\boldsymbol X_i$ is a bernoulli random variable with success probability $\pi(\boldsymbol X_i)$. In this method, we assume a parametric model for $\pi(\boldsymbol X)$ indexed by parameters $\boldsymbol\alpha$, specified by $\pi(\boldsymbol X)=\pi(\boldsymbol X,\boldsymbol \alpha)=\frac{e^{\boldsymbol \alpha'\boldsymbol X}}{1+e^{\boldsymbol \alpha'\boldsymbol X}}$.\\

\noindent
The likelihood function of $S_i|\boldsymbol X_i$ $\forall i \in \{1,2,..,N\}$ is given by
\begin{equation}
\mathcal{L}(\boldsymbol \alpha|\{S\}_{i=1}^N,\{\boldsymbol X\}_{i=1}^N)=\prod_{i=1}^{N}\{\pi(\boldsymbol X_i,\boldsymbol\alpha)\}^{S_i}\cdot \{1-\pi(\boldsymbol X_i,\boldsymbol\alpha)\}^{1-S_i}.\label{eq:eq8}
\end{equation}
Equivalently, the log likelihood is
\begin{equation}
l(\boldsymbol \alpha|\{S\}_{i=1}^N,\{\boldsymbol X\}_{i=1}^N)=\sum_{i=1}^{N}S_i \cdot \text{log}\left\{\frac{\pi(\boldsymbol X_i,\boldsymbol \alpha)}{1-\pi(\boldsymbol X_i,\boldsymbol \alpha)}\right\}+\sum_{i=1}^N \text{log}\{1-\pi(\boldsymbol X_i,\boldsymbol \alpha)\}.\label{eq:eq9}
\end{equation}
The first term of the above equation only involves values of $\boldsymbol X$ from the internal non-probability sample. Ideally, the selection parameters $\boldsymbol \alpha$ would have been obtained by maximizing the above log likelihood in equation  \eqref{eq:eq9}, however the second part of the log likelihood $\sum_{i=1}^N \text{log}(1-\pi(\boldsymbol X_i,\boldsymbol \alpha))$ cannot be calculated solely based on the available data from the internal sample. This term require the values of $\boldsymbol X$ from $S=0$ sample.
\citet{chen2020doubly} provide an approximation to the log likelihood using the following expression,
\begin{equation}
\sum_{i=1}^{N}S_i\cdot \text{log}\left\{\frac{\pi(\boldsymbol X_i,\boldsymbol \alpha)}{1-\pi(\boldsymbol X_i,\boldsymbol \alpha)}\right\}+\sum_{i=1}^N \left[\left(\textcolor{red}{\frac{S_{\text{ext},i}}{\pi_{\text{ext},i}}}\right)\text{log}\{1-\pi(\boldsymbol X_i,\boldsymbol \alpha)\}\right].
\label{eq:eq10}
\end{equation}
Since the exact sampling weights of the external probability sample are known, the second term of equation \eqref{eq:eq10} is an unbiased estimator of the second term in equation \eqref{eq:eq9}. Using the logistic form of the internal selection model $\pi(\boldsymbol X,\boldsymbol \alpha)$ and differentiating equation \eqref{eq:eq10} with respect to $\boldsymbol \alpha$, we obtain the following estimating equation
\begin{equation}
        \frac{1}{N}\sum_{i=1}^N S_i\boldsymbol X_i- \frac{1}{N}\cdot \sum_{i=1}^ N \left(\frac{S_{\text{ext},i}}{\pi_{\text{ext},i}}\right)\cdot \pi(\boldsymbol X_i,\boldsymbol\alpha)\cdot \boldsymbol X_i=\underset{}{\mathbf{0}}.\label{eq:eq7}
\end{equation}
Newton-Raphson method is used to estimate $\boldsymbol \alpha$ from the above equation. We obtain the estimates of internal selection probabilities,
$\pi(\boldsymbol X_i,\boldsymbol\alpha)$ by plugging the estimates of $\boldsymbol \alpha$ in the logistic functional form of the selection model.
\subsubsection{Simplex Regression (SR)}
The main idea underlying this method is based on the identity
\begin{equation}
       \pi(\boldsymbol X)=P(S=1|\boldsymbol X)=P(S_{\text{ext}}=1|\boldsymbol X)\cdot \left(\frac{p_{11}(\boldsymbol X)+p_{10}(\boldsymbol X)}{p_{11}(\boldsymbol X)+p_{01}(\boldsymbol X)}\right).\label{eq:eq11}
 \end{equation}
where, $p_{jk}(\boldsymbol X)=P(S=j,S_{\text{ext}}=k|\boldsymbol X,S=1\text{  or   }S_{\text{ext}}=1)$. The proof of this above identity \eqref{eq:eq11}, is provided in Supplementary Section \ref{sec:simproof}. From equation \eqref{eq:eq11}, we observe that we need to estimate $P(S_{\text{ext}}=1|\boldsymbol X)$ and $p_{11}(\boldsymbol X), p_{10}(\boldsymbol X), p_{01}(\boldsymbol X)$ for each internal sample individual to calculate the internal selection probabilities $\pi(\boldsymbol X)$. We adopt two separate regression frameworks to model  the dependencies of $P(S_{\text{ext}}=1|\boldsymbol X)$ and $p_{jk}(\boldsymbol X)$ on $\boldsymbol X$ respectively.\\

\noindent
\textbf{Estimation of $P(S_{\text{ext}}=1|\boldsymbol X)$ :} We used simplex regression \citep{barndorff1991some} to model dependence of  $P(S_{\text{ext}}=1|\boldsymbol X)$ on $\boldsymbol X$. Simplex regression is one of the glm regression methods with proportions as the response.  The main idea is to fit the best possible model in the external sample to the known design probabilities $\pi_{\text{ext}}$ as a function $\boldsymbol X$, say, $\text{logit}(\mathbb{E}(\pi_{\text{ext}}|\boldsymbol X))=\boldsymbol \delta' \boldsymbol X$. The parameter $\boldsymbol\delta$ is estimated by maximizing the following likelihood function based on the external probability sample obtained using the simplex distribution
\begin{equation}
    \prod_{i=1}^N S_{\text{ext},i}\cdot[2\pi\sigma^2\{\pi_{\text{ext},i}(1-\pi_{\text{ext},i})\}^3]^{-\frac{1}{2}}\cdot e^{-\frac{1}{2\sigma^2}d(\pi_{\text{ext},i},\mathbb{E}(\pi_{\text{ext},i}|\boldsymbol X_i))}\cdot \label{eq:eq12}
\end{equation}
with the unit deviance function, $$d(\pi_{\text{ext},i},\mathbb{E}(\pi_{\text{ext},i}|\boldsymbol X_i))=\frac{(\pi_{\text{ext},i}-\mathbb{E}(\pi_{\text{ext},i}|\boldsymbol X_i))^2}{\pi_{\text{ext},i}(1-\pi_{\text{ext},i})\mathbb{E}(\pi_{\text{ext},i}|\boldsymbol X_i)^2(1-\mathbb{E}(\pi_{\text{ext},i}|\boldsymbol X_i))^2}\cdot$$ 
In \texttt{R}, the \texttt{simplexreg} package \citep{simplexreg} provides estimates of $\widehat{\boldsymbol\delta}$ by maximizing the likelihood in \eqref{eq:eq12}. 
$P(S_{\text{ext}}=1|\boldsymbol X)$ for individuals in the internal non-probability sample are then estimated by the plug-in  estimate
$(\text{exp}(\widehat{{\boldsymbol\delta}}'\boldsymbol X)/{1+\text{exp}({\widehat{\boldsymbol\delta}'\boldsymbol X}}))$.\\

\noindent
\textbf{Estimation of $p_{jk}(\boldsymbol X)$:} On the other hand, $p_{jk}(\boldsymbol X)$ is estimated based on the combined data (external union internal) sample. We define a nominal categorical variable with three levels corresponding to different values of $(j,k)$ pairs ($(j,k)=(1,1),(1,0),(0,1)$). An individual with level (1,1) is a member of both samples; (0,1) indicates a member of the exterior sample only, whereas (1,0) corresponds to the internal sample only. The multicategory response is again regressed on the internal selection model variables, $\boldsymbol X$ using a multinomial regression model and we obtain estimates of $p_{jk}(\boldsymbol X)$. \\

\noindent
Using the estimates of $p_{jk}(\boldsymbol X)$ from multinomial regression and $P(S_\text{ext}=1|\boldsymbol X)$ from simplex regression respectively, the selection probabilities for the internal sample,  $P(S=1|\boldsymbol X)$, were estimated from equation \eqref{eq:eq11} which serves as $\pi(\boldsymbol X)$ in equation \eqref{eq:eq6}. 
\subsection{Estimation of Weights Using Summary Level Statistics} \label{sec:summethods}
In this section, we discuss two methods to account for selection bias using summary level information that correspond to the target population. These summary information may be obtained directly from the target population (such as from census data) or  from summary data that has been made available by applying known survey design weights to an external sample drawn from the target population. We consider two types of summary level statistics namely, joint and marginal probabilities of $\boldsymbol X=(D,\boldsymbol Z_2,\boldsymbol W)$. Similar to \citet{beesley2022statistical}, we adopt post-stratification methods \citep{holt1979post} when joint probabilities of the selection variables $\boldsymbol X$ are available to us. On the other hand when only marginal probabilities are available, we modify the calibration method used in \citet{wu2003optimal} originally proposed for obtaining modified sampling probabilities from survey data.

\subsubsection{Poststratification (PS)}
We assume the joint distribution of the selection variables in the target population, namely $P(\boldsymbol X)$ are available to us. In case of continuous selection variables, we can at best expect to have access to joint probabilities of discretized versions of those variables.  Beyond this coarsening, obtaining joint probabilities of a large multivariate set of predictors become extremely challenging. In such cases, several conditional independence assumptions will be needed to specify a joint distribution from sub-conditionals.\\

\noindent
We consider the scenario where both $\boldsymbol Z_2$ and $\boldsymbol W$ are continuous variables. Let $\boldsymbol Z_2'$ and $W'$ be the discretized versions of $\boldsymbol Z_2$ and $\boldsymbol W$ respectively. We assume that the joint distributions for $\boldsymbol X'=(D,\boldsymbol Z_2',\boldsymbol W')$ in the target population are available from external sources. The post stratification method estimates the selection weights (inverse of selection probabilities into the internal sample) for the $i^{\text{th}}$ individual belonging to the internal sample by,
$$w_i\propto \frac{P(\boldsymbol X_i')}{P(\boldsymbol X_i'|S_i=1)}\cdot $$
The numerator of the above expression is the known population level joint distribution for the discretized selection variables $\boldsymbol X_i'$  obtained from external sources. On the other hand, the denominator, is the same probability empirically estimated from the internal sample. The inverses of the weights, $w_i$ are normalized  to obtain estimates of $\pi(\boldsymbol X_i)$ in equation \eqref{eq:eq6}. 

\subsubsection{Calibration (CL)}\label{sec:cal}
Calibration methods are often used in survey sampling to obtain corrected sampling weights in probability samples \citep{wu2003optimal}. We borrowed this idea to estimate internal selection probabilities $\pi(\boldsymbol X)$ by a model, $\pi(\boldsymbol X,\boldsymbol\alpha)$ indexed by parameters $\boldsymbol\alpha$, when marginal population means of the selection variables $\boldsymbol X$ are available from external sources. 
%method depends on the correct specification of the selection model, $\pi(\boldsymbol X,\boldsymbol \alpha)$.\\
Using target population size $N$ and the given marginal population means of the selection variables $\boldsymbol X$, we derive the population totals, namely $\sum_{i=1}^N \boldsymbol X_i$. We obtain the estimate of $\boldsymbol\alpha$ by solving the following calibration equation,
\begin{equation}
    \sum_{i=1}^N\frac{S_i\boldsymbol X_i}{\pi(\boldsymbol X_i,\boldsymbol\alpha)}=\sum_{i=1}^ N\boldsymbol X_i.  \label{eq:eq13}
\end{equation}
\noindent
In this approach, we match the sum of each selection variable in the internal sample (as estimated by inverse probability weighted sum on the LHS in equation \eqref{eq:eq13}) with the available total from the target population (RHS of equation \eqref{eq:eq13}), analogous to the method of estimation by first moment matching. Similar to Section \ref{sec:chen}, Newton-Raphson method is used to solve equation \eqref{eq:eq13} to estimate $\boldsymbol \alpha$ and henceforth obtain $\widehat{\pi}(\boldsymbol X_i,\boldsymbol \alpha)$ for each individual in the internal sample. We used a logistic specification of $\pi(\boldsymbol X,\boldsymbol\alpha)$ in our numerical work, but any selection model consistent for $\pi(\boldsymbol X)$ will lead to consistent estimates of $\boldsymbol \theta$ in equation \eqref{eq:eq6}.

\subsection{Asymptotic Distribution and Variance Estimation}\label{sec:avar}
We study the asymptotic distribution of the IPW estimator $\widehat{\boldsymbol \theta}$ under each of the four weighting methods.  We consider infinite population inference with population size $N$ going to infinity. We assume that all the variables, including $S$, $S_{\text{ext}}$, $\boldsymbol Z_1$ and $\boldsymbol X=(D,\boldsymbol Z_2,\boldsymbol W$) are random. This asymptotic setting is intrinsically different than finite population asymptotics, often followed in the survey literature where all the variables other than the selection indicators are considered to be non-random. This asymptotic analysis allows us to derive consistent estimators of the variance of $\widehat{\boldsymbol\theta}$ to be used in subsequent inference.\\

\noindent
\textbf{PL:} For PL, we derive the consistency, asymptotic normality and asymptotic variance estimator of $\widehat{\boldsymbol\theta}$ in Supplementary Section \ref{sec:proofPL}. The two-step variance estimation procedure incorporates uncertainty associated with estimates of the selection model parameter $\widehat{\boldsymbol \alpha}$ that are obtained by solving equation \eqref{eq:eq7}.\\
\noindent
\textbf{SR:} For SR, due to composite nature of the selection model, we use an approximation of the variance ignoring uncertainty in the estimates of the selection model parameters. The details of this approach are provided in Supplementary Section \ref{sec:ipw}.\\
\noindent
\textbf{PS:} For PS, the weights are known from summary statistics of the target population and the variance formula is provided in Supplementary Section \ref{sec:ipw}.\\
\noindent
\textbf{CL:} Similar to PL, for CL  we considered the uncertainty associated with estimates of the selection model parameter $\widehat{\boldsymbol \alpha}$ that are obtained by solving equation \eqref{eq:eq13} while deriving the estimated asymptotic variance of $\widehat{\boldsymbol\theta}$. Supplementary Section \ref{sec:proofcl} contains the details.\\

\noindent
We compared the average of the variance estimators proposed above across simulated datasets with the empirical Monte Carlo variances of the obtained parameter estimates. In particular, we quantify the potential inconsistency of our variance estimator for the SR method due to omission or ignorance of the uncertainty associated with estimation of the parameters of the selection model.

\section{Simulation Study}\label{sec:onesimu}
In this section, we present three simulation scenarios for each of the four DAGs introduced in Figure \ref{fig:Fig2}. The three setups differed in the assumption of the functional form of the selection model of the internal sample, namely $\pi({\boldsymbol X})$. For all three setups we consider the following generative distributions
\singlespacing
\begin{itemize}
    \item Disease model covariates \textbf{$\boldsymbol Z_1$ and $\boldsymbol Z_2$:} The joint distribution of $(Z_1,Z_2)$ is specified as,
     \begin{align*}
    \begin{pmatrix}
           Z_1 \\
           Z_2
         \end{pmatrix} \sim \mathcal{N}_2\left(\begin{pmatrix}
           0\\
           0
         \end{pmatrix},\begin{pmatrix}
           1 & 0.5\\
           0.5 & 1
         \end{pmatrix}\right)\cdot
  \end{align*}
    \item Disease outcome $D$: $D$ is simulated from the conditional distribution specified by,
    $$D|Z_1,Z_2\sim\text{Ber}\left(\frac{e^{\theta_0+\theta_1Z_1+\theta_2Z_2}}{1+e^{\theta_0+\theta_1Z_1+\theta_2Z_2}}\right)\cdot$$
    where, $\theta_0=-2$, $\theta_1=0.5$ and $\theta_2=0.5$.
    \item Selection model covariate $\boldsymbol W$\textbf{:} $W$ is an univariate random variable simulated from the conditional distribution of 
    $W|Z_1,Z_2,D$, specified by,
     $$W|Z_1,Z_2,D \sim  \mathcal{N}(\gamma_1\cdot D + \gamma_2\cdot Z_1 + \gamma_3\cdot Z_2,1)\cdot$$ 
    to incorporate the dependencies of $W$ on $D$, $Z_1$ and $Z_2$ respectively. We set $(\gamma_1,\gamma_2,\gamma_3)=(0,0,0),(0,1,0),(0,1,0),(1,1,1)$ for the four DAGs respectively.
   \item The internal sample selection models for the three setups are specified as follows.
   \begin{itemize}
       \item \textbf{Setup 1:} We set target population size to $N=50,000$. The functional form of the selection model is given by,
    \begin{equation}
    \text{logit}(P(S=1|Z_2,W,D))=\alpha_{0}+\alpha_{1}\cdot Z_2+\alpha_{2}\cdot W + \alpha_{3}\cdot D.
    \end{equation}
     We set $\alpha_0=-0.8$, $\alpha_1=0$ for DAGs 1 and 2 and $\alpha_1=0.7$ for DAGs 3 and 4, $\alpha_2=0.3$, $\alpha_3=1$.
     \item \textbf{Setup 2:}  The internal selection model in Setup 1 is perturbed by a constant multiplication given by,
$$P(S=1|Z_2,W,D)=0.4\cdot\left(\frac{e^{\alpha_{0}+\alpha_{1}\cdot Z_2+\alpha_{2}\cdot W + \alpha_{3}\cdot D}}{1+e^{\alpha_{0}+\alpha_{1}\cdot Z_2+\alpha_{2}\cdot W + \alpha_{3}\cdot D}}\right)\cdot$$
We set the exact same values for $\alpha_0,\alpha_1,\alpha_2,\alpha_3$ as in Setup 1. This pertubation of the selection model leads to a misspecification issue for pseudolikelihood and calibration methods, when we fit the two methods using a logistic form.
In order to ensure comparable sample size of internal data for both the simulation scenarios, we increased the target population size to 125,000 which is 2.5 times the previous population size, 50000. 
\item \textbf{Setup 3:} In this setup, we incorporate interaction terms of $(D,Z_2)$ and $(D,W)$ in the selection model. The new selection model is given by,
\begin{align*}
    \text{logit}(P(S=1|Z_2,W,D))=\alpha_{0}+\alpha_{1} Z_2+\alpha_{2}W + \alpha_{3} D + \alpha_4  D Z_2 + \alpha_5 D  W.
\end{align*}
    The  values of $\alpha_0,\alpha_1,\alpha_2,\alpha_3$ are identical to Setup 1. We set $(\alpha_4,\alpha_5)=(0,0.4)$ for DAGs 1 and 2 and $(0.5,0.4)$ for DAGs 3 and 4. Therefore this setup leads to a misspecification issue in pseudolikelihood, simplex regression and calibration methods when we fit these models without considering the interaction terms.
   \end{itemize}
    \item \textbf{External Selection Model:} For external data, the selection model can take any functional form and the selection probabilities are known to us. In our case, we assumed that the functional form of the external selection model is given by,
    \begin{align*}
    \text{logit}(P(S_{\text{ext}}=1|Z_2,W,D))=\nu_{0}+\nu_{1}\cdot Z_2+\nu_{2}\cdot W + \nu_{3}\cdot D. \label{eq:eqsimuext}
    \end{align*}
    The values of $(\nu_0,\nu_1,\nu_2,\nu_3)$ are given by $(-0.6,1.2,0.4,0.5)$. The probabilities $P(S_{\text{ext}}=1|Z_2,W,D)$ from the above equation  were multiplied by a factor of 0.75.
\end{itemize}

\noindent

\noindent
For the PS method, the joint distribution of $(D,Z_2',W^{'})$ are available from external sources, where both $Z_2'$ and $W^{'}$ are the coarsened versions of $Z_2$ and $W$. The criteria that we used to discretize these variables in the simulations is described in Supplementary Section \ref{sec:coarse}. All simulation results are summarized over 1000 replications.\\

\noindent
\textbf{Evaluation Metrics for Comparing Methods}\\

\noindent
In all the simulation setups, we compared the bias, relative bias and relative mean squared error (RMSE) relative to the unweighted method for both $\theta_1$ and $\theta_2$ across the four different weighted methods introduced in the previous section. The bias and relative bias \% in estimation for a parameter $\theta$ using $\widehat{\theta}$ are given by,
\begin{align*}
    & \text{Bias}(\theta,\widehat{\theta})=\frac{1}{R}\sum_{r=1}^R(\widehat{\theta_r}-\theta)\hspace{0.2cm}\text{and}\hspace{0.2cm}\text{Relative Bias \%}(\theta,\widehat{\theta})=\frac{|\text{Bias}(\theta,\widehat{\theta})|}{|\theta|}*100.
\end{align*}
where, $\widehat{\theta_r}$ is the estimate of $\theta$ in the $r^{\text{th}}$ simulated dataset,  $r=1,2,..R$ and $R=1000$.\\

\noindent
RMSE of $\widehat{\theta}$ with respect to the unweighted estimator $\widehat{\theta}_{\text{naive}}$ is defined as the ratio of the two MSEs given by,
$$\text{RMSE}(\widehat{\theta},\widehat{\theta}_{\text{naive}})=\frac{\frac{1}{R}\sum_{r=1}^R(\widehat{\theta_r}-\theta)^2}{\frac{1}{R}\sum_{r=1}^R(\widehat{\theta}_{\text{naive,r}}-\theta)^2}\cdot$$

\subsection{Results from the Simulation Study}\label{sec:resultsimu}
\subsubsection{Example DAG 1: unbiased case}
Under DAG 1, as described in Section \ref{sec:setup1method}, the selection bias-inducing term in the observed disease model $(D|Z_1,Z_2,S=1)$, namely $r(Z_1,Z_2)$ is a constant function in $Z_1,Z_2$.  We proved in Supplementary Section \ref{sec:pr2.1} that the unweighted method produces unbiased estimates of $\theta_1$ and $\theta_2$ for this DAG.  This theoretical result is evident from the simulation results in Table \ref{tab:Table 3} and Figure \ref{fig:resultstheta} under all the three setups. All five methods including the unweighted approach estimate both the disease model parameters with high accuracy. The highest relative bias  among all the three setups is 0.82\%, which implies that all the methods accurately estimate both the disease model parameters. Therefore, the results show that the different specifications of the functional form of the selection model do not affect the performances of any of the models significantly other than minor inflation in variance of the parameter estimates in some cases. The RMSE of all four weighted methods are close to 1.

\subsubsection{Example  DAG 2: $Z_1\rightarrow W$ arrow induced bias for coefficient of $Z_1$}
Under DAG 2, we showed in Section \ref{sec:setup2method} that $r(Z_1,Z_2)$ is a function of $Z_1$ only.  With introduction of the dependence, ($Z_1\rightarrow W$), the relative bias  in estimation of $\theta_1$ using the unweighted method increases to at least 9.6\% compared to 0.30\% in DAG 1, under all the three simulation setups. Under setup 1, the selection model for PL and CL are correctly specified.
In this setup, PL and CL perform best in terms of both bias and RMSE in estimating $\theta_1$, whereas SR estimates $\theta_1$ with a higher bias (1.72\%). Due to loss of information in discretizing selection variables, the relative bias and RMSE of PS is the highest among all the five methods (20.78\% and 2.07 respectively) under setup 1. However under simulation setups 2 and 3, the biases and RMSEs of both PL and CL increase significantly due to misspecification of selection models. For PS, the functional form of the selection model affect neither bias
nor RMSE. The estimate of $\theta_1$ using SR under setup 2 is close to the estimate in setup 1 since the estimation procedure of SR do not depend on the logistic form of the selection model. Therefore the effect of perturbation of the selection probabilities by a constant in setup 2 for SR is inconsequential. However, the introduction of interaction term in setup 3 increases the relative bias and RMSE for SR to 25.96\% and 6.33 respectively since it assumes no interaction in the estimation method. In Setup 3, the RMSE of all the four unweighted methods are remarkably high (atleast 6) due to severe misspecification.\\

\noindent
On the other hand, due to lack of dependence of $r(Z_1,Z_2)$ on $Z_2$, all the methods produce accurate estimate of $\theta_2$ in terms of both bias and RMSE under all the three simulation setups. 

\subsubsection{Example DAG 3: $Z_1\rightarrow W$ and $Z_2\rightarrow S$ induced bias for coefficients of $Z_1$ and $Z_2$ }
Under DAG 3, $r(Z_1,Z_2)$ is a function of both $Z_1$ and $Z_2$. Consequently under all the setups, the relative biases in estimation of $\theta_2$ increased to at least 16\% using the unweighted logistic method. Due to correct specification of the selection model for PL and CL in Setup 1, we observe that these two methods accurately estimate both $\theta_1$ and $\theta_2$. However, under both setups 2 and 3, the relative bias of estimates of $\theta_1$ and $\theta_2$ using PL and CL increases by a large amount. The relative bias in estimation of $\theta_2$ using SR increase to 16.74 in DAG 3 \% from 0.28\% in DAG 2 under setup 1 due to incorrect model specification. The bias in  estimation of $\theta_2$ using SR did not change much in Setups 2 and 3 from Setup 1. For $\theta_1$, we observe a big increase in bias in Setup 3 using SR. On the other hand in terms of RMSE, PL, CL and SR perform better than the unweighted logistic regression except for estimation of $\theta_1$ in Setup 3, where RMSE increased to atleast 15. All the four weighted  estimates being highly biased in compared to the naive estimator lead to this abrupt hike in RMSE. In all the other cases, the RMSE of these methods are below 1. The estimate of $\theta_1$ using PS in all the three setups performs poorly with high relative bias (at least 24\%) and RMSE (at least 3.61). On the other hand, both relative bias and RMSE in estimation of $\theta_2$ using PS is fairly low (at most 1.95\% and 0.09 respectively).
\subsubsection{Example DAG 4: strong dependence, increased bias for coefficients of $Z_1$ and $Z_2$}
Due to increase in dependence of $r(Z_1,Z_2)$ on $Z_2$, the bias in estimation of $\theta_2$ is the highest among all the DAGs for the unweighted method. The relative bias in estimation of $\theta_2$ increases to at least 40.34\% in all the three setups using the unweighted method. Similar to the previous DAGs, under Setup 1, PL and CL perform best in terms of both RMSE and bias among all the methods in estimation of the disease parameters due to correct specification of the selection model. Under Setups 2 and 3, these two methods perform poorly in terms of model misspecification. For SR, we observe an increase in relative bias to 29\% in estimation of $\theta_2$ compared to DAG 3 in setup 1. The bias in estimation of both $\theta_1$ and $\theta_2$ decrease compared to other DAGs using PS. For all the methods in most scenarios, the RMSE is less than 1, which implies better performance of the weighted methods compared to the unweighted logistic method.

\subsubsection{Summary Takeaways}
The comparative performances of the different methods under all varying simulation scenarios are  summarized in Figure \ref{fig:modelresults}.\\

\noindent
\textbf{Setup 1- Correctly Specified Individual Selection Model:}  As expected PL and CL estimate both the disease model parameters accurately when the selection model is correctly specified under all the four DAGs. They offer better solutions than using the naive logistic regression across all scenarios. It is not fair to compare PS and SR since they use different types of external data. Still, between PS and SR, there is no clear winner. While PS does well in DAG 4, SR has better performance in simpler DAGs. However SR is also always better than naive logistic regression in all simulations.  While for PS there could be very large RMSEs as we noticed in DAGs 2 and 3 (Table \ref{tab:Table 1}) due to high bias. The loss in information in discretizing the selection variables leads to incorrect selection weights estimation using PS. As a result, we observe that even with help of only marginal means of the selection variables from target population, under correct specification of selection model, PL works better than PS. However in DAG 4 due to high dependence among the different variables, the information contained in the discretized versions become adequate to estimate accurate weights for PS. \\

\noindent
\textbf{Variance Estimation/Uncertainty Quantification:} Supplementary Figures \ref{fig:Var} and \ref{fig:cp} assess the performances of the proposed variance estimators for the weighted methods under all the DAGs in Setup 1. Supplementary Figure \ref{fig:Var} shows the deviation of the estimated variance of $\widehat{\boldsymbol \theta}$ using the variance estimators discussed in Section \ref{sec:avar} from the Monte Carlo variance under all the four DAGs. We observe that the variance estimators for the four methods estimate accurately the Monte Carlo variance except for SR variance estimator in case of DAG 4. Supplementary Figure \ref{fig:cp} shows the coverage probabilities of the 95\% confidence intervals constructed using the proposed variance estimators. The coverage probabilities of PL and CL are close to 0.95 for all the four DAGs in Setup 1. The coverage probabilities are conservative for PS in DAGs 1,2 and 3. In DAG 4 the coverage probability is less than 0.5 for PS. The coverage probabilities of SR in DAGs 1 and 2 are comparable to the other methods. On the other hand in DAGs 3 and 4 the coverage probability of SR is close to 0. The main reason behind the low coverage probability is due to high bias of SR in DAGs 3 and 4 observed from subfigure (B) of Figure \ref{fig:resultstheta}. \\

\noindent
\textbf{Setup 2- Incorrectly Specified Selection Model 1:} In Setup 2, our results indicate that all methods performed remarkably well in DAG 1, similar to the previous setup since the bias term $r(Z_1,Z_2)$ is constant in $(Z_1,Z_2)$. SR and PS did not show major changes from the previous setup and the performance in terms of relative bias \% and RMSE  are better than PL and CL. In DAG 4, PS estimate  both the disease model parameters with low relative bias \% and RMSE. On the other hand, in DAGs 3 and 4, we observed highly inaccurate estimates for PL and CL in terms of relative bias (\%). 
Our findings suggest that these models are highly sensitive to selection model misspecification. \\

\noindent
\textbf{Setup 3- Incorrectly Specified Selection Model 2:} The key takeaways in this setup are similar to the previous one. However, the RMSE of the all four weighted methods are extremely high in DAGs 2 and 3 in estimation of $\theta_1$. Due to high degree of selection model misspecification with introduction of interaction among the selection variables, the selection weights estimates of the four unweighted methods are extremely inaccurate which leads to a huge increase in RMSE. However in DAG 4, the performance of the unweighted method degraded by a huge extent and as a result, the RMSEs of the weighted methods are much less compared to DAGs 2 and 3.

\section{Data Application: The Michigan Genomics Initiative}\label{sec:realex}

\subsection{Introduction}
The Michigan Genomics Initiative (MGI) is a rolling enrollment health EHR-linked biorepository within the University of Michigan Healthcare System consisting of over 93,000 participants primarily recruited through surgical encounters at Michigan Medicine. Due to the perioperative recruitment strategy, participants in MGI exhibit a lower overall health status and higher prevalence of cancer compared to the general population \citep{zawistowski2021michigan}. Time-stamped ICD (International Classification of Disease) diagnosis data are available for each patient. A rich ecosystem of additional information is available, including lifestyle and behavioral risk factors, lab and medication data, geo-coded residential information, socioeconomic metrics, and other patient-level, census tract-level, and provider-level characteristics.\\

\noindent
In this section we use the MGI data to study the association between cancer ($D$) and biological sex ($Z_1$) in the target US adult population. The direction of association in this case is well known from national SEER  (Surveillance, Epidemiology, and End Results) registry estimates. SEER data indicates lower lifetime cancer risk among women relative to men, with corresponding marginal log-odds ratios of -0.24 (2008-2010), -0.19 (2010-2012), -0.08 (2012-2014), and -0.07 (2014-2016) respectively (\url{seer.cancer.gov}). This known target national-level true association presents us with an opportunity to assess and compare the methods when applied to MGI in terms of bias in $\hat{\theta}_1$. In this analysis we investigate the marginal/unadjusted and age ($Z_2$) adjusted association between cancer and biological sex. For all the methods, we divided age into three categories, namely (18-39) (reference level), (40-59) and ($\geq 60$). 
For the selection model we use  diabetes, race, smoking currently, BMI (body mass index) and CHD (coronary heart disease) as $\boldsymbol W$. BMI has four categories, namely (0-18) (reference level), [18.5-25), [25-30) and ($\geq 30$).
For the individual level data methods (PL and SR), we use publicly available NHANES 2017-18 (National Health and Nutrition Examination Survey) data to construct IPW weights (\url{cdc.gov/nchs/nhanes}). NHANES is a complex multistage probability sampling design used to select participants representative of the civilian, non-institutionalized US population.
On the other hand, we use age specific and marginal summary statistics from SEER, the US Census and the US CDC (Centers for Disease Control and Prevention) to construct post-stratification and calibration weights respectively.

\subsection{Descriptive Summaries}
We select adult participants in NHANES since MGI consists of participants with age 18 years or older. After removing observations with incomplete data on the variables of interest, we are left with 80947 and 5153 participants in MGI and NHANES respectively. Table \ref{tab:TableMGI} presents a comprehensive summary of the variables of interest in both the MGI and NHANES datasets. The reported statistics for the NHANES dataset in this table are unweighted. As expected, MGI is enriched with cancer patients, with 48.7\% participants having a past or current cancer diagnosis ($D$).  The NHANES dataset demonstrates a prevalence of cancer at 10.3\%. The two studies differ in terms of the distribution of sex ($Z_1$), age ($Z_2$) and other selection covariates ($\boldsymbol W$).

%The percentages of females in MGI and NHANES are 46.2\% and 51.8\% respectively. The MGI cohort is older (mean: 57.5, standard deviation (sd): 18.1) than the NHANES cohort (mean: 51.2, sd: 18.1). MGI predominantly consists of a high percentage of Non Hispanic Whites (85.3\%), in contrast NHANES comprise of 34\% Non Hispanic Whites. Comorbidities like CHD (Coronary Heart Disease) and diabetes have higher occurrences in MGI (16.5\% and 33.3\%) compared to NHANES (4.6\% and 15.7\%). The mean BMI (Body Mass Index) in MGI and NHANES are 29.9 and 29.8 with sd of 7.26 and 7.4 respectively.\\

\subsection{Analyses of MGI Data} \label{sec:realana}
In this data example in the disease model we consider cancer, sex and age as $D$, $Z_1$ and $Z_2$ respectively. The sex variable is coded as 1 for female participants.
We are primarily interested in estimation of the marginal and age adjusted association parameters between cancer and sex, $\theta_1$, which is defined by the following equation.
\begin{equation*}
    \text{logit}(P(D=1|Z_1,\color{red}Z_2 \color{black}))=\theta_0+\theta_1\cdot  Z_1 +\color{red}\theta_2 \cdot \mathbb{I}(40\leq Z_2 \leq 59) +  \theta_3 \cdot \mathbb{I}(60\leq Z_2 \leq 150)\color{black}.
\end{equation*}
In the marginal association model, we did not adjust for age. The additional terms in the adjusted model are displayed in red. Note that the reference data from SEER corresponds to the marginal association model of cancer on sex without adjusting for age.\\

\noindent
For all the four weighting methods, we first estimated the IPW weights without including cancer ($D)$ as a selection variable (defined as $w_0$, inverse of $P(S=1|\boldsymbol{W},\color{red}Z_2 \color{black})$). This is due to the small number of cancer cases in NHANES compared to MGI as displayed in Table \ref{tab:TableMGI}. Then we modified the weights $w_0$ for the two individual level methods (PL and SR) using the following expression to incorporate cancer into the selection model,
\begin{align}
    w=\frac{1}{P(S=1|D,\boldsymbol W,\color{red}Z_2 \color{black})}&=\frac{P(D|\boldsymbol W,\color{red}Z_2 \color{black})}{P(D|S=1,\boldsymbol W,\color{red}Z_2 \color{black})}\cdot \frac{1}{P(S=1|\boldsymbol W,\color{red}Z_2 \color{black})}\label{eq:eqcwt1}\\
    &=  w_{0}\cdot \frac{P(D=1|\boldsymbol W,\color{red}Z_2 \color{black})^{D}\cdot(1-P(D=1|\boldsymbol W,\color{red}Z_2 \color{black}))^{(1-D)}}{P(D=1|\boldsymbol W, \color{red}Z_2 \color{black},S=1)^{D}\cdot (1-P(D=1|\boldsymbol W,\color{red}Z_2 \color{black},S=1))^{(1-D)}}\cdot\label{eq:eqcwt}
\end{align}
\noindent
where $P(D=1|\boldsymbol W,\color{red}Z_2 \color{black},S=1)$ is obtained from fitting a logistic regression model of $D$ on $\color{red}Z_2 \color{black},\boldsymbol W$ in MGI. On the other hand, we fit a weighted logistic regression in the NHANES data with the given sampling weights to obtain $P(D=1|\color{red}Z_2 \color{black},\boldsymbol W)$. The details of deriving equation \eqref{eq:eqcwt1} is provided in Supplementary Section \ref{sec:realcancerweights}. In case of the summary level methods PS and CL, estimation of $P(D=1|\boldsymbol W,\color{red}Z_2 \color{black})$ in equation \eqref{eq:eqcwt} is not possible due to limited availability of joint summary statistics from population. Therefore we approximate $P(D=1|\boldsymbol W,\color{red}Z_2 \color{black})$ in equation \eqref{eq:eqcwt} by $P(D=1|\color{red}Z_2 \color{black})$  using SEER estimate of age specific cancer SEER estimate. Similarly $P(D=1|\boldsymbol W,\color{red}Z_2 \color{black},S=1)$ in the denominator is obtained using a logistic regression $D$ on age in MGI. We still need the joint distribution of $\boldsymbol W,\color{red}Z_2$ to estimation of $w_0$ to implement PS. Due to limited availability of joint and conditional summary data on $\boldsymbol W,\color{red}Z_2$ from the US target population we made an assumption that given $\color{red}Z_2$ all the other selection variables are independent of each other. For all the weighted methods, we winsorized the selection weights by replacing the extreme 2.5\% and 97.5\% intervals by their respective quantiles to stabilise the methods. 

\subsection{Results}
We present the estimates of marginal and age adjusted association parameters between cancer and sex in Subfigures (A) and (B) of Figure \ref{fig:realdata} respectively using all the four weighted methods and unweighted logistic regression.\\

\noindent
\textbf{Marginal/Unadjusted Association:} 
We consider the SEER estimates of cancer-sex association to be the target truth (-0.24, -0.07). The estimate using the naive unweighted logistic regression method is -0.05 [95\% Confidence Interval (C.I) (-0.08,-0.03)]. The corresponding estimates obtained using the four IPW weighted methods namely PL, SR, PS, CL and without including cancer as a selection variable are 0.08 [95\% C.I (0.04,0.12)], 0.12 [95\% C.I (0.06,0.18)], 0.19 [95\% C.I (0.15,0.23)], 0.22 [95\% C.I (0.15,0.23)] respectively, showing that misspecified weights can sway the OR estimates in the wrong direction further away from the truth than the unweighted estimator. On the other hand, the estimates obtained using the four IPW weighted methods namely PL, SR, PS, CL and including cancer as a selection variable are -0.13 [95\% C.I (-0.16,-0.09)], -0.11 [95\% C.I (-0.17,-0.06)], -0.11 [95\% C.I (-0.15,-0.07)], -0.12 [95\% C.I (-0.15,-0.08)] respectively. The 95\% C.I of $\theta_1$ using all the four weighted methods largely lie within the SEER confidence estimate (-0.24, -0.07).\\

\noindent
\textbf{Age-adjusted Association:}  The age-adjusted estimate using the unweighted logistic method  is 0.10 [95\% C.I (0.07,0.13)] which lies in the opposite direction of the SEER confidence estimate. The estimates obtained using the four IPW weighted methods namely PL, SR, PS, CL and without including cancer as a selection variable skew the OR estimates in the opposite direction. In contrast, the estimates obtained using the four IPW weighted methods namely PL, SR, PS, CL  and including cancer as a selection variable are -0.07 [95\% C.I (-0.10,-0.03)], -0.09 [95\% C.I (-0.15,-0.02)], -0.07 [95\% C.I (-0.12,-0.02)], -0.05 [95\% C.I (-0.15,-0.08)] respectively. We observe that all the four weighted methods have reduced the bias of the estimated association parameter.

\subsection{Effects of different sub-sampling strategies within MGI}

In this section we carry out an idealized experiment using the MGI data. In real data, we do not know the actual variables $\boldsymbol W$ that are driving the selection mechanism. However, when we subsample data intentionally based on certain variables from MGI, the selection model and variables are known to us. This intentional and known subsampling strategy provide a framework to study the extent of selection bias introduced due to different choices of selection variablesand allow us to study the performance of different methods in recovering the truth in a more realistic situation. Let $S_{\text{sub}}$ denotes the selection indicator of being included into the subsample of MGI. We incorporate four subsampling strategies using a logistic selection model with varying parameter values. The first one is a random sample, the second depends on only cancer ($D$), third on cancer ($D$) and sex ($D$) and finally the fourth on cancer ($D$), sex ($Z$) and diabetes ($W$). In this exercise we do not include age in the disease model. The details of the subsampling strategies are given in Supplementary Section \ref{sec:subdet}. Using the above four subsamples of the MGI data, we evaluate the performances of the different methods in estimating the association parameter between cancer and biological sex. We consider two scenarios with two target population (MGI and US populations respectively) as we develop the weights. In both the scenarios, we assume that the true subsampling strategy is known.\\

\noindent
\textbf{First Scenario:} In the first scenario, we assume that the MGI cohort is the target population. Therefore in this case, the unweighted estimate obtained from MGI [-0.05, 95\% C.I (-0.08,-0.03)] is assumed to be the truth and we compare the estimates of the different methods under varying subsamples. The different subsamples serve as the non-probability samples of interest drawn from the target MGI population. For the individual level methods, in this scenario external data and target are same which is MGI and hence $\pi_{\text{ext}}=1$ for each participant. Therefore it does not make sense to apply SR since the response variable for Simplex Regression step is 1 for all datapoints. For PS and CL, we constructed joint probabilities and marginal means from the MGI data. The performances of three weighted and the unweighted logistic method are presented in subfigure (A) of Figure \ref{fig:subsamp}. Under random sampling, all the four methods accurately estimate $\theta_1$ in terms of bias as expected. In case of only cancer affecting subsampling, all the methods including the unweighted logistic are unbiased. This case is exactly same as DAG 1 which justifies the accurate performances of all the methods. However when sex ($Z$) and cancer ($D$) impacts selection, the estimate using the unweighted logistic method is severely biased. The association changes to an entirely wrong direction [0.20, 95\% C.I (0.15, 0.25)]. All three weighted methods, namely PL [-0.05, 95\% C.I (-0.10, 0)], PS [-0.05, 95\% C.I (-0.1, 0)] and CL [-0.05, 95\% C.I (-0.10, 0)] estimate the association parameter with negligible bias. We observe similar results in the fourth case where diabetes ($W$) affects selection along with cancer and sex. In all the cases, we observe that the variances of the methods increase in comparison to the true MGI C.I due to smaller sample size of the subsamples.\\

\noindent
\textbf{Second Scenario:} In this scenario, we assume that the US adult population is the target population, not MGI. Therefore in this case, the SEER estimates are assumed to be the truth and we compare the estimates of the different methods under varying subsampling schemes. For each of the three weighted methods, we apply a two stage weighting approach to obtain the final weights for the IPW regression. The first and second step of weights transport the subsample estimates to the MGI and then the US adult population respectively. In the second weighting step we use all the variables in $\boldsymbol W$ in Section \ref{sec:realana} including age. All the three weighted methods have reduced the bias in estimating the association parameter compared to the estimate of the unweighted method. We observe from subfigure (B) of Figure \ref{fig:subsamp} that under the first two subsampling strategies, all the three weighted methods perform well in terms of bias. For the last two subsampling strategies, CL and PL perform have a large overlap with the SEER band. For example when subsampling is based on both cancer and sex, majority portion of the 95\% C.I bands of PL [-0.14, 95\% C.I (-0.21,-0.08)] and CL [-0.14, 95\% C.I (-0.22,-0.07)] are with the SEER band. Compared to PL and CL, PS on the other hand did not perform well since a large portion of PS is outside the SEER band. Again in all the cases, we observe that the variances of the methods increase due to smaller sample size of the subsamples.\\

\noindent
Similar to the simulation results obtained in Section \ref{sec:resultsimu}, we observe when either $Z$ or $W$ or both affect selection along with $D$, the unweighted estimate is highly biased. The IPW methods help in reducing the bias of the parameter of interest.

\section{Discussion and Conclusion} \label{sec:discussion}
Selection bias is a major concern in EHR studies since it is extremely difficult to ascertain the process through which a patient from the target population enters the analytic sample or why a particular observation or lab result appears in the health record of a patient. The mechanism of patients' interactions with the healthcare system may be influenced by a variety of patient characteristics such as age, sex, race, healthcare access and other health related co-morbidities. If the issue of selection bias is overlooked, association analyses are generally biased because unadjusted inference from these non-probability samples from EHR data is generally not transportable to the target population. Therefore, there is a pressing need to understand the structure of selection bias and correct for it when needed, in order to draw valid inferences for the target population.\\

\noindent
Hospital-based biobanks are enriched with specific diseases. For example, the dataset we used, MGI, \citep{zawistowski2021michigan} recruits patients while they are waiting for surgery. Consequently it is enriched for many diseases including skin cancer \citep{fritsche2019exploring}. Thus the results from MGI are not directly generalizable to the Michigan or US population which is evident from the results shown in Section \ref{sec:realex} on the cancer sex association.  On the other hand, population based biobanks such as the UK Biobank, Estonian Genome Center Biobank, and Taiwan Biobank attempt to recruit participants nationally by inviting volunteers. Even these large population-based biobanks like the UK Biobank suffer from healthy control bias \citep{fry2017comparison,van2022reweighting}.  Nationally representative studies such as the NIH All of Us often have a purposeful sampling strategy that leads to, say, oversampling certain underrepresented subgroups \citep{all2019all}. The problem of selection bias may be maginfied when multiple biobanks all over the world are being harmonized together for massive meta-analysis. For example, the Global Biobank Meta-analysis Initiative (GBMI) \citep{vogan2022global} has linked 24 biobanks with more than 2.2 million genotyped samples linked with health records. For turning such big data into meaningful knowledge, one needs to characterize the different sampling mechanisms underlying the recruitment strategies of these diverse biobanks. We hope this paper provides a conceptual and analytic framework towards understanding selection bias and a set of the tools that are available to us. \\

\noindent
In this work, we introduce a framework to assess selection bias using DAGs in case of estimating association of a binary response variable with other independent variables of interest. We develop four IPW methods and present using a simulation study the extent by which these methods are able to reduce selection bias across a diverse set of simulation settings. Finally we discuss a data example of estimating the association of biological sex with cancer in a hospital based biobank namely, MGI and compare the results obtained from different methods to the population based SEER estimate.\\

\noindent
This work has several limitations. All the methods we consider suffer when the selection probability model is misspecified. We only considered functional misspecification of the selection model in our simulation studies but there will likely be many omitted covariates. It is nearly impossible to  measure all the variables driving selection. Gathering more data on a representative sub-sample of the population embedded within EHR may also lead to more substantial reduction of bias. Chart review \citep{yin2022cost}, multi-wave sampling \citep{liu2022sat}, double sampling approaches \citep{chen2000unified} should also be considered as possible avenues.  We also ignored selection model uncertainty in the simplex regression method. Bootstrap can offer a potential solution to consistent variance estimation. Finally, as described in Table \ref{tab:Table 1}, selection bias occurs not in isolation but in conjunction with several other sources of bias, for example with outcome misclassification \citep{beesley2020analytic}. We need sensitivity analysis tools and source of bias diagnostics for EHR data to identify a hierarchy of the different sources of bias for a given problem. In this analysis we did not consider the time stamps of the observations in longitudinal EHR data. The relationships between covariates and outcomes in the DAGs are highly dependent on the relative ordering. Extension of the discussed methods to longitudinal data may address this issue.\\

\noindent
Finally, creation of nationally integrated databases, where all health encounters for everyone are recorded in the same data system will enable researchers to harness the full potential of real-world healthcare data for everyone, not just for some selected (often historically privileged) sub-populations. Use of exclusionary cohorts and data disparity is at the heart of fairness in modern machine learning methods \citep{mhasawade2021machine,parikh2019addressing}. In that sense, equal probability sample selection method (EPSEM) is a tool to ensure equity and fairness in data science. In absence of EPSEM in real world data, thinking about selection bias is at the heart of doing inclusive science with data. Our hope is that our paper will contribute to that important discourse.

\section{Acknowledgements}
This research is supported by NSF DMS 1712933, NIH/NCI  CA267907 and NIH R01GM139926. The authors would like to thank Professor Ruth Keogh and organizers and attendees of the symposium on 50 years of the Cox Model for including this work in the program and providing feedback during the presentation.

\section{Competing interests}
Nothing to declare. 

\section{Materials and Correspondence}
All correspondence should be directed to Bhramar Mukherjee (\href{mailto:bhramar@umich.edu}{bhramar@umich.edu}). All codes are available in \url{https://github.com/Ritoban1/Short-Note-Selection-Bias.git}. Michigan Genomics Initiative Data are available after institutional review board approval to select researchers. See \url{https://precisionhealth.umich.edu/our-research/michigangenomics/}.

\bibliographystyle{abbrvnat}
\bibliography{references.bib}
\newpage
\begin{figure}[H]
    \centering
    \includegraphics[width =\linewidth]{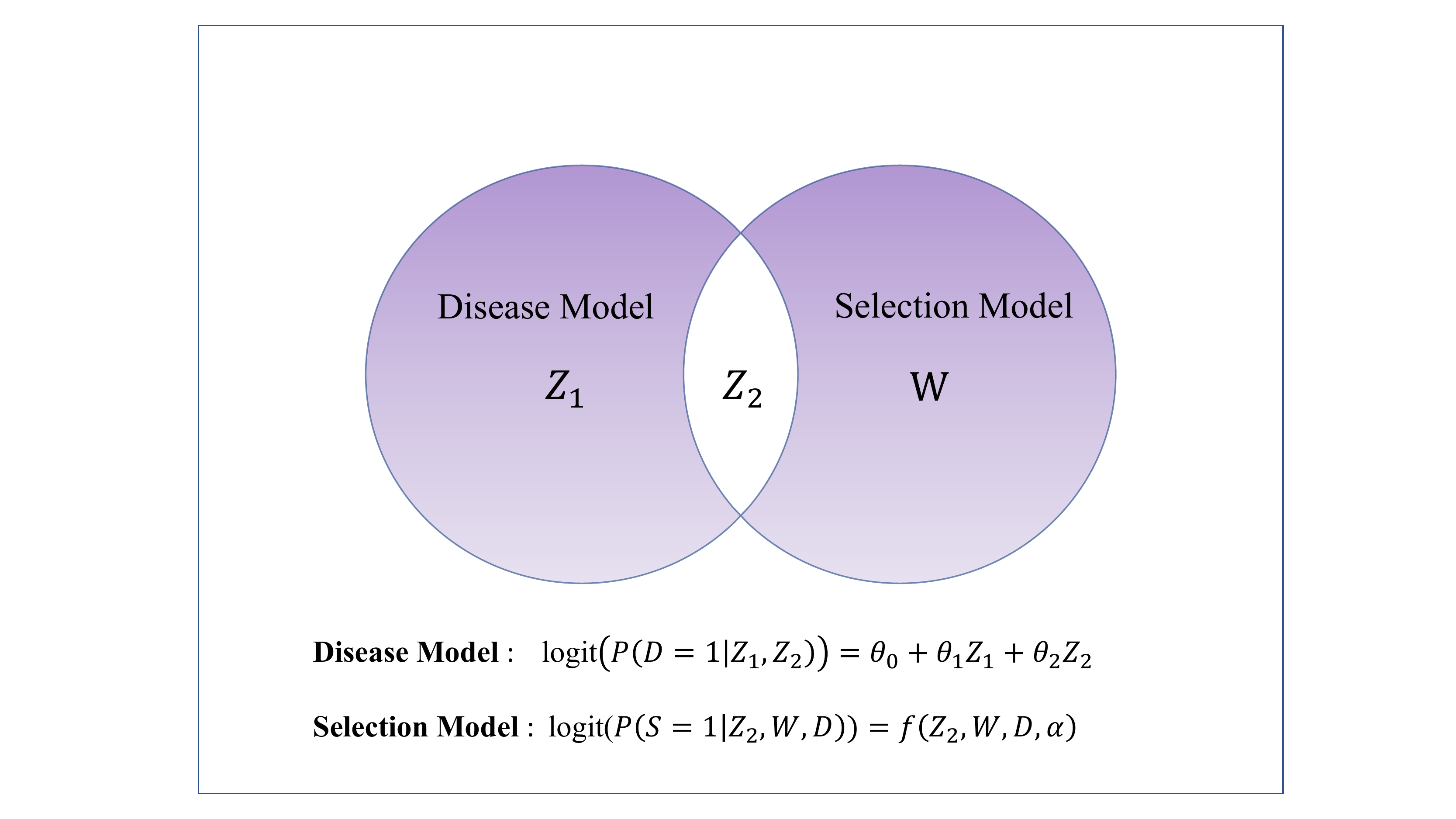}
    \caption{Figure depicting the disease and selection models along with the different variables present in both the models.}
    \label{fig:models}
\end{figure}

\begin{figure}[H]

\centering
\subfloat[DAG 1]{
  \centering
  \includegraphics[width=0.51\textwidth]{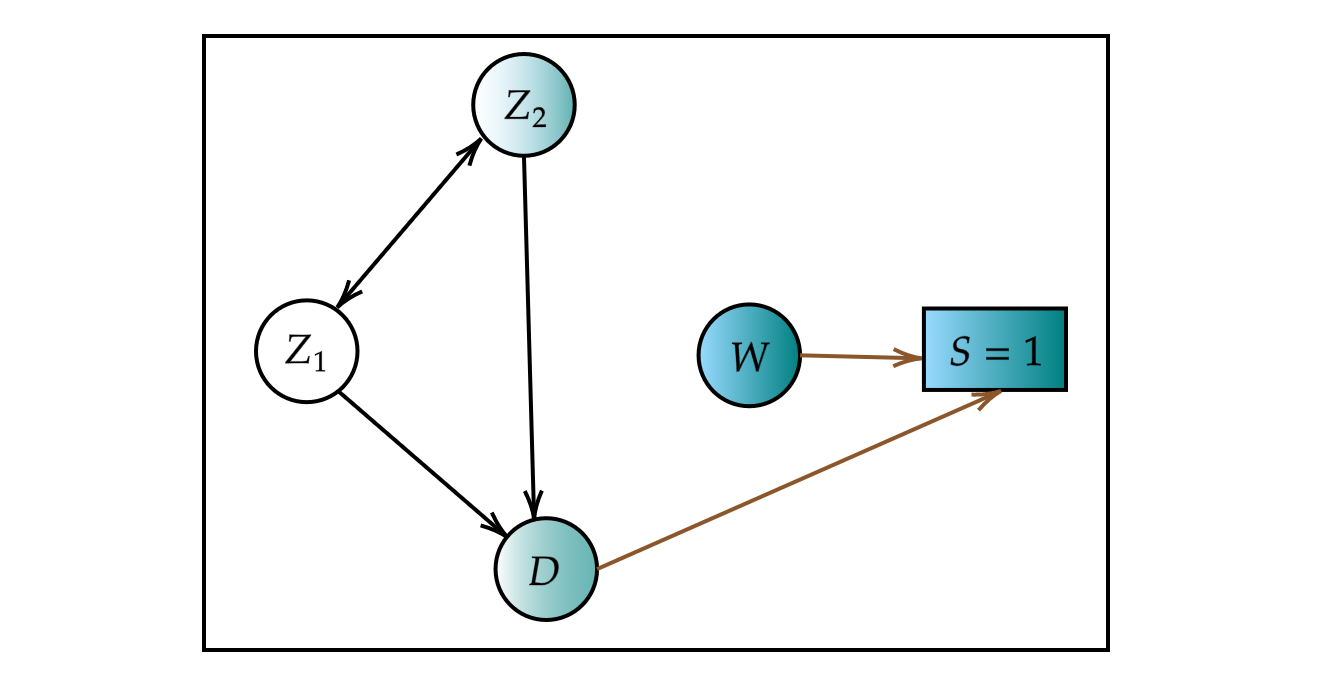}
}
\subfloat[DAG 2]{
  \centering
  \includegraphics[width=0.51\textwidth]{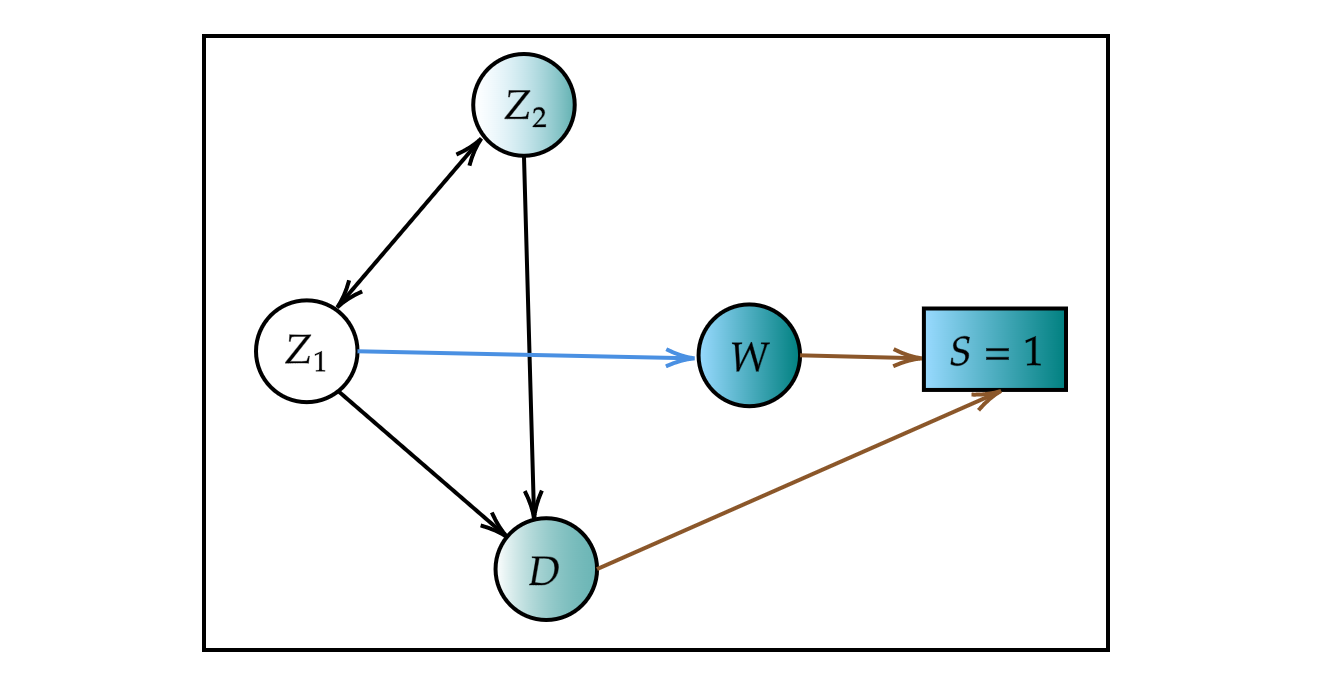}
}
\newline
\subfloat[DAG 3]{
\centering
  \includegraphics[width=0.51\textwidth]{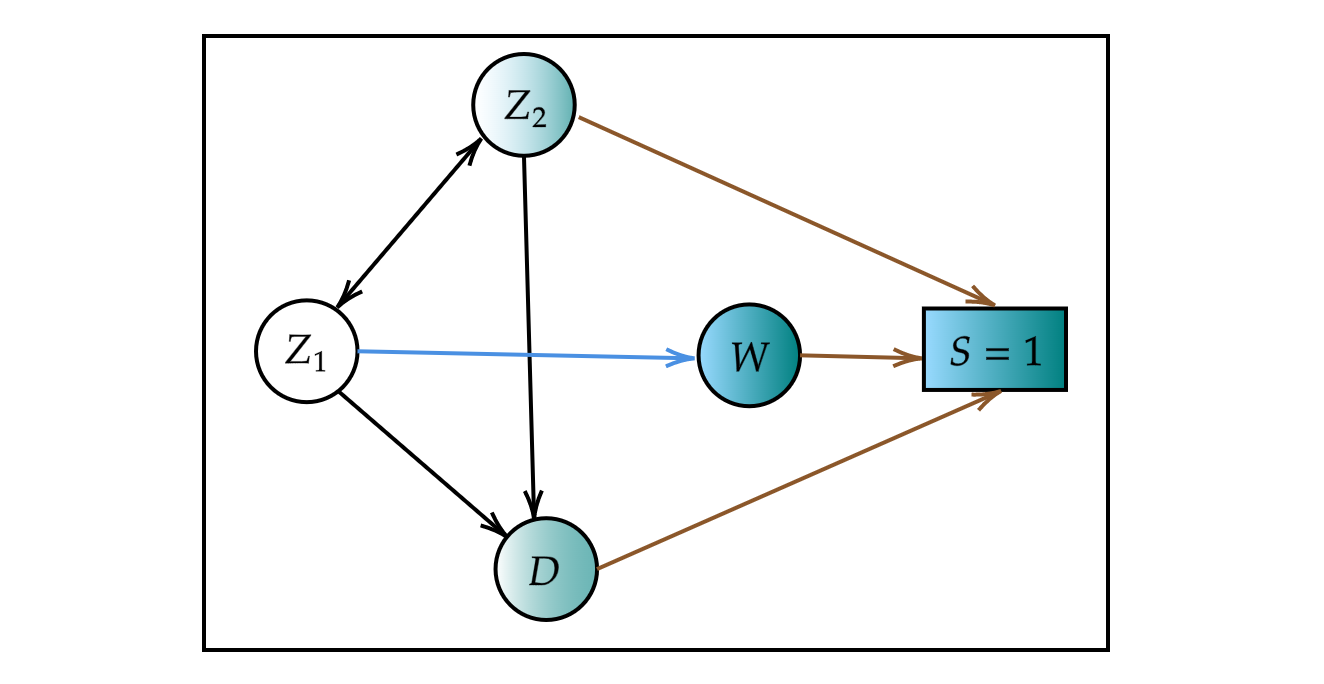}
}
\subfloat[DAG 4]{
\centering
  \includegraphics[width=0.51\textwidth]{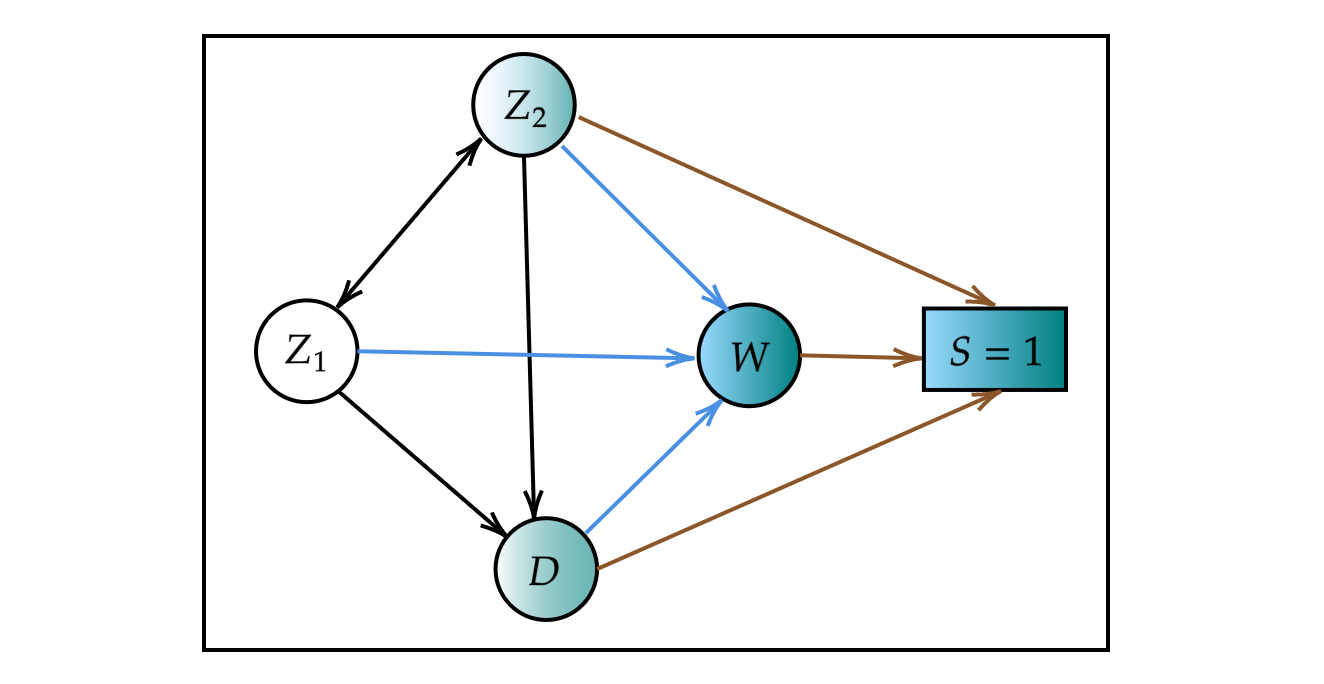}
}\\
\subfloat{
\centering
  \includegraphics[width=0.9\textwidth]{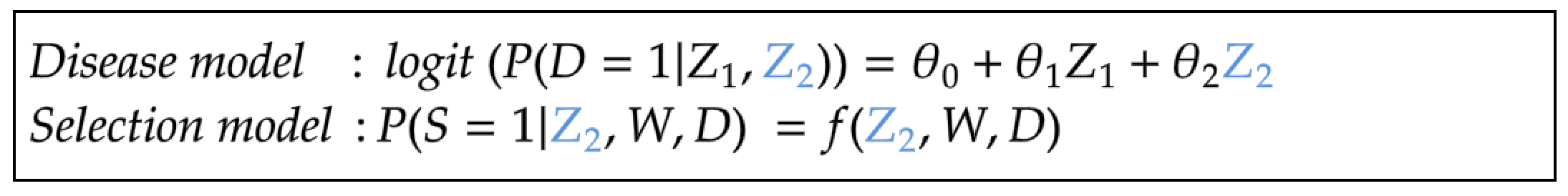}
}
\caption{Selection Directed Acyclic Graphs (DAGs)  representing some plausible relationships between different variables of interest: $D$ (Disease Indicator), $S$ (Selection Indicator into the internal sample) $\boldsymbol Z_1$ (Predictors in the disease model only (White), $\boldsymbol Z_2$ (Predictors both in disease and selection models (Mixture of Blue and White)) and  $\boldsymbol W$ (Predictors in the selection model only (Blue)).}\label{fig:Fig2}
\end{figure}
\begin{figure}[H]
    \centering
    \includegraphics[width =\linewidth]{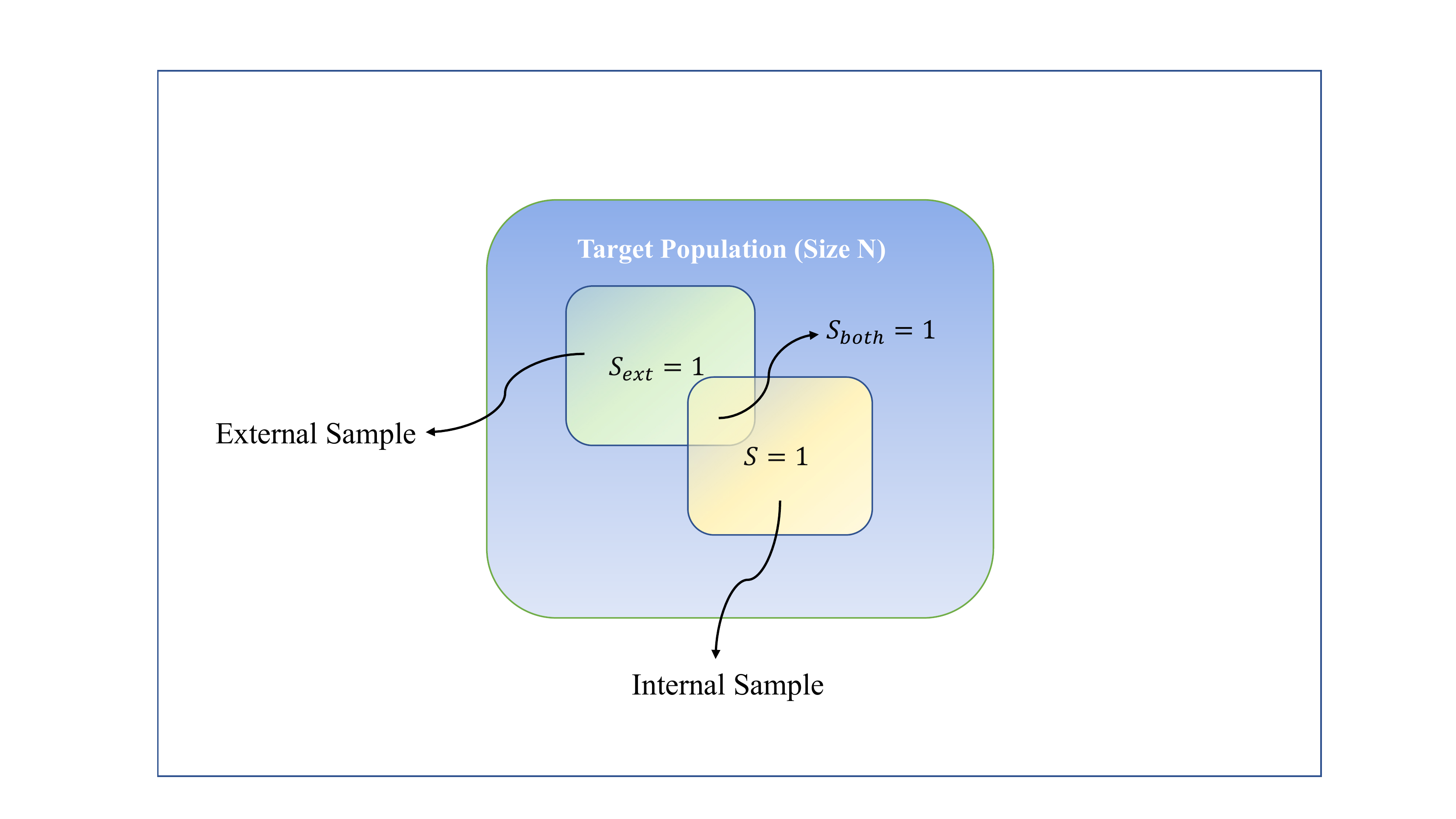}
    \caption{Figure depicting the relationship between the target population, internal non-probability and external probability samples. $S$ and $S_{\text{ext}}$ are the selection indicator variables of internal and external samples respectively. $S_{\text{both}}$ is the selection indicator variable for a person present in both internal and external samples.}
    \label{fig:population}
\end{figure}
\begin{figure}[H]
    \centering
    \includegraphics[width =\linewidth]{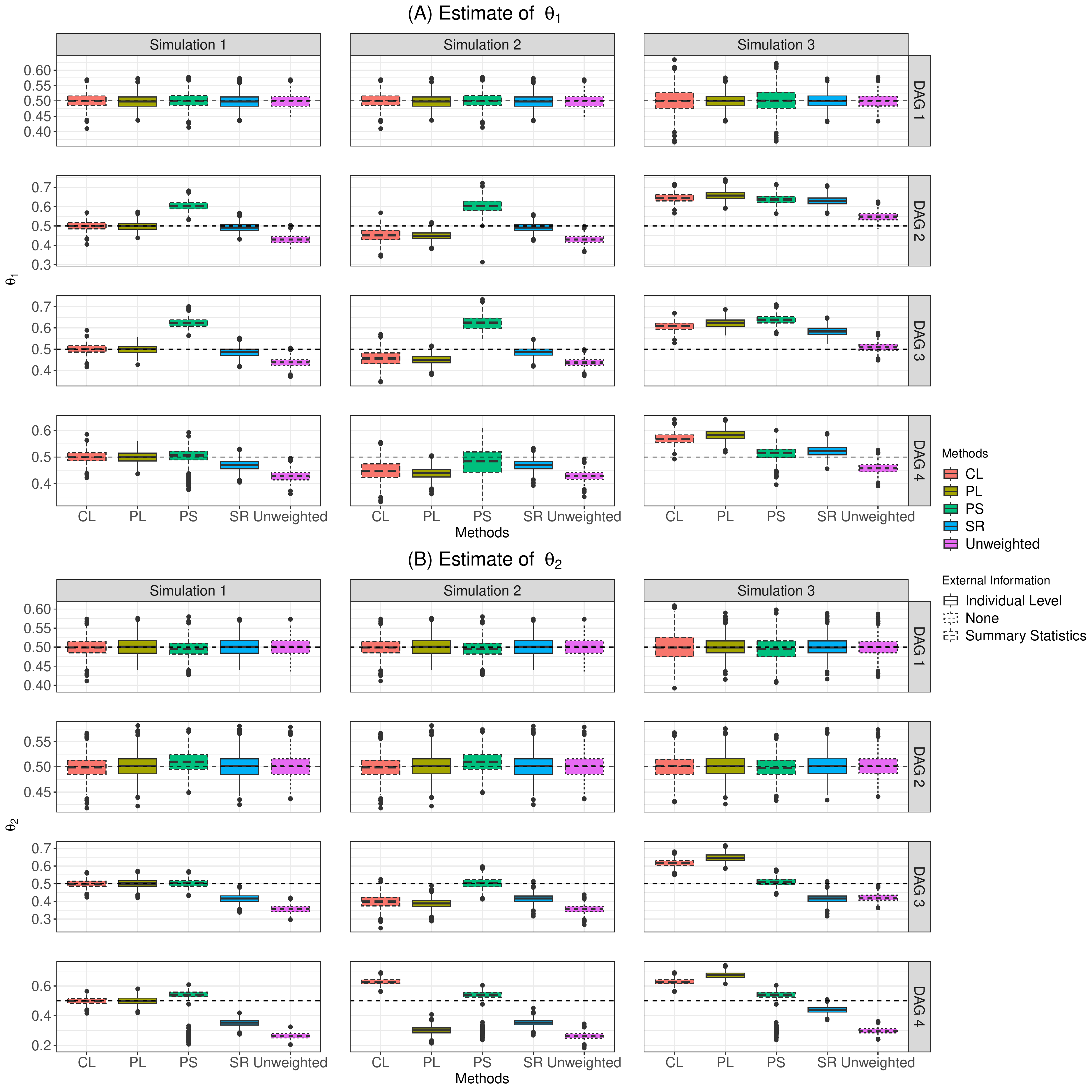}
    \caption{(A) Estimates of $\theta_1$ in, coefficient of $Z_1$ in the disease model along with 95\% C.I using the unweighted and the four weighted methods under the three simulation setups for each of the four DAGs. 
    (B) Estimates of $\theta_2$ in, coefficient of $Z_2$ in the disease model along with 95\% C.I using the unweighted and the four weighted methods under the three simulation setups for each of the four DAGs. 
    Unweighted : unweighted logistic regression, SR : simplex regression, PL : pseudolikelihood, PS : poststratification and CL : calibration.}
    \label{fig:resultstheta}
\end{figure}

\begin{figure}[H]
    \centering
    \includegraphics[width =\linewidth]{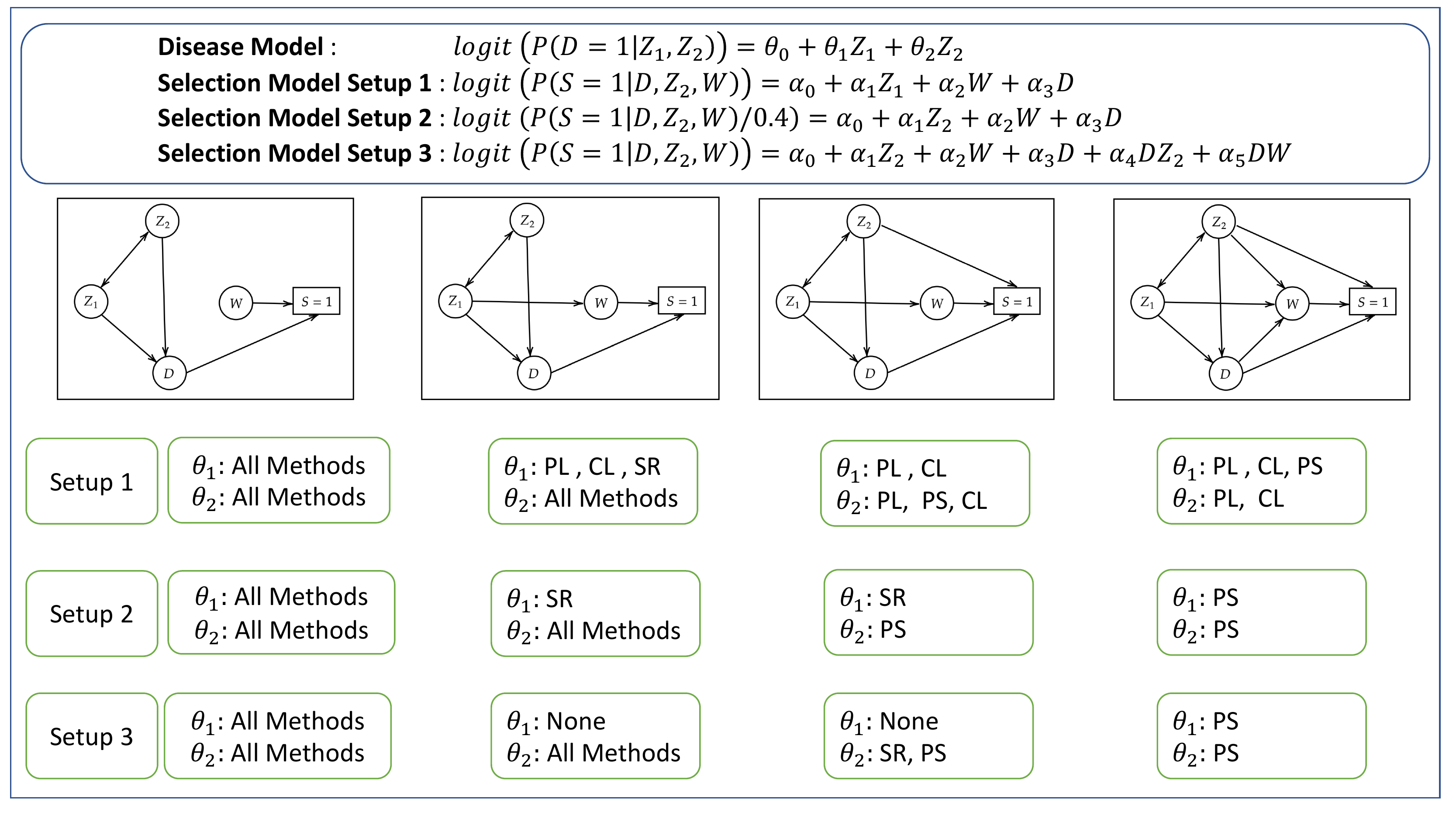}
    \caption{Preferred methods of estimation for including the unadjusted and the four weighted ones in terms of bias of estimation of the disease model parameters under different DAG setups in all the three considered simulation setups.
    Unweighted : unweighted logistic regression, SR : simplex regression, PL : pseudolikelihood, PS : poststratification and CL : calibration.}
    \label{fig:modelresults}
\end{figure}

\begin{figure}[H]
    \centering
    \includegraphics[width =1.2\linewidth,angle =90]{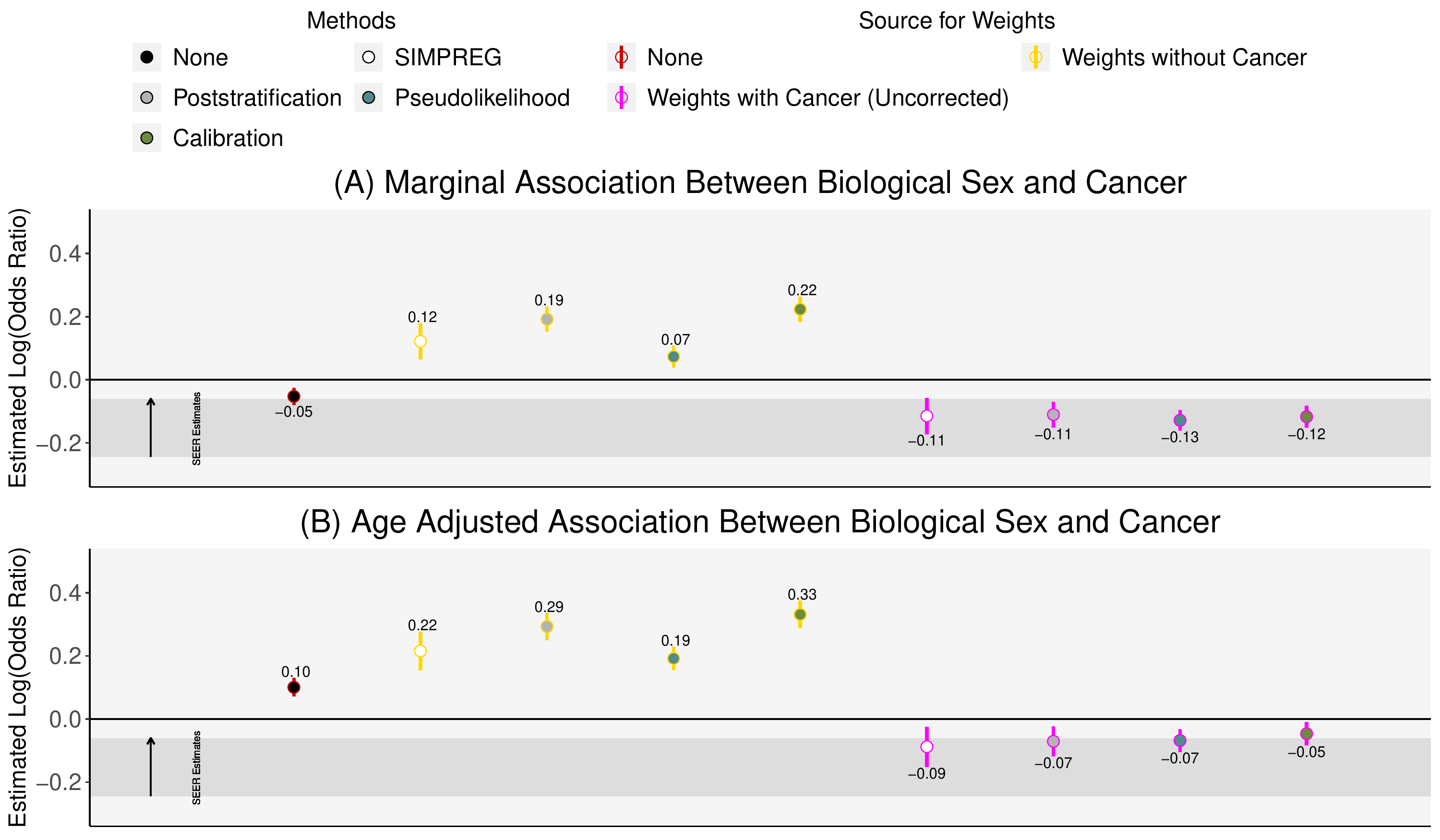}
    \caption{(A) Estimates of the marginal association between cancer and sex along with 95\% C.I in  using all the four weighted methods and the unweighted logistic regression with and without including cancer as a selection variable.
    (B) Estimates of the age adjusted association between cancer and sex along with 95\% C.I in  using all the four weighted methods and the unweighted logistic regression with and without including cancer as a selection variable.
    The grey band represents the true SEER estimates. The red bar corresponds to the unweighted logistic method. The yellow and pink bars correspond to the estimates of the four IPW methods without and with including cancer in the selection model respectively.}
    \label{fig:realdata}
\end{figure}

\begin{figure}[H]
    \centering
    \includegraphics[width =\linewidth]{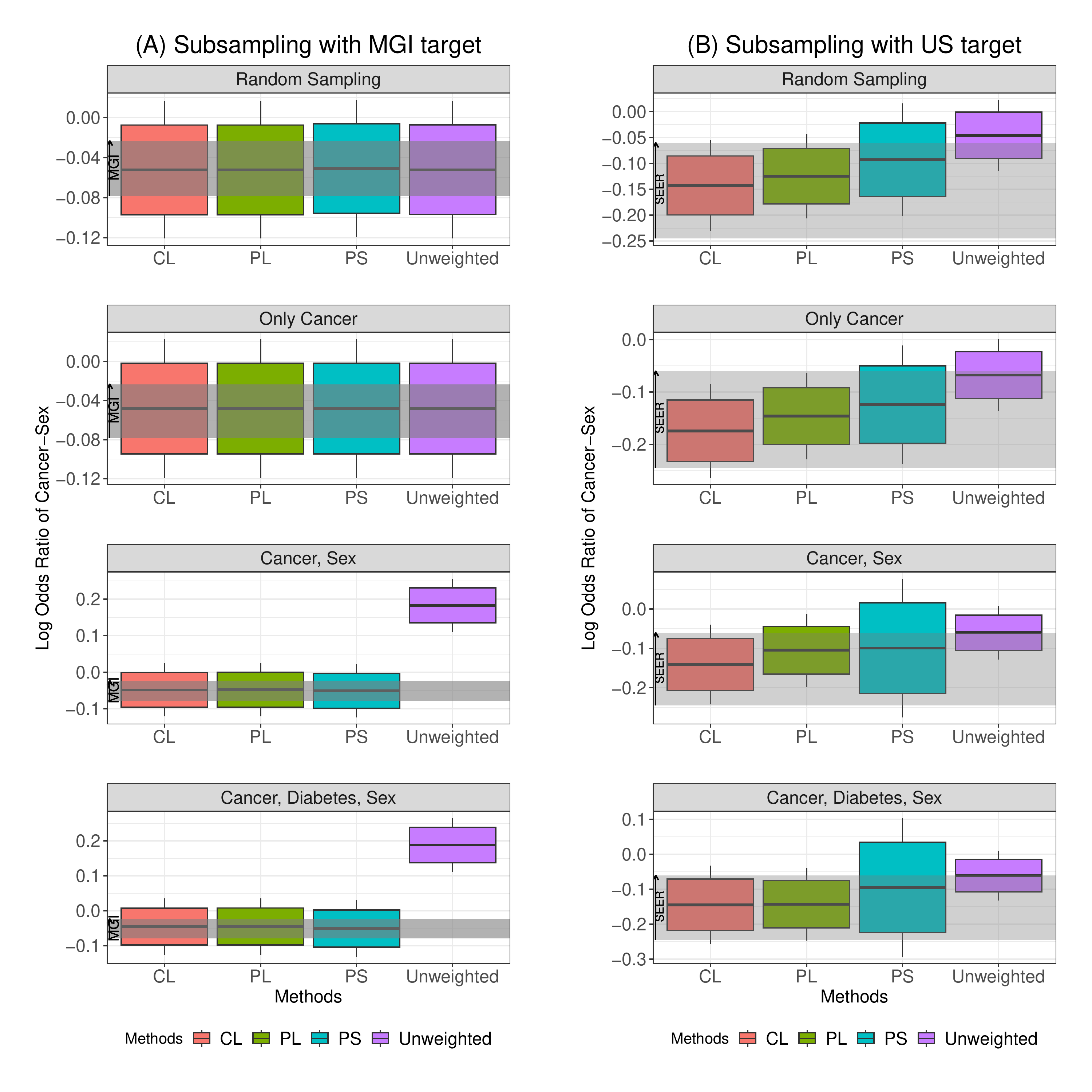}
    \caption{(A) Estimates of the association between cancer and sex along with 95\% C.I using three weighted methods and the unweighted logistic regression under the four subsampling strategies when MGI is assumed to be the target population. (B) Estimates of the association between cancer and sex along with 95\% C.I using three weighted methods and the unweighted logistic regression under the four subsampling strategies when US is assumed to be the target population.
    The grey band represents the 95\% C.I of estimate of $\theta_1$ obtained from MGI using unweighted logistic regression. Unweighted : unweighted logistic regression, PL : pseudolikelihood, PS : poststratification and CL : calibration.}
    \label{fig:subsamp}
\end{figure}

\newpage
\doublespacing
\begin{table}[H]
\centering
\setlength\tabcolsep{10pt}
\rotatebox{90}{
\begin{adjustbox}{max width=1.3\textwidth,keepaspectratio}
\Huge
    \begin{tabular}{|c|c|c|c|}
\hline
\textbf{Type of Bias}                                                                            & \textbf{Definition}                                                                                                                                                                                                                                                                                                                                                                                                                                                                                                                                                                                                                                                                                                                                                                    & \textbf{Literature}        & \textbf{Software}                                                                                                                                                                                                                                                                                                                                                                                                                                                                                                                                                                                                                                                                                                                                                                                                                                                                                                  \\ \hline
\textbf{Imperfect Phenotyping}                                                                   & \begin{tabular}[c]{@{}c@{}} Major bias in EHR studies \\ Misclassification of derived disease phenotypes. \\
Overreporting and underreporting can both occur.\\ Underreporting is the primary source.
\end{tabular}                                                                              & \begin{tabular}[c]{@{}c@{}} \citet{neuhaus1999bias}\citet{beesley2020emerging} \\\citet{tong2020augmented,chen2019inflation} \\  \citet{yin2022cost,liu2022sat}\\ \citet{huang2018pie} \end{tabular}   & \begin{tabular}[c]{@{}c@{}} 
\texttt{SAMBA}\\\end{tabular}                                     \\ \hline
\textbf{Missing Data}                                                                            & \begin{tabular}[c]{@{}c@{}}Lack of routine checkup.\\ Loss of follow up reasons.\end{tabular}                                                                                                                                                                                                                                                                                                                                                     & \begin{tabular}[c]{@{}c@{}}Hot Deck \citet{madow1983incomplete}, Tree based methods \citet{doove2014recursive}\\ Expectation-Maximization Algorithm \citet{dempster1977maximum}\\ IPW and AIPW techniques \citet{seaman2018introduction}\\ Full information maximum likelihood \citet{marcoulides2013advanced}\\ Multiple imputation \citet{rubin2004multiple}, Pattern mixture models \citet{little1993pattern}\\ Heckman imputation  \citet{galimard2016multiple} \end{tabular} & \begin{tabular}[c]{@{}c@{}} 
\texttt{MICE} \\ \texttt{MissForest}\end{tabular} \\ \hline
\textbf{Confounding}                                                                             & \begin{tabular}[c]{@{}c@{}}Direct cause of both exposure and the response.\\ Major challenge caused by unmeasured confounders.\end{tabular}                                                               & \begin{tabular}[c]{@{}c@{}} \citet{toh2011confounding} \citet{sun2022use} \\ Negative and Double Negative Controls\\
\citet{shi2020multiply,lipsitch2010negative}
\end{tabular}& \begin{tabular}[c]{@{}c@{}} 
\end{tabular} \\ \hline
\textbf{\begin{tabular}[c]{@{}c@{}}Lack of data \\ harmonization across \\ cohorts\end{tabular}} & \begin{tabular}[c]{@{}c@{}}Integrating disparate data of various sources and formats. \\
Different clinics recruit patients using varying selection  criteria. \\
Misclassification of phenotypes differs across cohorts.\end{tabular} & \begin{tabular}[c]{@{}c@{}}\citet{almeida2021methodology,abbasizanjani2023harmonising}\\
\citet{fu2020assessment,glynn2019heterogeneity,zawistowski2021michigan}.\end{tabular} &  \begin{tabular}[c]{@{}c@{}}\texttt{meta}\end{tabular}                                                                                                                                                                                             \\ \hline
\textbf{\begin{tabular}[c]{@{}c@{}}Heterogeneity of \\ studied populations\end{tabular}}         & \begin{tabular}[c]{@{}c@{}} Systematic differences between the \\ population characteristics or sampling mechanisms \end{tabular}                                                                                                                                                                                                                                                                                      & \begin{tabular}[c]{@{}c@{}}\citet{beesley2020emerging}\end{tabular}      &    \begin{tabular}[c]{@{}c@{}}\texttt{meta}\end{tabular}                                                                                                                                                                                                                                                                                                                                                                                                                                                                                                                                                                                                                                                                                                                                                  \\ \hline
\end{tabular}
\end{adjustbox}

}

\caption{Different types of biases in EHR studies other than selection bias along with their description and relevant literature to reduce the corresponding bias.}
\label{tab:Table 1}

\end{table}
\doublespacing
\begin{table}[H]
\huge
\setlength\tabcolsep{10pt}
\centering
\rotatebox{90}{
\resizebox{1.4\textwidth}{!}{
\begin{tabular}{|c|c|c|c|c|c|}
\hline
\textbf{Class}                                                                                       & \textbf{Method}                                                     & \textbf{Reference}                                                                                                                 & \textbf{Description}                                                                                                                                                                                                                                                                                                                                        & \textbf{Features}   & \textbf{Software}                                                                                                                                                                                                                                  \\ \hline
\textbf{Unweighted}                                                                                  & Naive                                                               &                                                                                                                                    & \begin{tabular}[c]{@{}c@{}}Unweighted logistic regression with D and $(\boldsymbol Z_1,\boldsymbol Z_2)$ \\ as response and predictors respectively.\end{tabular}                                                                                                                                                                                                                & \begin{tabular}[c]{@{}c@{}}Estimates $\boldsymbol \theta_1$ and $\boldsymbol\theta_2$ unbiasedly\\ only when $r(\boldsymbol Z_1, \boldsymbol Z_2)$ is free of  \\ $\boldsymbol Z_1$ and $\boldsymbol Z_2$ respectively.\end{tabular}  & \begin{tabular}[c]{@{}c@{}}\texttt{glm}\end{tabular} \\ \hline
\multirow{6}{*}{\textbf{\begin{tabular}[c]{@{}c@{}}Individual External\\ Patient Data\end{tabular}}} & \begin{tabular}[c]{@{}c@{}}Simplex\\ Regression (SR)\end{tabular}   & \begin{tabular}[c]{@{}c@{}}\citet{barndorff1991some},\\ \citet{elliot2009combining},\\ \citet{beesley2022statistical}\end{tabular} & \begin{tabular}[c]{@{}c@{}}Weighted logistic regression with D and $(\boldsymbol Z_1,\boldsymbol Z_2)$ \\ as response and predictors respectively. \\
Integrating external level individual data, the weights \\obtained using Simplex \\ and Multinomial Regressions.\end{tabular}                                                                     & \begin{tabular}[c]{@{}c@{}} Bias in estimation of $\boldsymbol \theta_2$ increases \\ with the complexity of DAG. 
Difference in predictors\\ of the internal and external selection models leads to biased\\ estimates for SR. \end{tabular}   & \begin{tabular}[c]{@{}c@{}}\texttt{simplexreg}\\\texttt{SAMBA}\end{tabular}                                  \\ \cline{2-6} 
                                                                                                     & \begin{tabular}[c]{@{}c@{}}Pseudo \\ Likelihood (PL)\end{tabular}   & \citet{chen2020doubly}                                                                                            & \begin{tabular}[c]{@{}c@{}}An weighted logistic regression with D and $(\boldsymbol Z_1,\boldsymbol Z_2)$ \\ as response and predictors respectively. With external level\\ individual data, the maximum likelihood estimating equations \\ for weights are approximated by the measures from external \\ level data resulting in pseudolikelihood estimating methods.\end{tabular} & \begin{tabular}[c]{@{}c@{}}With correct specification of selection model, \\ the estimates are highly accurate for all setups,\\ However highly sensitive to model misspecification.\end{tabular}                                          & \begin{tabular}[c]{@{}c@{}}\texttt{nleqslv}\end{tabular}                             \\ \hline
\multirow{6}{*}{\textbf{\begin{tabular}[c]{@{}c@{}}Summary Level \\ Statistics Data\end{tabular}}}   & \begin{tabular}[c]{@{}c@{}}Post \\ Stratification (PS)\end{tabular} & \citet{beesley2022statistical}                                                                                    & \begin{tabular}[c]{@{}c@{}}An weighted logistic regression with D and $(\boldsymbol Z_1,\boldsymbol Z_2)$ \\ as response and predictors respectively. With external level \\ joint probabilities of ($D,\boldsymbol Z_2',\boldsymbol W'$), where $\boldsymbol Z_2',\boldsymbol W'$ \\ are the discrete versions of $\boldsymbol Z_2,\boldsymbol W$, the weights are obtained.\end{tabular}                                                      & \begin{tabular}[c]{@{}c@{}}Fails to work for Setup 2 and 3. However\\ the efficiency in terms of bias increases for Setup 4 and is \\ least sensitive to model mis specification.\end{tabular}             & \begin{tabular}[c]{@{}c@{}}\texttt{survey}\\\texttt{SAMBA}\end{tabular}                                                                  \\ \cline{2-6} 
                                                                                                     & \begin{tabular}[c]{@{}c@{}}Calibration\\ (CL)\end{tabular}          & \citet{wu2003optimal}                                                                                             & \begin{tabular}[c]{@{}c@{}} Weighted logistic regression with D and $(\boldsymbol Z_1,\boldsymbol Z_2)$ \\ as response and predictors respectively. With external level\\ marginal probabilities on the selection variables $(D,\boldsymbol Z_2,\boldsymbol W)$, \\ the weights are estimated using pseudolikelihood estimating methods\\ similar to the PL method\end{tabular}    & \begin{tabular}[c]{@{}c@{}}With correct specification of selection model, \\ the estimates are highly accurate for all setups,\\ However highly sensitive to model misspecification.\end{tabular}                                       & \begin{tabular}[c]{@{}c@{}}\texttt{survey}\end{tabular}                              \\ \hline
\end{tabular}}
}
\caption{Short Summary of the five methods including the naive and the four weighted ones.}\label{tab:Table2}
\end{table}

\begin{table}[H]
\centering
\begin{adjustbox}{angle=90, max width=0.68\textwidth}
\begin{tabular}{|c|c|ccc|ccc|ccc|ccc|}
\hline
                                                       &                                   & \multicolumn{3}{c|}{\textbf{\begin{tabular}[c]{@{}c@{}}Bias($\theta_1$)\\ (Multiplied by 1000)\end{tabular}}}                                                                                                     & \multicolumn{3}{c|}{\textbf{\begin{tabular}[c]{@{}c@{}}Bias($\theta_2$)\\ (Multiplied by 1000)\end{tabular}}}                                                                                  & \multicolumn{3}{c|}{\textbf{RMSE($\theta_1$)}}                                                                                                                & \multicolumn{3}{c|}{\textbf{RMSE($\theta_2$)}}                                                                                                                \\ \cline{3-14} 
\multirow{-2}{*}{\textbf{DAG}}                         & \multirow{-2}{*}{\textbf{Method}} & \multicolumn{1}{c|}{\textbf{Setup 1}}                               & \multicolumn{1}{c|}{\textbf{Setup 2}}                                 & \textbf{Setup 3}                                                    & \multicolumn{1}{c|}{\textbf{Setup 1}}                                & \multicolumn{1}{c|}{\textbf{Setup 2}}                                 & \textbf{Setup 3}                                & \multicolumn{1}{c|}{\textbf{Setup 1}}                     & \multicolumn{1}{l|}{\textbf{Setup 2}}                     & \multicolumn{1}{l|}{\textbf{Setup 3}} & \multicolumn{1}{l|}{\textbf{Setup 1}}                     & \multicolumn{1}{l|}{\textbf{Setup 2}}                     & \multicolumn{1}{l|}{\textbf{Setup 3}} \\ \hline
                                                       & \textbf{Unweighted}               & \multicolumn{1}{c|}{{\color[HTML]{FE0000} \textbf{-1.48 (0.30\%)}}} & \multicolumn{1}{c|}{{\color[HTML]{FE0000} \textbf{-1.48 (0.30\%)}}}   & {\color[HTML]{FE0000} \textbf{-1.48 (0.30\%)}}                      & \multicolumn{1}{c|}{{\color[HTML]{FE0000} \textbf{1.00 (0.20\%)}}}   & \multicolumn{1}{c|}{{\color[HTML]{FE0000} \textbf{1.00 (0.20\%)}}}    & {\color[HTML]{FE0000} \textbf{1.00 (0.20\%)}}   & \multicolumn{1}{c|}{{\color[HTML]{FE0000} \textbf{1}}}    & \multicolumn{1}{c|}{{\color[HTML]{FE0000} \textbf{1}}}    & {\color[HTML]{FE0000} \textbf{1}}     & \multicolumn{1}{c|}{{\color[HTML]{FE0000} \textbf{1}}}    & \multicolumn{1}{c|}{{\color[HTML]{FE0000} \textbf{1}}}    & {\color[HTML]{FE0000} \textbf{1}}     \\ \cline{2-14} 
                                                       & \textbf{PL}                       & \multicolumn{1}{c|}{{\color[HTML]{FE0000} \textbf{-1.56 (0.32\%)}}} & \multicolumn{1}{c|}{{\color[HTML]{FE0000} \textbf{-1.56 (0.32\%)}}}   & {\color[HTML]{FE0000} \textbf{-1.56 (0.32\%)}}                      & \multicolumn{1}{c|}{{\color[HTML]{FE0000} \textbf{1.21 (0.24\%)}}}   & \multicolumn{1}{c|}{{\color[HTML]{FE0000} \textbf{1.21 (0.24\%)}}}    & {\color[HTML]{FE0000} \textbf{1.21 (0.24\%)}}   & \multicolumn{1}{c|}{{\color[HTML]{FE0000} \textbf{1.03}}} & \multicolumn{1}{c|}{{\color[HTML]{FE0000} \textbf{1.03}}} & {\color[HTML]{FE0000} \textbf{1.03}}  & \multicolumn{1}{c|}{{\color[HTML]{FE0000} \textbf{1.07}}} & \multicolumn{1}{c|}{{\color[HTML]{FE0000} \textbf{1.07}}} & {\color[HTML]{FE0000} \textbf{1.07}}  \\ \cline{2-14} 
                                                       & \textbf{SR}                       & \multicolumn{1}{c|}{{\color[HTML]{FE0000} \textbf{-1.58 (0.32\%)}}} & \multicolumn{1}{c|}{{\color[HTML]{FE0000} \textbf{-1.58 (0.32\%)}}}   & {\color[HTML]{FE0000} \textbf{-1.58 (0.32\%)}}                      & \multicolumn{1}{c|}{{\color[HTML]{FE0000} \textbf{1.22 (0.24\%)}}}   & \multicolumn{1}{c|}{{\color[HTML]{FE0000} \textbf{1.22 (0.24\%)}}}    & {\color[HTML]{FE0000} \textbf{1.22 (0.24\%)}}   & \multicolumn{1}{c|}{{\color[HTML]{FE0000} \textbf{1.06}}} & \multicolumn{1}{c|}{{\color[HTML]{FE0000} \textbf{1.06}}} & {\color[HTML]{FE0000} \textbf{1.06}}  & \multicolumn{1}{c|}{{\color[HTML]{FE0000} \textbf{1.07}}} & \multicolumn{1}{c|}{{\color[HTML]{FE0000} \textbf{1.07}}} & {\color[HTML]{FE0000} \textbf{1.07}}  \\ \cline{2-14} 
                                                       & \textbf{PS}                       & \multicolumn{1}{c|}{{\color[HTML]{FE0000} \textbf{0.83 (0.17\%)}}}  & \multicolumn{1}{c|}{{\color[HTML]{FE0000} \textbf{0.83 (0.17\%)}}}    & {\color[HTML]{FE0000} \textbf{0.83 (0.17\%)}}                       & \multicolumn{1}{c|}{{\color[HTML]{FE0000} \textbf{-4.10 (0.82\%)}}}  & \multicolumn{1}{c|}{{\color[HTML]{FE0000} \textbf{-4.10 (0.82\%)}}}   & {\color[HTML]{FE0000} \textbf{-4.10 (0.82\%)}}  & \multicolumn{1}{c|}{{\color[HTML]{FE0000} \textbf{1.21}}} & \multicolumn{1}{c|}{{\color[HTML]{FE0000} \textbf{1.21}}} & {\color[HTML]{FE0000} \textbf{1.21}}  & \multicolumn{1}{c|}{{\color[HTML]{FE0000} \textbf{0.88}}} & \multicolumn{1}{c|}{{\color[HTML]{FE0000} \textbf{0.88}}} & {\color[HTML]{FE0000} \textbf{0.88}}  \\ \cline{2-14} 
\multirow{-5}{*}{\textbf{DAG  1}}                      & \textbf{CL}                       & \multicolumn{1}{c|}{{\color[HTML]{FE0000} \textbf{0.68 (0.14\%)}}}  & \multicolumn{1}{c|}{{\color[HTML]{FE0000} \textbf{0.68 (0.14\%)}}}    & {\color[HTML]{FE0000} \textbf{0.68 (0.14\%)}}                       & \multicolumn{1}{c|}{{\color[HTML]{FE0000} \textbf{-0.61 (0.12\%)}}}  & \multicolumn{1}{c|}{{\color[HTML]{FE0000} \textbf{-0.61 (0.12\%)}}}   & {\color[HTML]{FE0000} \textbf{-0.61 (0.12\%)}}  & \multicolumn{1}{c|}{{\color[HTML]{FE0000} \textbf{1.13}}} & \multicolumn{1}{c|}{{\color[HTML]{FE0000} \textbf{1.13}}} & {\color[HTML]{FE0000} \textbf{1.13}}  & \multicolumn{1}{c|}{{\color[HTML]{FE0000} \textbf{1}}}    & \multicolumn{1}{c|}{{\color[HTML]{FE0000} \textbf{1}}}    & {\color[HTML]{FE0000} \textbf{1}}     \\ \hline
                                                       & \textbf{Unweighted}               & \multicolumn{1}{c|}{-70.80 (14.16\%)}                               & \multicolumn{1}{c|}{-70.42 (14.08\%)}                                 & \multicolumn{1}{l|}{47.61 (9.53\%)}                                 & \multicolumn{1}{c|}{{\color[HTML]{FE0000} \textbf{1.02 (0.20\%)}}}   & \multicolumn{1}{c|}{{\color[HTML]{FE0000} \textbf{0.66 (0.13\%)}}}    & 1.69 (0.34\%)                                   & \multicolumn{1}{c|}{1}                                    & \multicolumn{1}{c|}{1}                                    & 1                                     & \multicolumn{1}{c|}{{\color[HTML]{FE0000} \textbf{1}}}    & \multicolumn{1}{c|}{{\color[HTML]{FE0000} \textbf{1}}}    & {\color[HTML]{FE0000} \textbf{1}}     \\ \cline{2-14} 
                                                       & \textbf{PL}                       & \multicolumn{1}{c|}{{\color[HTML]{FE0000} \textbf{-1.39 (0.28\%)}}} & \multicolumn{1}{c|}{{\color[HTML]{FE0000} \textbf{-51.02 (10.20\%)}}} & 157.33 (31.47\%)                                                    & \multicolumn{1}{c|}{{\color[HTML]{FE0000} \textbf{1.29 (0.25\%)}}}   & \multicolumn{1}{c|}{{\color[HTML]{FE0000} \textbf{0.61 (0.12\%)}}}    & 2.11 (0.42\%)                                   & \multicolumn{1}{c|}{{\color[HTML]{FE0000} \textbf{0.09}}} & \multicolumn{1}{c|}{{\color[HTML]{000000} 0.57}}          & 9.22                                  & \multicolumn{1}{c|}{{\color[HTML]{FE0000} \textbf{1.06}}} & \multicolumn{1}{c|}{{\color[HTML]{FE0000} \textbf{1.05}}} & {\color[HTML]{FE0000} \textbf{1.07}}  \\ \cline{2-14} 
                                                       & \textbf{SR}                       & \multicolumn{1}{c|}{{\color[HTML]{FE0000} \textbf{-8.61 (1.72\%)}}} & \multicolumn{1}{c|}{-8.25 (1.65\%)}                                   & 129.80 (25.96\%)                                                    & \multicolumn{1}{c|}{{\color[HTML]{FE0000} \textbf{1.39 (0.28\%)}}}   & \multicolumn{1}{c|}{{\color[HTML]{FE0000} \textbf{0.68 (0.14\%)}}}    & 2.34 (0.47\%)                                   & \multicolumn{1}{c|}{{\color[HTML]{FE0000} \textbf{0.10}}} & \multicolumn{1}{c|}{{\color[HTML]{FE0000} \textbf{0.12}}} & 6.33                                  & \multicolumn{1}{c|}{{\color[HTML]{FE0000} \textbf{1.05}}} & \multicolumn{1}{c|}{{\color[HTML]{FE0000} \textbf{1.06}}} & {\color[HTML]{FE0000} \textbf{1.05}}  \\ \cline{2-14} 
                                                       & \textbf{PS}                       & \multicolumn{1}{c|}{103.90 (20.78\%)}                               & \multicolumn{1}{c|}{103.85 (20.77\%)}                                 & 137.93 (27.59\%)                                                    & \multicolumn{1}{c|}{{\color[HTML]{FE0000} \textbf{10.05 (2.01\%)}}}  & \multicolumn{1}{c|}{{\color[HTML]{FE0000} \textbf{10.29 (2.06\%)}}}   & -0.82 (0.16\%)                                  & \multicolumn{1}{c|}{2.07}                                 & \multicolumn{1}{c|}{2.22}                                 & 7.15                                  & \multicolumn{1}{c|}{{\color[HTML]{FE0000} \textbf{1.14}}} & \multicolumn{1}{c|}{{\color[HTML]{000000} 1.86}}          & {\color[HTML]{FE0000} \textbf{0.94}}  \\ \cline{2-14} 
\multirow{-5}{*}{\textbf{DAG 2}}                       & \textbf{CL}                       & \multicolumn{1}{c|}{{\color[HTML]{FE0000} \textbf{0.44 (0.09\%)}}}  & \multicolumn{1}{c|}{{\color[HTML]{FE0000} \textbf{-47.79 (9.56\%)}}}  & 145.26 (29.05\%)                                                    & \multicolumn{1}{c|}{{\color[HTML]{FE0000} \textbf{-0.67 (0.13\%)}}}  & \multicolumn{1}{c|}{{\color[HTML]{FE0000} \textbf{-0.13 (0.03\%)}}}   & 0.40 (0.08\%)                                   & \multicolumn{1}{c|}{{\color[HTML]{FE0000} \textbf{0.09}}} & \multicolumn{1}{c|}{{\color[HTML]{000000} 0.65}}          & 7.89                                  & \multicolumn{1}{c|}{{\color[HTML]{FE0000} \textbf{1.01}}} & \multicolumn{1}{c|}{{\color[HTML]{FE0000} \textbf{1.13}}} & {\color[HTML]{FE0000} \textbf{1.00}}  \\ \hline
                                                       & \textbf{Unweighted}               & \multicolumn{1}{c|}{-62.83 (12.6\%)}                                & \multicolumn{1}{c|}{-62.28 (12.46\%)}                                 & 9.24 (1.85\%)                                                       & \multicolumn{1}{c|}{-143.75 (28.74\%)}                               & \multicolumn{1}{c|}{-143.86 (28.77\%)}                                & 79.05 (15.81\%)                                 & \multicolumn{1}{c|}{1}                                    & \multicolumn{1}{c|}{1}                                    & 1                                     & \multicolumn{1}{c|}{1}                                    & \multicolumn{1}{c|}{1}                                    & 1                                     \\ \cline{2-14} 
                                                       & \textbf{PL}                       & \multicolumn{1}{c|}{{\color[HTML]{FE0000} \textbf{-0.30 (0.06\%)}}} & \multicolumn{1}{c|}{{\color[HTML]{000000} -48.90 (9.78\%)}}           & 122.56 (24.51\%)                                                    & \multicolumn{1}{c|}{{\color[HTML]{FE0000} \textbf{1.27 (0.25\%)}}}   & \multicolumn{1}{c|}{{\color[HTML]{000000} -112.28 (22.46\%)}}         & 147.67 (29.53\%)                                & \multicolumn{1}{c|}{{\color[HTML]{FE0000} \textbf{0.11}}} & \multicolumn{1}{c|}{{\color[HTML]{000000} 0.67}}          & 31.01                                 & \multicolumn{1}{c|}{{\color[HTML]{FE0000} \textbf{0.03}}} & \multicolumn{1}{c|}{{\color[HTML]{FE0000} \textbf{0.63}}} & 3.33                                  \\ \cline{2-14} 
                                                       & \textbf{SR}                       & \multicolumn{1}{c|}{-13.80 (2.76\%)}                                & \multicolumn{1}{c|}{{\color[HTML]{FE0000} \textbf{-13.41 (2.68\%)}}}  & 84.07 (16.81\%)                                                     & \multicolumn{1}{c|}{-83.72 (16.74\%)}                                & \multicolumn{1}{c|}{-84.22 (16.84\%)}                                 & 77.48 (15.50\%)                                 & \multicolumn{1}{c|}{{\color[HTML]{FE0000} \textbf{0.15}}} & \multicolumn{1}{c|}{{\color[HTML]{FE0000} \textbf{0.15}}} & 15.08                                 & \multicolumn{1}{c|}{0.36}                                 & \multicolumn{1}{c|}{0.36}                                 & 0.98                                  \\ \cline{2-14} 
                                                       & \textbf{PS}                       & \multicolumn{1}{c|}{123.44 (24.7\%)}                                & \multicolumn{1}{c|}{124.32 (24.86\%)}                                 & 138.89 (27.78\%)                                                    & \multicolumn{1}{c|}{{\color[HTML]{FE0000} \textbf{2.15 (0.43\%)}}}   & \multicolumn{1}{c|}{{\color[HTML]{FE0000} \textbf{2.10 (0.42\%)}}}    & {\color[HTML]{FE0000} \textbf{9.74 (1.95\%)}}   & \multicolumn{1}{c|}{3.61}                                 & \multicolumn{1}{c|}{3.86}                                 & 39.64                                 & \multicolumn{1}{c|}{{\color[HTML]{FE0000} \textbf{0.02}}} & \multicolumn{1}{c|}{{\color[HTML]{FE0000} \textbf{0.04}}} & 0.09                                  \\ \cline{2-14} 
\multirow{-5}{*}{\textbf{DAG 3}}                       & \textbf{CL}                       & \multicolumn{1}{c|}{{\color[HTML]{FE0000} \textbf{0.14 (0.03\%)}}}  & \multicolumn{1}{c|}{{\color[HTML]{000000} -42.95 (8.59\%)}}           & 108.11 (21.62\%)                                                    & \multicolumn{1}{c|}{{\color[HTML]{FE0000} \textbf{-0.62 (-0.12\%)}}} & \multicolumn{1}{c|}{{\color[HTML]{000000} -102.32 (20.46\%)}}         & 116.77 (23.35\%)                                & \multicolumn{1}{c|}{{\color[HTML]{FE0000} \textbf{0.11}}} & \multicolumn{1}{c|}{{\color[HTML]{000000} 0.74}}          & 24.32                                 & \multicolumn{1}{c|}{{\color[HTML]{FE0000} \textbf{0.02}}} & \multicolumn{1}{c|}{{\color[HTML]{FE0000} \textbf{0.56}}} & 2.10                                  \\ \hline
\multicolumn{1}{|l|}{}                                 & \textbf{Unweighted}               & \multicolumn{1}{c|}{-71.92 (14.38\%)}                               & \multicolumn{1}{c|}{-71.67 (14.33\%)}                                 & -42.33 (8.47\%)                                                     & \multicolumn{1}{c|}{-235.72 (47.14\%)}                               & \multicolumn{1}{c|}{-235.24 (47.05\%)}                                & -201.69 (40.34\%)                               & \multicolumn{1}{c|}{1}                                    & \multicolumn{1}{c|}{1}                                    & 1                                     & \multicolumn{1}{c|}{1}                                    & \multicolumn{1}{c|}{1}                                    & 1                                     \\ \cline{2-14} 
\multicolumn{1}{|l|}{}                                 & \textbf{PL}                       & \multicolumn{1}{c|}{{\color[HTML]{FE0000} \textbf{0.07 (0.01\%)}}}  & \multicolumn{1}{c|}{{\color[HTML]{000000} -60.53 (12.11\%)}}          & 82.38 (16.47\%)                                                     & \multicolumn{1}{c|}{{\color[HTML]{FE0000} \textbf{1.09 (0.22\%)}}}   & \multicolumn{1}{c|}{{\color[HTML]{000000} -198.12 (39.62\%)}}         & 173.74 (34.75\%)                                & \multicolumn{1}{c|}{{\color[HTML]{FE0000} \textbf{0.08}}} & \multicolumn{1}{c|}{{\color[HTML]{000000} 0.75}}          & 3.33                                  & \multicolumn{1}{c|}{{\color[HTML]{FE0000} \textbf{0.01}}} & \multicolumn{1}{c|}{{\color[HTML]{000000} 0.72}}          & 0.75                                  \\ \cline{2-14} 
\multicolumn{1}{|l|}{}                                 & \textbf{SR}                       & \multicolumn{1}{c|}{-30.11 (6.02\%)}                                & \multicolumn{1}{c|}{-30.23 (6.05\%)}                                  & \multicolumn{1}{l|}{21.63 (4.33\%)}                                 & \multicolumn{1}{c|}{-145.59 (29.20\%)}                               & \multicolumn{1}{c|}{-145.58 (29.12\%)}                                & -61.17 (12.23\%)                                & \multicolumn{1}{c|}{0.24}                                 & \multicolumn{1}{c|}{{\color[HTML]{FE0000} \textbf{0.24}}} & {\color[HTML]{FE0000} \textbf{0.41}}  & \multicolumn{1}{c|}{0.39}                                 & \multicolumn{1}{c|}{{\color[HTML]{FE0000} \textbf{0.39}}} & {\color[HTML]{FE0000} \textbf{0.10}}  \\ \cline{2-14} 
\multicolumn{1}{|l|}{}                                 & \textbf{PS}                       & \multicolumn{1}{c|}{{\color[HTML]{FE0000} \textbf{2.96 (0.59\%)}}}  & \multicolumn{1}{c|}{{\color[HTML]{FE0000} \textbf{-20.03 (4.01\%)}}}  & \multicolumn{1}{l|}{{\color[HTML]{FE0000} \textbf{12.68 (2.54\%)}}} & \multicolumn{1}{c|}{{\color[HTML]{FE0000} \textbf{29.11 (5.82\%)}}}  & \multicolumn{1}{c|}{{\color[HTML]{FE0000} \textbf{-57.71 (11.54\%)}}} & {\color[HTML]{FE0000} \textbf{31.38 (6.28 \%)}} & \multicolumn{1}{c|}{{\color[HTML]{FE0000} \textbf{0.15}}} & \multicolumn{1}{c|}{0.52}                                 & {\color[HTML]{FE0000} \textbf{0.37}}  & \multicolumn{1}{c|}{{\color[HTML]{FE0000} \textbf{0.10}}} & \multicolumn{1}{c|}{{\color[HTML]{FE0000} \textbf{0.40}}} & {\color[HTML]{FE0000} \textbf{0.11}}  \\ \cline{2-14} 
\multicolumn{1}{|l|}{\multirow{-5}{*}{\textbf{DAG 4}}} & \textbf{CL}                       & \multicolumn{1}{c|}{{\color[HTML]{FE0000} \textbf{0.45 (0.09\%)}}}  & \multicolumn{1}{c|}{{\color[HTML]{000000} -51.03 (10.21\%)}}          & \multicolumn{1}{l|}{68.70 (13.7\%)}                                 & \multicolumn{1}{c|}{{\color[HTML]{FE0000} \textbf{-0.65 (0.13\%)}}}  & \multicolumn{1}{c|}{{\color[HTML]{000000} -174.77 (34.75\%)}}         & 129.91 (25.98\%)                                & \multicolumn{1}{c|}{{\color[HTML]{FE0000} \textbf{0.08}}} & \multicolumn{1}{c|}{{\color[HTML]{000000} 0.72}}          & 2.38                                  & \multicolumn{1}{c|}{{\color[HTML]{FE0000} \textbf{0.01}}} & \multicolumn{1}{c|}{{\color[HTML]{000000} 0.57}}          & 0.42                                  \\ \hline
\end{tabular}
\end{adjustbox}

\caption{Bias and RMSE Comparison between the unweighted and four weighted methods in DAGs 1, 2, 3, 4, under simulation setups 1,2 and 3. The best models in terms of bias and RMSE under each setup have been highlighted in red separately. Unweighted : Unweighted Logistic Regression, SR : Simplex regression, PL : Pseudolikelihood, PS : Post Stratification and CL : Calibration.}
\label{tab:Table 3}
\end{table}

\begin{table}[H]
\centering
\begin{adjustbox}{width=1\textwidth}
\begin{tabular}{|c|c|c|}
\hline
\textbf{Variables}                                                    & \textbf{\begin{tabular}[c]{@{}c@{}}MGI \\ $(N = 80947)$\end{tabular}}                 & \textbf{\begin{tabular}[c]{@{}c@{}}NHANES\\ $(N=5153)$\end{tabular}}                  \\ \hline
\textbf{Cancer}                                                       & \begin{tabular}[c]{@{}c@{}}Yes (48.7\%)\\ No  (51.2\%)\end{tabular}                   & \begin{tabular}[c]{@{}c@{}}Yes (10.3\%)\\ No  (89.7\%)\end{tabular}                   \\ \hline
\textbf{Sex}                                                       & \begin{tabular}[c]{@{}c@{}}Female (53.8\%)\\ Male    (46.2\%)\end{tabular}            & \begin{tabular}[c]{@{}c@{}}Female (51.8\%)\\ Male     (48.2\%)\end{tabular}           \\ \hline
\textbf{Age}                                                          & 57.5 (18.1)                                                                           & 51.2 (17.6)                                                                           \\ \hline
\textbf{Race}                                                         & \begin{tabular}[c]{@{}c@{}}Non-Hispanic White (85.3\%)\\ Others (14.7\%)\end{tabular} & \begin{tabular}[c]{@{}c@{}}Non-Hispanic White (34.3\%)\\ Others (66.7\%)\end{tabular} \\ \hline
\textbf{BMI}(kg/{\rm m}$^2$)                                                          & 29.9 (7.26)                                                                           & 29.8 (7.4)                                                                            \\ \hline
\textbf{CHD}                                                          & \begin{tabular}[c]{@{}c@{}}Yes : 16.5\%\\ No:   83.5\%\end{tabular}                   & \begin{tabular}[c]{@{}c@{}}Yes : 4.6\%\\ No :  95.4\%\end{tabular}                    \\ \hline
\textbf{Diabetes}                                                     & \begin{tabular}[c]{@{}c@{}}Yes : 33.3\%\\ No : 66.7\%\end{tabular}                    & \begin{tabular}[c]{@{}c@{}}Yes : 15.7\%\\ No :  84.3\%\end{tabular}                   \\ \hline
\textbf{\begin{tabular}[c]{@{}c@{}}Current \\ Smoking\end{tabular}} & \begin{tabular}[c]{@{}c@{}}Yes : 9.8\%\\ No : 90.2\%\end{tabular}                     & \begin{tabular}[c]{@{}c@{}}Yes: 18.2\%\\ No: 81.8\%\end{tabular}                      \\ \hline
\end{tabular}
\end{adjustbox}
\caption{Descriptive Summaries of the different variables of interest in both MGI and NHANES data. The statistics for NHANES provided here are unweighted. CHD stands for Coronary heart disease. For continuous variables we reported Mean (Sd).}
\label{tab:TableMGI}
\end{table}

\newpage
\pagenumbering{arabic}
\singlespacing
\section*{Supplementary Section}
\setcounter{figure}{0}
\renewcommand{\thefigure}{S\arabic{figure}}
\setcounter{section}{1}
\renewcommand{\thesection}{S\arabic{section}}
In this section, we present the detailed proofs of the results and theorem stated above.

\subsection{Distribution of $D|\boldsymbol Z_1,\boldsymbol Z_2,S=1$}\label{sec:pr1}
By Bayes Theorem  we obtain that, 
\begin{align*}
   & P(D=1|\boldsymbol Z_1,\boldsymbol Z_2,S=1)=\\
   & \frac{P(S=1|D=1,\boldsymbol Z_1,\boldsymbol Z_2)P(D=1|\boldsymbol Z_1,\boldsymbol Z_2)}{P(S=1|D=1,\boldsymbol Z_1,\boldsymbol Z_2)P(D=1|\boldsymbol Z_1,\boldsymbol Z_2) + P(S=1|D=0,\boldsymbol Z_1,\boldsymbol Z_2)P(D=0|\boldsymbol Z_1,\boldsymbol Z_2)}\cdot
\end{align*}
From equation \eqref{eq:eq1} in Section \ref{sec:DAGs}, we know that, $r(\boldsymbol Z_1,\boldsymbol Z_2)=\frac{P(S=1|D=1,\boldsymbol Z_1,\boldsymbol Z_2)}{P(S=1|D=0,\boldsymbol Z_1,\boldsymbol Z_2)}$ and $\text{logit}(P(D=1|\boldsymbol Z_1,\boldsymbol Z_2))=\theta_0+\underset{}{\boldsymbol\theta_1} \boldsymbol Z_1 +\boldsymbol \theta_2 \boldsymbol Z_2$. Therefore, dividing both the numerator and denominator by $P(S=1|D=1,\boldsymbol Z_1,\boldsymbol Z_2)$, we obtain
\begin{align*}
    &P(D=1|\boldsymbol Z_1,\boldsymbol Z_2,S=1)  =\frac{P(D=1|\boldsymbol Z_1,\boldsymbol Z_2)}{P(D=1|\boldsymbol Z_1,\boldsymbol Z_2) + \frac{1}{r(\boldsymbol Z_1,\boldsymbol Z_2)} \cdot (1-P(D=1|\boldsymbol Z_1,\boldsymbol Z_2))}\\
    & = \frac{e^{\theta_0+\underset{}{\boldsymbol\theta_1}.\boldsymbol Z_1 + \boldsymbol \theta_2.\boldsymbol Z_2}}{e^{\theta_0+\underset{}{\boldsymbol\theta_1}\boldsymbol Z_1 + \boldsymbol \theta_2\boldsymbol Z_2}+ \frac{1}{r(\boldsymbol Z_1,\boldsymbol Z_2)}}\\
     &\frac{P(D=1|\boldsymbol Z_1,\boldsymbol Z_2,S=1)}{1- P(D=1|\boldsymbol Z_1,\boldsymbol Z_2,S=1)}=e^{\theta_0+\underset{}{\boldsymbol\theta_1}\boldsymbol Z_1 + \boldsymbol \theta_2\boldsymbol Z_2}\cdot r(\boldsymbol Z_1,\boldsymbol Z_2)\cdot\\
     &\text{Thus we obtain,} \hspace{0.2cm}
     \text{logit}(P(D=1|\boldsymbol Z_1,\boldsymbol Z_2,S=1))=\theta_0 +\underset{}{\boldsymbol\theta_1}\boldsymbol Z_1 + \boldsymbol \theta_2\boldsymbol Z_2 +\text{log}(r(\boldsymbol Z_1,\boldsymbol Z_2))\cdot
\end{align*}

\subsection{Expression of $r(\boldsymbol Z_1,\boldsymbol Z_2)$ in different Setups}\label{sec:pr2}
The original disease model in absence of selection bias as in equation \eqref{eq:eq1} in Section \ref{sec:DAGs} is :
\begin{align*}
    \text{logit}(P(D=1|\boldsymbol Z_1,\boldsymbol Z_2))=\theta_0+\underset{}{\boldsymbol\theta_1} \boldsymbol Z_1 +\boldsymbol \theta_2 \boldsymbol Z_2\cdot
\end{align*}
However in presence of selection bias from equation \eqref{eq:eq2}, the observed disease model $D|\boldsymbol Z_1,\boldsymbol Z_2,S=1$ is :
\begin{align*}
    \text{logit}(P(D=1|\boldsymbol Z_1,\boldsymbol Z_2,S=1))=\theta_0+\underset{}{\boldsymbol\theta_1} \boldsymbol Z_1 +\boldsymbol \theta_2 \boldsymbol Z_2 + \text{log}(r(\boldsymbol Z_1,\boldsymbol Z_2))\cdot
\end{align*}
\subsubsection{Example DAG 1: unbiased case}\label{sec:pr2.1}
Under DAG 1 of Figure \ref{fig:Fig2}, we observe that all the paths from $\boldsymbol Z_1$ to $S$, as well as from $\boldsymbol Z_2$ to $S$ are blocked when we condition on $D$. Therefore, both $\boldsymbol Z_1$ and  $S$ and $\boldsymbol Z_2$ and $S$ are d-separated when we condition on $D$. Therefore $P(S=1|D=1,\boldsymbol Z_1,\boldsymbol Z_2)=P(S=1|D=1)$ and $P(S=1|D=0,\boldsymbol Z_1,\boldsymbol Z_2)=P(S=1|D=0)$ and as a result, the simplified $r(\boldsymbol Z_1,\boldsymbol Z_2)=\frac{P(S=1|D=1)}{P(S=1|D=0)}$ independent of $\boldsymbol Z_1$ and $\boldsymbol Z_2$ or a constant function of $(\boldsymbol Z_1,\boldsymbol Z_2)$. As a result, the observed model simplifies to :
\begin{align*}
    \text{logit}(P(D=1|\boldsymbol Z_1,\boldsymbol Z_2,S=1))=(\theta_0+\text{log}(r)) + \underset{}{\boldsymbol\theta_1} \boldsymbol Z_1 +\boldsymbol \theta_2 \boldsymbol Z_2 =\theta_0' + \underset{}{\boldsymbol\theta_1} \boldsymbol Z_1 +\boldsymbol \theta_2 \boldsymbol Z_2\cdot
\end{align*}
Therefore, in this case when we perform a logistic regression with $D$ and $\boldsymbol Z_1,\boldsymbol Z_2$ as response and predictors, $\widehat{\boldsymbol\theta_1}$ and $\widehat{\boldsymbol \theta_2}$ consistently estimate the actual disease model parameters $\boldsymbol\theta_1$ and $\boldsymbol \theta_2$. But still $\widehat{\theta_0'}$ is biased since it has the extra $\text{log(r)}$ term.
\subsubsection{Example  DAG 2: $Z_1\rightarrow W$ arrow induced bias for coefficient of $Z_1$}\label{sec:pr2.2}
Under DAG 2 of Figure \ref{fig:Fig2}, we observe that conditioned on D and $\boldsymbol Z_1$, all the paths between from $\boldsymbol Z_2$ to $S$ are blocked. Therefore, given $\boldsymbol Z_1$ and $D$, $\boldsymbol Z_2$ and $S$ are d-separated which implies,
\begin{equation}
   r(\boldsymbol Z_1,\boldsymbol Z_2) =\frac{P(S=1|D=1,\boldsymbol Z_1,\boldsymbol Z_2)}{P(S=1|D=0,\boldsymbol Z_1,\boldsymbol Z_2)}=\frac{P(S=1|D=1,\boldsymbol Z_1)}{P(S=1|D=0,\boldsymbol Z_1)}\cdot
\end{equation}
\begin{equation}
     r(\boldsymbol Z_1,\boldsymbol Z_2)=\frac{\int P(S=1|D=1,\boldsymbol Z_1,\boldsymbol W)dP(\boldsymbol W|D=1,\boldsymbol Z_1)}{\int P(S=1|D=0,\boldsymbol Z_1,\boldsymbol W)dP(\boldsymbol W|D=0,\boldsymbol Z_1)}\cdot
\end{equation}
Under DAG 2, conditioned on $\boldsymbol W$ and $D$, $\boldsymbol Z_1$ is independent of $S$ since given $\boldsymbol W$ and $D$ all the paths from $\boldsymbol Z_1$ to S are blocked, Thus,
   \begin{equation}
     r(\boldsymbol Z_1,\boldsymbol Z_2)=\frac{\int P(S=1|D=1,\boldsymbol W)dP(\boldsymbol W|\boldsymbol Z_1)}{\int P(S=1|D=0,\boldsymbol W)dP(\boldsymbol W|\boldsymbol Z_1)}\cdot\label{eq:ap1}
   \end{equation}
Since $r(\boldsymbol Z_1,\boldsymbol Z_2)$ is a function of $\boldsymbol Z_1$ only, using an exactly similar argument as in DAG 1, we conclude that the naive method unbiasedly estimates the coefficient $\boldsymbol \theta_2$ of $\boldsymbol Z_2$. Equation \eqref{eq:ap1}, shows that the dependence $r(\boldsymbol Z_1,\boldsymbol Z_2)$ on $\boldsymbol Z_1$ is only through the second term (distribution of $\boldsymbol W|\boldsymbol Z_1$) of the integral for both the numerator and denominator.                                                                        

\subsubsection{Example DAG 3: $Z_1\rightarrow W$ and $Z_2\rightarrow S$ induced bias for coefficients of $Z_1$ and $Z_2$ }\label{sec:pr2.3}
Under DAG 3 of Figure \ref{fig:Fig2},  $\boldsymbol Z_2$ is a direct parent of the selection indicator. Therefore, under this setup, $r(\boldsymbol Z_1,\boldsymbol Z_2)$ is a function of both the disease model predictors $(\boldsymbol Z_1, \boldsymbol Z_2)$. The expression of $r(\boldsymbol Z_1,\boldsymbol Z_2)$ is given by,
\begin{align*}
   r(\boldsymbol Z_1,\boldsymbol Z_2) &=\frac{P(S=1|D=1,\boldsymbol Z_1,\boldsymbol Z_2)}{P(S=1|D=0,\boldsymbol Z_1,\boldsymbol Z_2)}\\
   & =\frac{\int P(S=1|D=1,\boldsymbol Z_1,\boldsymbol Z_2,\boldsymbol W)dP(\boldsymbol W|D=1,\boldsymbol Z_1,\boldsymbol Z_2)}{\int P(S=1|D=0,\boldsymbol Z_1,\boldsymbol Z_2,\boldsymbol W)dP(\boldsymbol W|D=0,\boldsymbol Z_1,\boldsymbol Z_2)}\cdot
\end{align*}
Under DAG 3, conditioned on $(\boldsymbol W,D,\boldsymbol Z_2)$, $\boldsymbol Z_1$ is independent of $S$ since given $(\boldsymbol W,D,\boldsymbol Z_2)$ all paths from $\boldsymbol Z_1$ to S are blocked. Thus,
\begin{equation}
    r(\boldsymbol Z_1,\boldsymbol Z_2) =\frac{\int P(S=1|D=1,\boldsymbol W,\boldsymbol Z_2)dP(\boldsymbol W|\boldsymbol Z_1)}{\int P(S=1|D=0,\boldsymbol W,\boldsymbol Z_2)dP(\boldsymbol W|\boldsymbol Z_1)}\cdot \label{eq:apdag3}
\end{equation}
Since $r(\boldsymbol Z_1,\boldsymbol Z_2)$ is a function of both $\boldsymbol Z_1$ and $\boldsymbol Z_2$, the naive method fails to produce unbiased estimates for both the disease parameters. Equation \eqref{eq:apdag3} also shows that the dependence $r(\boldsymbol Z_1,\boldsymbol Z_2)$ on $\boldsymbol Z_1$ is only through the second term (distribution of $\boldsymbol W|\boldsymbol Z_1$) of the integral for both the numerator and denominator. On the other hand the first term in the integral, namely distribution of $(S|D,\boldsymbol Z_1\boldsymbol Z_2)$ for both numerator and denominator summarizes the dependence of $r(\boldsymbol Z_1,\boldsymbol Z_2)$ on $\boldsymbol Z_2$.        
\subsubsection{Example DAG 4: strong dependence, increased bias for coefficients of $Z_1$ and $Z_2$}\label{sec:pr2.4}
The additional arrows $Z_2\rightarrow W$ and $D\rightarrow W$ lead to increase in dependence between the selection and disease models. Similar to DAG 3, $r(\boldsymbol Z_1,\boldsymbol Z_2)$ is a function of both the disease model predictors $(\boldsymbol Z_1, \boldsymbol Z_2)$. The expression of $r(\boldsymbol Z_1,\boldsymbol Z_2)$ is given by,
\begin{align*}
   r(\boldsymbol Z_1,\boldsymbol Z_2) &=\frac{P(S=1|D=1,\boldsymbol Z_1,\boldsymbol Z_2)}{P(S=1|D=0,\boldsymbol Z_1,\boldsymbol Z_2)}\\
   & =\frac{\int P(S=1|D=1,\boldsymbol Z_1,\boldsymbol Z_2,\boldsymbol W)dP(\boldsymbol W|D=1,\boldsymbol Z_1,\boldsymbol Z_2)}{\int P(S=1|D=0,\boldsymbol Z_1,\boldsymbol Z_2,\boldsymbol W)dP(\boldsymbol W|D=0,\boldsymbol Z_1,\boldsymbol Z_2)}\cdot
\end{align*}
Under DAG 4, conditioned on $(\boldsymbol W,D,\boldsymbol Z_2)$, $\boldsymbol Z_1$ is independent of $S$ since given $(\boldsymbol W,D,\boldsymbol Z_2)$ all the paths from $\boldsymbol Z_1$ to S are blocked. Thus,
\begin{equation}
    r(\boldsymbol Z_1,\boldsymbol Z_2)=\frac{\int P(S=1|D=1,\boldsymbol W,\boldsymbol Z_2)dP(\boldsymbol W|\boldsymbol Z_1,\boldsymbol Z_2,D=1)}{\int P(S=1|D=0,\boldsymbol W,\boldsymbol Z_2)dP(\boldsymbol W|\boldsymbol Z_1,\boldsymbol Z_2,D=0)}\cdot\label{eq:apdag4}
\end{equation}
  Since $r(\boldsymbol Z_1,\boldsymbol Z_2)$ is a function of both $\boldsymbol Z_1$ and $\boldsymbol Z_2$, the naive method fails to produce unbiased estimates for both the disease parameters. Step 3 also shows that the dependence $r(\boldsymbol Z_1,\boldsymbol Z_2)$ on $\boldsymbol Z_1$ is only through the second term (distribution of $\boldsymbol W|\boldsymbol Z_1$) of the integral for both the numerator and denominator. On the other hand both terms in the integral depends on $\boldsymbol Z_2$. Therefore, the dependence of $r(\boldsymbol Z_1,\boldsymbol Z_2)$ on $\boldsymbol Z_2$ is extremely high in this scenario.

\subsection{One Step IPW regression }\label{sec:ipw}
\textbf{Assumptions}\\

\noindent
\textit{A}1. $\boldsymbol \theta^* \in \Theta$, where $\Theta$ is compact.\\
\noindent
\textit{A}2.  We did not consider finite population inference here unlike most survey sampling literature. For all the asymptotic results we let $N \rightarrow \infty$.
\noindent
\textit{A}3.  All the variables, including the selection indicator $S$, $\boldsymbol X=(D,\boldsymbol Z_1,\boldsymbol Z_2$) are considered to be random which are generally considered fixed in finite population literature.\\

\begin{theorem}
     Under assumptions \textit{A}1, \textit{A}2 and \textit{A}3 if the true selection weights are known, then estimator $\widehat{\boldsymbol \theta}=(\widehat{\theta}_0,\widehat{\boldsymbol \theta}_1, \widehat{\boldsymbol \theta}_2)$ from equation \eqref{eq:eq6} is consistent for $\boldsymbol\theta^*=(\theta_0^*,\boldsymbol \theta_1^*, \boldsymbol \theta_2^*)$, where $\boldsymbol\theta^*$ is the true value of $\boldsymbol \theta$ satisfying equation \eqref{eq:eq1}. \label{thm1}
\end{theorem}
\begin{proof}
    Let 
$$\phi_N(\boldsymbol\theta)=\frac{1}{N}\sum_{i=1}^{N}\frac{S_i}{\pi(\boldsymbol X_i)}\left\{D_i \boldsymbol Z_i-\frac{e^{\boldsymbol\theta'\boldsymbol Z_i}}{(1+e^{\boldsymbol\theta'\boldsymbol Z_i})}\cdot \boldsymbol Z_i\right\}\cdot$$
From \citet{tsiatis2006semiparametric}, under assumptions \textit{A}1, \textit{A}2 and \textit{A}3 it is enough to show that $\mathbb{E}(\phi_N(\boldsymbol\theta^*))=0$, in order to prove that $\widehat{\boldsymbol\theta}\xrightarrow{p}\boldsymbol\theta^*$, where $\widehat{\boldsymbol\theta}$ is obtained from solving $\phi_N(\boldsymbol\theta)=\mathbf{0}$.
\begin{align*}
     \mathbb{E}[\phi_N(\boldsymbol\theta^*)] & = \mathbb{E}\left[\frac{1}{N}\sum_{i=1}^{N}\frac{S_i}{\pi(\boldsymbol X_i)}\left\{D_i \boldsymbol Z_i-\frac{e^{\boldsymbol\theta^{*'}\boldsymbol Z_i}}{(1+e^{\boldsymbol\theta^{*'}\boldsymbol Z_i})}\cdot \boldsymbol Z_i\right\}\right]\\
     & =\frac{1}{N}
   \sum_{i=1}^{N} \mathbb{E}\left[\frac{S_i}{\pi(\boldsymbol X_i)}\left\{D_i \boldsymbol Z_i-\frac{e^{\boldsymbol\theta'\boldsymbol Z_i}}{(1+e^{\boldsymbol\theta^{*'}\boldsymbol Z_i})}\cdot \boldsymbol Z_i\right\}\right]\\
   & = \frac{1}{N}
   \sum_{i=1}^{N} \mathbb{E}_{\boldsymbol X_i,\boldsymbol Z_{1i}}\left[\mathbb{E}\left\{\frac{S_i}{\pi(\boldsymbol X_i)}\left.\left(D_i \boldsymbol Z_i-\frac{e^{\boldsymbol\theta^{*'}\boldsymbol Z_i}}{(1+e^{\boldsymbol\theta^{*'}\boldsymbol Z_i})}\cdot \boldsymbol Z_i\right)\right| \boldsymbol X_i,\boldsymbol Z_{1i}\right\}\right]\\
   & = \frac{1}{N}
   \sum_{i=1}^{N} \mathbb{E}_{\boldsymbol X_i,\boldsymbol Z_{1i}}\left[\left\{D_i \boldsymbol Z_i-\frac{e^{\boldsymbol\theta ^{*'}\boldsymbol Z_i}}{(1+e^{\boldsymbol\theta ^{*'}\boldsymbol Z_i})}\cdot \boldsymbol Z_i\right\}\cdot \frac{1}{\pi(\boldsymbol X_i)}\mathbb{E}\left(S_i|\boldsymbol X_i,\boldsymbol Z_{1i}\right)\right]\cdot
\end{align*}
Since $S\independent \boldsymbol Z_1|\boldsymbol X$ therefore $\mathbb{E}(S_i|\boldsymbol X_i,\boldsymbol Z_{1i})=
\mathbb{E}(S_i|\boldsymbol X_i)=\pi(\boldsymbol X_i)$. Using this result and equation \eqref{eq:eq1} we obtain
\begin{align*}
     &\mathbb{E}[\phi_N(\boldsymbol\theta^*)]  =\frac{1}{N}
   \sum_{i=1}^{N} \mathbb{E}_{\boldsymbol X_i,\boldsymbol Z_{1i}}\left\{D_i \boldsymbol Z_i-\frac{e^{\boldsymbol\theta^{*'}\boldsymbol Z_i}}{(1+e^{\boldsymbol\theta^{*'}\boldsymbol Z_i})}\cdot \boldsymbol Z_i\right\}\\
   &=\frac{1}{N}
   \sum_{i=1}^{N} \mathbb{E}_{D_i,\boldsymbol Z_i}\left\{D_i \boldsymbol Z_i-\frac{e^{\boldsymbol\theta^{*'}\boldsymbol Z_i}}{(1+e^{\boldsymbol\theta^{*'}\boldsymbol Z_i})}\cdot \boldsymbol Z_i\right\}\\
   & = \frac{1}{N}
   \sum_{i=1}^{N} \mathbb{E}_{\boldsymbol z_i}\left[\mathbb{E}\left \{\left. \left(D_i \boldsymbol Z_i-\frac{e^{\boldsymbol\theta^{*'}\boldsymbol Z_i}}{(1+e^{\boldsymbol\theta^{*'}\boldsymbol Z_i})}\cdot \boldsymbol Z_i\right) \right|\boldsymbol Z_i\right\}\right] \\
   &= \frac{1}{N}
   \sum_{i=1}^{N} \mathbb{E}\left\{\frac{e^{\boldsymbol\theta^{*'}\boldsymbol Z_i}}{(1+e^{\boldsymbol\theta^{*'}\boldsymbol Z_i})}\cdot \boldsymbol Z_i-\frac{e^{\boldsymbol\theta^{*'}\boldsymbol Z_i}}{(1+e^{\boldsymbol\theta^{*'}\boldsymbol Z_i})} \cdot \boldsymbol Z_i\right\}=\mathbf{0}\cdot
\end{align*}
The last step is obtained from the relation between $D$ and $(\boldsymbol Z_1,\boldsymbol Z_2)$ given by equation \eqref{eq:eq1}.
\end{proof}

\begin{theorem}
    Under assumptions of Theorem \ref{thm1} when the selection weights are known and we do not take into consideration estimation of selection model parameter, the asymptotic distribution of $\widehat{\boldsymbol \theta}$ is given by 
$$\sqrt{N}(\widehat{\boldsymbol \theta}-\boldsymbol\theta^*)\xrightarrow{d}\mathcal{N}(0,V)\hspace{0.5cm} \text{as} \hspace{0.5cm} N\rightarrow \infty\cdot$$
where%, the expression of $V$ is given by, 
\begin{align*}
    & V=(G_{\boldsymbol \theta^*})^{-1}\cdot \mathbb{E}[\boldsymbol g\cdot \boldsymbol g']\cdot(G_{\boldsymbol \theta^*}^{-1})^{'}\hspace{0.8cm} G_{\boldsymbol \theta^*} = \mathbb{E}\left\{-\frac{S}{\pi(\boldsymbol X)}\cdot \frac{e^{\boldsymbol \theta^{*'}\boldsymbol Z}}{(1+e^{\boldsymbol \theta^{*'}\boldsymbol Z})^2}\cdot \boldsymbol Z\boldsymbol Z'\right\}\\
    & \boldsymbol g(\boldsymbol \theta^*) = \frac{S}{\pi(\boldsymbol X)}\left\{D \boldsymbol Z-\frac{e^{\boldsymbol \theta^{*'}\boldsymbol Z}}{(1+e^{\boldsymbol \theta^{*'}\boldsymbol Z})}\cdot \boldsymbol Z\right\}\cdot
\end{align*}\label{thm2}
\end{theorem}
\begin{proof}
By \citet{tsiatis2006semiparametric}'s arguments for a Z-estimation problem under assumptions of Theorem \ref{thm1} we obtain  
$$\sqrt{N}(\widehat{\boldsymbol \theta}-\boldsymbol\theta^*)\xrightarrow{d}\mathcal{N}(0,V)\hspace{0.5cm} \text{as} \hspace{0.5cm} N\rightarrow \infty\cdot$$
where%, the expression of $V$ is given by, 
\begin{align*}
    & V=(G_{\boldsymbol \theta^*})^{-1}\cdot \mathbb{E}[\boldsymbol g\cdot \boldsymbol g']\cdot(G_{\boldsymbol \theta^*}^{-1})^{'}\hspace{0.8cm} G_{\boldsymbol \theta^*} = \left.\mathbb{E}\left\{\frac{\partial g(\boldsymbol \theta^*)}{\partial \boldsymbol \theta}\right\}\right|_{\boldsymbol\theta=\boldsymbol \theta^*}\\
    & \boldsymbol g(\boldsymbol \theta^*) = \frac{S}{\pi(\boldsymbol X)}\left\{D \boldsymbol Z-\frac{e^{\boldsymbol \theta^{*'}\boldsymbol Z}}{(1+e^{\boldsymbol \theta^{*'}\boldsymbol Z})}\cdot \boldsymbol Z\right\}\cdot
\end{align*}
    This proof just requires the calculation of $G_{\boldsymbol \theta^*}$.
\subsubsection*{Calculation for $G_{\boldsymbol \theta^*}$}
\begin{align*}
    & \left.\frac{\partial \boldsymbol g(\boldsymbol \theta)}{\partial \boldsymbol \theta}\right|_{\boldsymbol \theta= \boldsymbol \theta^*}=\left.\frac{\partial}{\partial \boldsymbol \theta}\left[\frac{S}{\pi(\boldsymbol X)}\left\{D\boldsymbol Z-\frac{e^{\boldsymbol \theta^{*'}\boldsymbol Z}}{(1+e^{\boldsymbol \theta^{*'}\boldsymbol Z})}\cdot \boldsymbol Z\right\}\right]\right|_{\boldsymbol \theta= \boldsymbol \theta^*}\\
    & = -\frac{S}{\pi(\boldsymbol X)} \cdot \frac{e^{\boldsymbol \theta^{*'}\boldsymbol Z}}{(1+e^{\boldsymbol \theta^{*'}\boldsymbol Z})^2}\cdot \boldsymbol Z\boldsymbol Z'
\end{align*} 
Therefore we obtain 
\begin{align*}
     G_{\boldsymbol \theta^*} = \mathbb{E}\left\{-\frac{S}{\pi(\boldsymbol X)} \cdot \frac{e^{\boldsymbol \theta^{*'} \boldsymbol Z}}{(1+e^{\boldsymbol \theta^{*'}\boldsymbol Z})^2}\cdot \boldsymbol Z \boldsymbol Z'\right\}\cdot
\end{align*}
\end{proof}
\noindent
Next we derive a consistent estimator of asymptotic variance of $\widehat{\boldsymbol\theta}$ when the selection probabilities are known. We use this variance estimator for SR and PS. Apart from assumptions \textit{A}1, \textit{A}2 and \textit{A}3, we make an additional assumption to derive a consistent estimator of asymptotic variance of $\widehat{\boldsymbol\theta}$.\\

\noindent
\textit{A}4. $\mathbb{E}[\sup_{\boldsymbol \theta \in \Theta}(G_{\boldsymbol \theta})]<\infty$ and $\mathbb{E}[\sup_{\boldsymbol \theta \in \Theta}\{\boldsymbol g(\boldsymbol \theta) \boldsymbol g(\boldsymbol \theta)'\}]<\infty$.\\

\noindent
\begin{theorem}
     Under assumptions \textit{A}1, \textit{A}2, \textit{A}3 and \textit{A}4, $\frac{1}{N} \cdot \widehat{G_{\boldsymbol \theta}}^{-1}\cdot \widehat{E}\cdot (\widehat{G_{\boldsymbol \theta}}^{-1})^{'}$ is a consistent estimator of the asymptotic variance of $\widehat{\boldsymbol \theta}$ where
\begin{align*}
    & \widehat{G_{\boldsymbol \theta}}=\frac{1}{N}\sum_{i=1}^N \left\{\frac{S_i}{\pi(\boldsymbol X_i)}\cdot \frac{e^{\widehat{\boldsymbol \theta}'\boldsymbol Z_i} }{(1+e^{\widehat{\boldsymbol \theta}'\boldsymbol Z_i})^2}\cdot \boldsymbol Z_i \boldsymbol Z_i'\right\}\cdot \\
    & \widehat{E}=\frac{1}{N}\sum_{i =1}^N S_i \cdot \left\{\frac{1}{\pi(\boldsymbol X_i)}\right\}^2\left\{D_i-\frac{e^{\widehat{\boldsymbol \theta}'\boldsymbol Z_i}}{(1+e^{\widehat{\boldsymbol\theta}' \boldsymbol Z_i})}\right\}^2\cdot \boldsymbol Z_i \boldsymbol Z_i' \cdot
\end{align*}\label{thm3}
\end{theorem}
\begin{proof}
    First we prove that as $N \rightarrow \infty$
$$\frac{1}{N}\sum_{i=1}^N \left\{\frac{S_i}{\pi(\boldsymbol X_i)}\cdot \frac{e^{\widehat{\boldsymbol \theta}'\boldsymbol Z_i} }{(1+e^{\widehat{\boldsymbol \theta}'\boldsymbol Z_i})^2}\cdot \boldsymbol Z_i \boldsymbol Z_i'\right\}
\xrightarrow{p}\mathbb{E}\left\{\frac{S}{\pi(\boldsymbol X)}\cdot \frac{e^{\boldsymbol \theta^{*'}\boldsymbol Z}}{(1+e^{\boldsymbol \theta^{*'}\boldsymbol Z})^2}\cdot \boldsymbol Z\boldsymbol Z'\right\} \cdot$$
Using assumptions \textit{A}1, \textit{A}2, \textit{A}3 and \textit{A}4 and Uniform Law of Large Numbers (ULLN), we obtain 
\begin{align}
    \sup_{\boldsymbol \theta \in \Theta}\left|\left|
    \frac{1}{N}\sum_{i=1}^N \left\{\frac{S_i}{\pi(\boldsymbol X_i)}\cdot \frac{e^{\boldsymbol \theta'\boldsymbol Z_i} }{(1+e^{\boldsymbol \theta'\boldsymbol Z_i})^2}\cdot \boldsymbol Z_i \boldsymbol Z_i'\right\}-\mathbb{E}\left\{\frac{S}{\pi(\boldsymbol X)}\cdot \frac{e^{\boldsymbol \theta^{'}\boldsymbol Z}}{(1+e^{\boldsymbol \theta^{'}\boldsymbol Z})^2}\cdot \boldsymbol Z\boldsymbol Z'\right\}
    \right|\right|\xrightarrow{p}\mathbf{0}\cdot \label{eq:st11}
\end{align}
Since this above expression holds for any $\boldsymbol \theta \in \Theta$, therefore it is true for $\widehat{\boldsymbol \theta}$. This implies
\begin{align}
    \left|\left|
    \frac{1}{N}\sum_{i=1}^N \left. \left\{\frac{S_i}{\pi(\boldsymbol X_i)}\cdot \frac{e^{\widehat{\boldsymbol \theta}'\boldsymbol Z_i} }{(1+e^{\widehat{\boldsymbol \theta}'\boldsymbol Z_i})^2}\cdot \boldsymbol Z_i \boldsymbol Z_i'\right\}-\mathbb{E}\left\{\frac{S}{\pi(\boldsymbol X)}\cdot \frac{e^{\boldsymbol \theta^{'}\boldsymbol Z}}{(1+e^{\boldsymbol \theta^{'}\boldsymbol Z})^2}\cdot \boldsymbol Z\boldsymbol Z'\right\}\right|_{\boldsymbol \theta=\widehat{\boldsymbol \theta}}
    \right|\right|\xrightarrow{p}\mathbf{0}\cdot \label{eq:st12}
\end{align}
Using Triangle Inequality we obtain
\begin{align}
    &  \left|\left|\frac{1}{N}\sum_{i=1}^N \left\{\frac{S_i}{\pi(\boldsymbol X_i)}\cdot \frac{e^{\widehat{\boldsymbol \theta}'\boldsymbol Z_i} }{(1+e^{\boldsymbol \theta'\boldsymbol Z_i})^2}\cdot \boldsymbol Z_i \boldsymbol Z_i'\right\}-\mathbb{E}\left\{\frac{S}{\pi(\boldsymbol X)}\cdot \frac{e^{\boldsymbol \theta^{*'}\boldsymbol Z}}{(1+e^{\boldsymbol \theta^{*'}\boldsymbol Z})^2}\cdot \boldsymbol Z\boldsymbol Z'\right\}
    \right|\right|\leq \label{eq:st132}\\
    &  \left|\left|
    \frac{1}{N}\sum_{i=1}^N \left. \left\{\frac{S_i}{\pi(\boldsymbol X_i)}\cdot \frac{e^{\widehat{\boldsymbol \theta}'\boldsymbol Z_i} }{(1+e^{\widehat{\boldsymbol \theta}'\boldsymbol Z_i})^2}\cdot \boldsymbol Z_i \boldsymbol Z_i'\right\}-\mathbb{E}\left\{\frac{S}{\pi(\boldsymbol X)}\cdot \frac{e^{\boldsymbol \theta^{'}\boldsymbol Z}}{(1+e^{\boldsymbol \theta^{'}\boldsymbol Z})^2}\cdot \boldsymbol Z\boldsymbol Z'\right\}\right|_{\boldsymbol \theta=\widehat{\boldsymbol \theta}}
    \right|\right|+\label{eq:st131}\\
    & \left. \left|\left|
   \mathbb{E}\left\{\frac{S}{\pi(\boldsymbol X)}\cdot \frac{e^{\boldsymbol \theta^{'}\boldsymbol Z}}{(1+e^{\boldsymbol \theta^{'}\boldsymbol Z})^2}\cdot \boldsymbol Z\boldsymbol Z'\right\}\right|_{\boldsymbol \theta=\widehat{\boldsymbol \theta}}- \mathbb{E}\left\{\frac{S}{\pi(\boldsymbol X)}\cdot \frac{e^{\boldsymbol \theta^{*'}\boldsymbol Z}}{(1+e^{\boldsymbol \theta^{*'}\boldsymbol Z})^2}\cdot \boldsymbol Z\boldsymbol Z'\right\}
    \right|\right|\cdot \label{eq:st13}
\end{align}
Since we proved that $\widehat{\boldsymbol \theta}\xrightarrow{p} \boldsymbol \theta^{*}$ in Theorem 1, therefore by Continuous Mapping Theorem
\begin{align}
    \left. \left|\left|
   \mathbb{E}\left\{\frac{S}{\pi(\boldsymbol X)}\cdot \frac{e^{\boldsymbol \theta^{'}\boldsymbol Z}}{(1+e^{\boldsymbol \theta^{'}\boldsymbol Z})^2}\cdot \boldsymbol Z\boldsymbol Z'\right\}\right|_{\boldsymbol \theta=\widehat{\boldsymbol \theta}}- \mathbb{E}\left\{\frac{S}{\pi(\boldsymbol X)}\cdot \frac{e^{\boldsymbol \theta^{*'}\boldsymbol Z}}{(1+e^{\boldsymbol \theta^{*'}\boldsymbol Z})^2}\cdot \boldsymbol Z\boldsymbol Z'\right\}
    \right|\right|\xrightarrow{p} \mathbf{0}\cdot \label{eq:st14}
\end{align}
Using equations \eqref{eq:st12}, \eqref{eq:st132}, \eqref{eq:st131}, \eqref{eq:st13} and \eqref{eq:st14} we obtain 
$$\frac{1}{N}\sum_{i=1}^N \left\{\frac{S_i}{\pi(\boldsymbol X_i)}\cdot \frac{e^{\widehat{\boldsymbol \theta}'\boldsymbol Z_i} }{(1+e^{\widehat{\boldsymbol \theta}'\boldsymbol Z_i})^2}\cdot \boldsymbol Z_i \boldsymbol Z_i'\right\}
\xrightarrow{p}\mathbb{E}\left\{\frac{S}{\pi(\boldsymbol X)}\cdot \frac{e^{\boldsymbol \theta^{*'}\boldsymbol Z}}{(1+e^{\boldsymbol \theta^{*'}\boldsymbol Z})^2}\cdot \boldsymbol Z\boldsymbol Z'\right\}\cdot$$
Therefore we obtain  
$$-\widehat{G_{\boldsymbol \theta}}\xrightarrow{p} -G_{\boldsymbol \theta^*}\hspace{0.5cm}
    \text{which implies} \hspace{0.5cm} \widehat{G_{\boldsymbol \theta}}\xrightarrow{p} G_{\boldsymbol \theta^*}\cdot$$
Using the exact same set of arguments we obtain 
\begin{align*}
    & \widehat{E} =\frac{1}{N}\sum_{i =1}^N S_i \cdot \left\{\frac{1}{\pi(\boldsymbol X_i)}\right\}^2\left\{D_i-\frac{e^{\widehat{\boldsymbol \theta}'\cdot \boldsymbol Z_i}}{(1+e^{\hat{\theta}' \cdot \boldsymbol Z_i})}\right\}^2\cdot \boldsymbol Z_i \boldsymbol Z_i'\\
     & \xrightarrow{p}\mathbb{E}[\boldsymbol g(\boldsymbol \theta^*)\cdot \boldsymbol g(\boldsymbol \theta^*)']=  \mathbb{E}\left[S \cdot \left\{\frac{1}{\pi(\boldsymbol X)}\right\}^2\cdot\left\{D-\frac{e^{\boldsymbol \theta^{*'}\boldsymbol Z}}{(1+e^{\boldsymbol \theta^{*'}\boldsymbol Z})}\right\}^2\cdot \boldsymbol Z\boldsymbol Z'\right]\cdot
\end{align*}
Combining together the consistency of $\widehat{G_{\boldsymbol \theta}}$ and $\widehat{E}$ we obtain
$$\widehat{V}=\widehat{G_{\boldsymbol \theta}}^{-1}\cdot \widehat{E}\cdot(\widehat{G_{\boldsymbol \theta}}^{-1})^{'}\xrightarrow{p}V=(G_{\boldsymbol \theta^*})^{-1}\cdot \mathbb{E}[g(\boldsymbol \theta^*)\cdot \boldsymbol g(\boldsymbol \theta^*)'].(G_{\boldsymbol \theta^*}^{-1})^{'}\cdot$$
\begin{align*}
    \text{From Theorem \ref{thm2}}\hspace{0.2cm}&\sqrt{N}(\widehat{\boldsymbol \theta}-\boldsymbol \theta^*)\xrightarrow{d}\mathcal{N}\{0,(G_{\boldsymbol \theta^*})^{-1}\cdot\mathbb{E}[g(\boldsymbol \theta^*)\cdot \boldsymbol g(\boldsymbol \theta^*)']\cdot (G_{\boldsymbol \theta^*}^{-1})^{'}\}
\end{align*}
Therefore we obtain
\begin{align*}
    & \text{Var}\{ \sqrt{N}(\widehat{\boldsymbol \theta}-\boldsymbol \theta^*)\}=(G_{\boldsymbol \theta^*})^{-1}\cdot\mathbb{E}[g(\boldsymbol \theta^*)\cdot \boldsymbol g(\boldsymbol \theta^*)']\cdot(G_{\boldsymbol \theta^*}^{-1})^{'} + O_p(1)\\
    & N \cdot\text{Var}(\widehat{\boldsymbol \theta})=(G_{\boldsymbol \theta^*})^{-1}\cdot\mathbb{E}[g(\boldsymbol \theta^*)\cdot \boldsymbol g(\boldsymbol \theta^*)']\cdot(G_{\boldsymbol \theta^*}^{-1})^{'} + O_p(1)\\
     &\text{Var} (\widehat{\boldsymbol \theta})=\frac{1}{N}\cdot(G_{\boldsymbol \theta^*})^{-1}\cdot\mathbb{E}[g(\boldsymbol \theta^*)\cdot \boldsymbol g(\boldsymbol \theta^*)']\cdot(G_{\boldsymbol \theta^*}^{-1})^{'} + O_p\left(\frac{1}{N}\right)\\
     & \text{Var}(\widehat{\boldsymbol \theta})= \frac{1}{N}\cdot\widehat{G_{\boldsymbol \theta}}^{-1}\cdot\widehat{E}\cdot(\widehat{G_{\boldsymbol \theta}}^{-1})^{'}+ o_p\left(1\right)=\frac{\widehat{V}}{N}+ o_p\left(1\right)\cdot
\end{align*} 
Therefore $\frac{\widehat{V}}{N}$ is a consistent estimator of asymptotic variance of $\widehat{\boldsymbol \theta}$.
\end{proof}

\subsection{Asymptotic Distribution and Variance Estimation of PL}\label{sec:proofPL}
For the following theorem apart from assumptions \textit{A}1 and \textit{A}2 we make the following assumptions\\

\noindent
\textit{A}5.  The selection model parameter $\boldsymbol \alpha \in \Theta_{\alpha}$ where $\Theta_{\alpha}$ is compact.\\
\noindent
\textit{A}6.  All the variables, including the selection indicators $S,S_{\text{ext}}$, $\boldsymbol X=(D,\boldsymbol Z_2,\boldsymbol W$), $\boldsymbol Z_1$ are considered to be random which are generally considered fixed in finite population literature.\\
\noindent
\textit{A}7.  The known external selection probability is dependent on some variables $\boldsymbol M$ and external selection indicators $S_{\text{ext}}$ are independent Bernoulli random variables where $P(S_{\text{ext}}=1|\boldsymbol M)=\pi_{\text{ext}}(\boldsymbol M)$. $\boldsymbol M$  can be set of any variables even overlap with $\boldsymbol X,\boldsymbol Z_1$.\\

\begin{theorem}
     Under assumptions \textit{A}1, \textit{A}2 \textit{A}5, \textit{A}6 and \textit{A}7  and assuming the selection model is correctly specified, that is, $\pi(\boldsymbol X,\boldsymbol \alpha^*)=P(S=1|\boldsymbol X)=\pi(\boldsymbol X)$ where $\boldsymbol \alpha^*$ is the true value of $\boldsymbol\alpha$, then $\widehat{\boldsymbol \theta}$ estimated using PL is consistent for $\boldsymbol \theta^*$ as $N\rightarrow \infty$. \label{thm4}
\end{theorem}
\begin{proof}
    In this case, the estimating equation consists of both selection model and disease model parameter estimation. The two step estimating equation is given by
\begin{align*}
     & \delta_n(\boldsymbol \alpha)=\frac{1}{N}\sum_{i=1}^N S_i\boldsymbol X_i- \frac{1}{N}\cdot \sum_{i=1}^ N \left(\frac{S_{\text{ext},i}}{\pi_{\text{ext},i}}\right)\cdot \pi(\boldsymbol X_i,\boldsymbol\alpha)\cdot \boldsymbol X_i\\
     & \phi_n(\boldsymbol \theta,\widehat{\boldsymbol \alpha})= \frac{1}{N}\sum_{i=1}^{N}\frac{S_i}{\pi(\boldsymbol X_i,\widehat{\boldsymbol \alpha})}\left\{D_i \boldsymbol Z_i-\frac{e^{\boldsymbol\theta'\boldsymbol Z_i}}{(1+e^{\boldsymbol\theta'\boldsymbol Z_i})}\cdot \boldsymbol Z_i\right\}\cdot
\end{align*}
From \citet{tsiatis2006semiparametric} to show  $\widehat{\boldsymbol\theta}\xrightarrow{p}\boldsymbol\theta^*$ in a two step estimation procedure, we need to prove both $\mathbb{E}(\delta_N(\boldsymbol \alpha^*))=\mathbf{0}$ and $\mathbb{E}(\phi_N(\boldsymbol \theta^*,\boldsymbol \alpha^*))=\mathbf{0}$. First we show that $\mathbb{E}(\delta_N(\boldsymbol \alpha^*))=\mathbf{0}$.
\begin{align*}
    \mathbb{E}\left(\frac{1}{N}\sum_{i=1}^N S_i \boldsymbol X_i\right)&= \frac{1}{N}\sum_{i=1}^N \mathbb{E}(S_i\boldsymbol X_i)=\frac{1}{N}\sum_{i=1}^N \mathbb{E}_{\boldsymbol X_i,\boldsymbol M_i}\{\mathbb{E}(S_i\boldsymbol X_i|\boldsymbol X_i,\boldsymbol M_i)\}
\end{align*}
Since $S_i$ is independent of $\boldsymbol M_i$ given $\boldsymbol X_i$  we obtain
\begin{align}
    &\frac{1}{N}\sum_{i=1}^N \mathbb{E}_{\boldsymbol X_i,\boldsymbol M_i}\{\mathbb{E}(S_i\boldsymbol X_i|\boldsymbol X_i,\boldsymbol M_i)\}=\frac{1}{N}\sum_{i=1}^N \mathbb{E}_{\boldsymbol X_i,\boldsymbol M_i}(\boldsymbol X_i\cdot \mathbb{E}(S_i|\boldsymbol X_i))\\
    &= \frac{1}{N}\sum_{i=1}^N \mathbb{E}_{\boldsymbol X_i,\boldsymbol M_i}\{\boldsymbol X_i \cdot \pi(\boldsymbol X_i,\boldsymbol \alpha^*)\}\cdot\label{eq:st21}
\end{align}
Next we have
\begin{align}
    &\mathbb{E}\left\{\frac{1}{N}\sum_{i=1}^ N \frac{S_{\text{ext},i}}{\pi_{\text{ext},i}} \cdot \pi(\boldsymbol X_i,\boldsymbol\alpha^*) \cdot \boldsymbol X_i\right\}=
    \frac{1}{N}\sum_{i=1}^ N \mathbb{E}_{\boldsymbol X_i,\boldsymbol M_i}\left[\mathbb{E}\left\{\left.\frac{S_{\text{ext},i}} {\pi_{\text{ext}}(\boldsymbol M_i)}\cdot \pi(\boldsymbol X_i,\boldsymbol \alpha_0)\cdot \boldsymbol X_i \right|\boldsymbol X_i,\boldsymbol M_i\right\}\right]\\
    &=  \frac{1}{N}\sum_{i=1}^ N \mathbb{E}_{\boldsymbol X_i,\boldsymbol M_i}\left[\mathbb{E}\left\{\left.\frac{S_{\text{ext},i}} {\pi_{\text{ext}}(\boldsymbol M_i)}\cdot \pi(\boldsymbol X_i,\boldsymbol \alpha^*)\cdot \boldsymbol X_i \right|\boldsymbol X_i,\boldsymbol M_i\right\}\right]\\
    &=\frac{1}{N}\sum_{i=1}^ N \mathbb{E}_{\boldsymbol X_i,\boldsymbol M_i}\left[\pi(\boldsymbol X_i,\boldsymbol \alpha^*)\cdot \boldsymbol X_i \cdot \mathbb{E}\left\{\left.\frac{S_{\text{ext},i}} {\pi_{\text{ext}}(\boldsymbol M_i)}\right|\boldsymbol X_i,\boldsymbol M_i\right\}\right]\\
    & = \frac{1}{N}\sum_{i=1}^ N \mathbb{E}_{\boldsymbol X_i,\boldsymbol M_i}\{\boldsymbol X_i\cdot \pi(\boldsymbol X_i,\boldsymbol \alpha^*)\}\cdot\label{eq:st22}
\end{align}
Using equations \eqref{eq:st21} and \eqref{eq:st22} we obtain 
\begin{equation}
    \mathbb{E}(\delta_N(\boldsymbol \alpha^{*}))=\mathbf{0}\cdot \label{eq:st23}
\end{equation}
The proof to show that $\mathbb{E}(\phi_N(\boldsymbol \theta^*,\boldsymbol \alpha^*))=\mathbf{0}$ is exactly as Theorem \ref{thm1} except we use $\pi(\boldsymbol X_i,\boldsymbol \alpha^*)$ instead of $\pi(\boldsymbol X_i)$. Therefore we obtain $\widehat{\boldsymbol \theta}$ is a consistent estimator of $\boldsymbol \theta^*$.
\end{proof}
\begin{theorem}
    Under assumptions of Theorem \ref{thm4} and consistency of $\widehat{\boldsymbol \theta}$, the asymptotic distribution of $\widehat{\boldsymbol \theta}$ using PL is given by%,
\begin{equation}
\sqrt{N}(\widehat{\boldsymbol \theta}-\boldsymbol\theta^*)\xrightarrow{d}\mathcal{N}(\mathbf{0},V).
\end{equation}
where
\begin{align*}
    & V=(G_{\boldsymbol \theta^*})^{-1}.\mathbb{E}[\{\boldsymbol g(\boldsymbol \theta^*,\boldsymbol \alpha^*)+G_{\boldsymbol \alpha^*}.\boldsymbol \Psi(\boldsymbol \alpha^*)\}\{\boldsymbol g(\boldsymbol \theta^*,\boldsymbol \alpha^*)+G_{\boldsymbol \alpha^*}.\boldsymbol \Psi(\boldsymbol \alpha^*)\}'].(G_{\boldsymbol \theta^*}^{-1})^{'}\\
    & G_{\boldsymbol \theta^*} = \mathbb{E}\left\{-\frac{S}{\pi(\boldsymbol X,\boldsymbol \alpha^*)}\cdot \frac{e^{\boldsymbol \theta^{*'}\boldsymbol Z}}{(1+e^{\boldsymbol \theta^{*'}\boldsymbol Z})^2}\cdot \boldsymbol Z\boldsymbol Z'\right\}\\
    & \boldsymbol g(\boldsymbol \theta^*,\boldsymbol \alpha^*) = \frac{S}{\pi(\boldsymbol X,\boldsymbol \alpha^*)}\left\{D \boldsymbol Z-\frac{e^{\boldsymbol \theta^{*'}\boldsymbol Z}}{(1+e^{\boldsymbol \theta^{*'}\boldsymbol Z})}\cdot \boldsymbol Z\right\}\\
    & G_{\boldsymbol \alpha^*}=\mathbb{E}\left[-\frac{S}{\pi(\boldsymbol X,\boldsymbol \alpha^*)}\cdot \{1-\pi(\boldsymbol X,\boldsymbol \alpha^*)\}\left\{D-\frac{e^{\boldsymbol \theta^{*'}\boldsymbol Z}}{(1+e^{\boldsymbol \theta^{*'}\boldsymbol Z})}\right\}\cdot \boldsymbol Z\boldsymbol X'\right]\\
    & \boldsymbol \Psi(\boldsymbol \alpha^*)= \mathbb{E}\left[\frac{\pi(\boldsymbol X,\boldsymbol \alpha^*)}{\pi_{\text{ext}}(\boldsymbol M)}\cdot\{1-\pi(\boldsymbol X,\boldsymbol \alpha^*)\}\cdot S_{\text{ext}} \cdot\boldsymbol X\boldsymbol X'\right]^{-1}.\left\{S\boldsymbol X-\frac{\pi(\boldsymbol X,\boldsymbol \alpha^*)}{\pi_{\text{ext}}(\boldsymbol M)} \cdot S_{\text{ext}}\cdot \boldsymbol X\right\}.
\end{align*} \label{thm5}
\end{theorem}
\begin{proof}
Let $\boldsymbol h(\boldsymbol \alpha^*)=\left\{S\boldsymbol X-S_{\text{ext}}\cdot \frac{\pi(\boldsymbol X,\boldsymbol \alpha^*)}{\pi_{\text{ext}}(\boldsymbol M)}\boldsymbol \cdot \boldsymbol X\right\}$. By \citet{tsiatis2006semiparametric}'s arguments on a two step Z-estimation problem, under assumptions of Theorem \ref{thm4} and consistency of $\widehat{\boldsymbol \theta}$, we obtain that 
$$\sqrt{N}(\widehat{\boldsymbol \theta}-\boldsymbol\theta^*)\xrightarrow{d}\mathcal{N}(0,(G_{\boldsymbol \theta^*})^{-1}.\mathbb{E}[\{\boldsymbol g(\boldsymbol \theta^*,\boldsymbol \alpha^*)+G_{\boldsymbol \alpha^*}.\boldsymbol \Psi(\boldsymbol \alpha^*)\}\{\boldsymbol g(\boldsymbol \theta^*,\boldsymbol \alpha^*)+G_{\boldsymbol \alpha^*}.\boldsymbol \Psi(\boldsymbol \alpha^*)\}'].(G_{\boldsymbol \theta^*}^{-1})^{'}\cdot$$
We derive the expression of each of the terms in the above expression.
\begin{align*}
    & G_{\boldsymbol \theta^*}=\left.\mathbb{E}\left\{\frac{\partial g(\boldsymbol \theta,\boldsymbol \alpha^*)}{\partial \boldsymbol \theta}\right\}\right|_{\boldsymbol\theta=\boldsymbol \theta^*}\hspace{1cm} G_{\boldsymbol \alpha}=\left.G_{\boldsymbol \alpha^*}=\mathbb{E}\left\{\frac{\partial g(\boldsymbol \theta^*,\boldsymbol \alpha)}{\partial \boldsymbol \alpha}\right\}\right|_{\boldsymbol\alpha=\boldsymbol \alpha^*}\\
    & H = \left.\mathbb{E}\left\{\frac{\partial \boldsymbol h(\boldsymbol \alpha)}{\partial \boldsymbol \alpha}\right\}\right|_{\boldsymbol\alpha=\boldsymbol \alpha^*}\hspace{1.9cm} \boldsymbol \Psi(\boldsymbol \alpha^*)=-H^{-1}\boldsymbol h(\boldsymbol \alpha^*) \cdot
\end{align*}
First we calculate $G_{\boldsymbol \theta^*}$.
\subsubsection*{Calculation for $G_{\boldsymbol \theta^*}$}
\begin{align*}
    & \left.\frac{\partial \boldsymbol g(\boldsymbol \theta,\boldsymbol \alpha^*)}{\partial \boldsymbol \theta}\right|_{\boldsymbol\theta=\boldsymbol \theta^*}= \left.\frac{\partial}{\partial \boldsymbol \theta}\left[\frac{S}{\pi(\boldsymbol X,\boldsymbol \alpha^*)}\left\{D\boldsymbol Z-\frac{e^{\boldsymbol \theta^{*'}\boldsymbol Z}}{(1+e^{\boldsymbol \theta^{*'}\boldsymbol Z})}\cdot \boldsymbol Z\right\}\right]\right|_{\boldsymbol\theta=\boldsymbol \theta^*}\\
    & = -\frac{S}{\pi(\boldsymbol X,\boldsymbol \alpha^*)} \cdot \frac{e^{\boldsymbol \theta^{*'}\boldsymbol Z}}{(1+e^{\boldsymbol \theta^{*'}\boldsymbol Z})^2}\cdot \boldsymbol Z\boldsymbol Z'\cdot
\end{align*}
Therefore we obtain 
\begin{align*}
    G_{\boldsymbol \theta^*} = \mathbb{E}\left[-\frac{S}{\pi(\boldsymbol X,\boldsymbol \alpha^*)} \cdot \frac{e^{\boldsymbol \theta^{*'}\boldsymbol Z}}{(1+e^{\boldsymbol \theta^{*'}\boldsymbol Z})^2}\cdot \boldsymbol Z \boldsymbol Z'\right]\cdot
\end{align*}
Next we calculate $G_{\boldsymbol \alpha^*}$.
\subsubsection*{Calculation for $G_{\boldsymbol \alpha^*}$}
\begin{align*}
    & \left.\frac{\partial \boldsymbol g(\boldsymbol \theta^*,\boldsymbol \alpha)}{\partial \boldsymbol \alpha}\right|_{\boldsymbol\alpha=\boldsymbol \alpha^*}= \left.\frac{\partial}{\partial \boldsymbol \alpha}\left[\frac{S}{\pi(\boldsymbol X,\boldsymbol \alpha)}\left\{D\boldsymbol Z-\frac{e^{\boldsymbol \theta^{*'}\boldsymbol Z}}{(1+e^{\boldsymbol \theta^{*'}\boldsymbol Z})}\cdot \boldsymbol Z\right\}\right]\right|_{\boldsymbol\alpha=\boldsymbol \alpha^*}\\
    & = -\frac{S}{\pi(\boldsymbol X,\boldsymbol \alpha^*)^2}\cdot\left\{D\boldsymbol Z-\frac{e^{\boldsymbol \theta^{*'}\boldsymbol Z}}{(1+e^{\boldsymbol \theta^{*'}\boldsymbol Z})}\cdot \boldsymbol Z\right\}\cdot \left.\frac{\partial \pi(\boldsymbol X,\boldsymbol \alpha)}{\partial \boldsymbol \alpha}\right|_{\boldsymbol\alpha=\boldsymbol \alpha^*}\\
    & =-\frac{S}{\pi(\boldsymbol X,\boldsymbol \alpha^*)}\{1-\pi(\boldsymbol X,\boldsymbol \alpha^*)\}\left\{D-\frac{e^{\boldsymbol \theta^{*'}\boldsymbol Z}}{(1+e^{\boldsymbol \theta^{*'}\boldsymbol Z})}\right\}\boldsymbol Z\boldsymbol X'\cdot
\end{align*}
Therefore we obtain 
\begin{align*}
    & G_{\boldsymbol \alpha^*} = \mathbb{E}\left[-\frac{S}{\pi(\boldsymbol X,\boldsymbol \alpha^*)}\cdot\{1-\pi(\boldsymbol X,\boldsymbol \alpha^*)\}\left\{D-\frac{e^{\boldsymbol \theta^{*'}\boldsymbol Z}}{(1+e^{\boldsymbol \theta^{*'}\boldsymbol Z})}\right\}\boldsymbol Z\boldsymbol X'\right] \cdot
\end{align*}
Next we calculate $\Psi(\boldsymbol \alpha^*)$.
\subsubsection*{Calculation for $\boldsymbol\Psi(\boldsymbol \alpha^*)$}
\begin{align*}
   & \left.\frac{ \partial \boldsymbol h(\boldsymbol \alpha)}{\partial \boldsymbol \alpha}\right|_{\boldsymbol \alpha=\boldsymbol \alpha^*}= \left.\frac{ \partial}{\partial \boldsymbol \alpha}\left\{S\boldsymbol X-S_{\text{ext}}\cdot \frac{\pi(\boldsymbol X,\boldsymbol \alpha)}{\pi_{\text{ext}}(\boldsymbol M)}\cdot \boldsymbol X \right\}\right|_{\boldsymbol \alpha=\boldsymbol \alpha^*}=-\frac{\left.\frac{ \partial}{\partial \alpha}\pi(\boldsymbol X,\boldsymbol \alpha)\right|_{\boldsymbol \alpha=\boldsymbol \alpha^*}}{\pi_{\text{ext}}(\boldsymbol M)}\cdot S_{\text{ext}}\cdot \boldsymbol X\\
   & = -\frac{\pi(\boldsymbol X,\boldsymbol \alpha^*)}{\pi_{\text{ext}}(\boldsymbol M)}\cdot \{1-\pi(\boldsymbol X,\boldsymbol \alpha^*)\}\cdot S_{\text{ext}}\cdot \boldsymbol X \cdot \boldsymbol X'
\end{align*}
This implies
\begin{align*}
  H = \mathbb{E}\left[-\frac{\pi(\boldsymbol X,\boldsymbol \alpha^*)}{\pi_{\text{ext}}(\boldsymbol M)}\cdot \{1-\pi(\boldsymbol X,\boldsymbol \alpha^*)\}\cdot S_{\text{ext}}\cdot \boldsymbol X \cdot \boldsymbol X'\right]\cdot
\end{align*}
Therefore we obtain 
\begin{align*}
    & \boldsymbol \Psi(\boldsymbol \alpha^*)= \mathbb{E}\left[\frac{\pi(\boldsymbol X,\boldsymbol \alpha^*)}{\pi_{\text{ext}}(\boldsymbol M)}\cdot \{1-\pi(\boldsymbol X,\boldsymbol \alpha^*)\}\cdot S_{\text{ext}}\cdot \boldsymbol X \boldsymbol X'\right]^{-1}\cdot \left\{S\boldsymbol X-\frac{\pi(\boldsymbol X,\boldsymbol \alpha^*)}{\pi_{\text{ext}}(\boldsymbol M)}\cdot S_{\text{ext}}\cdot \boldsymbol X\right\}\cdot
\end{align*}
This gives the asymptotic distribution of $\widehat{\boldsymbol \theta}$ for PL.
\end{proof}
\noindent
Next we derive a consistent estimator of asymptotic variance of $\widehat{\boldsymbol\theta}$. Apart from assumptions of Theorem \ref{thm4} and \ref{thm5} we make the following additional assumption.\\% to derive a consistent estimator of asymptotic variance of $\widehat{\boldsymbol\theta}$.\\

\noindent
Let $\eta=(\boldsymbol \theta,\boldsymbol \alpha)$ and $\mathcal{N}=\Theta \times \Theta_{\boldsymbol \alpha}$.\\

\noindent
\textit{A}8. Each of the following expectations are finite.
\begin{align*}
    & \mathbb{E}[\sup_{\boldsymbol \eta \in \mathcal{N}}(G_{\boldsymbol \theta})]<\infty \hspace{0.8cm} 
    \mathbb{E}[\sup_{\boldsymbol \eta \in \mathcal{N}}(G_{\boldsymbol \alpha})]<\infty \hspace{0.8cm} \mathbb{E}[\sup_{\boldsymbol \alpha \in \Theta_{\boldsymbol \alpha}}(H)]<\infty\\
    & \mathbb{E}\left(\sup_{\boldsymbol \eta \in \mathcal{N}}\left[S \cdot \left\{\frac{1}{\pi(\boldsymbol X,\boldsymbol \alpha)}\right\}^2\left\{D-\frac{e^{\boldsymbol \theta' \boldsymbol Z}}{(1+e^{\boldsymbol \theta'\boldsymbol Z})}\right\}^2\cdot \boldsymbol Z \boldsymbol Z'\right]\right)<\infty\\
    & \mathbb{E}\left(\sup_{\boldsymbol \eta \in \mathcal{N}}\left[\frac{S\boldsymbol X}{\pi(\boldsymbol X,\boldsymbol \alpha)}\cdot \left\{D\boldsymbol Z'-\frac{e^{\boldsymbol \theta'\boldsymbol Z}}{(1+e^{\boldsymbol \theta' \boldsymbol Z})} \cdot \boldsymbol Z'\right\}-\frac{SS_{\text{ext}}\boldsymbol X}{\pi_{\text{ext}}(\boldsymbol M)} \cdot \left\{D \boldsymbol Z'-\frac{e^{\boldsymbol \theta' \boldsymbol Z}}{(1+e^{\boldsymbol \theta' \boldsymbol Z})} \cdot \boldsymbol Z'\right\}\right]\right)<\infty\\
    & \mathbb{E}\left(\sup_{\boldsymbol \alpha \in \Theta_{\boldsymbol \alpha}}\left[S\cdot\boldsymbol X \boldsymbol X'-2 SS_{\text{ext}} \cdot \frac{\pi(\boldsymbol X,\boldsymbol \alpha)}{\pi_{\text{ext}}(\boldsymbol M)}\cdot\boldsymbol X\boldsymbol X'+  S_{\text{ext}} \cdot \left\{\frac{\pi(\boldsymbol X,\boldsymbol \alpha)}{\pi_{\text{ext}}(\boldsymbol M)}\right\}^2 \cdot \boldsymbol X \boldsymbol X'\right]\right) <\infty \cdot
\end{align*}
\begin{theorem}
    Under all the assumptions of Theorems \ref{thm4}, \ref{thm5} and \textit{A}8, $\frac{1}{N} \cdot \widehat{G_{\boldsymbol \theta}}^{-1}\cdot \widehat{E}\cdot (\widehat{G_{\boldsymbol \theta}}^{-1})^{'}$ is a consistent estimator of the variance of $\widehat{\boldsymbol \theta}$ for PL where
\begin{align*}
   & \widehat{G_{\boldsymbol \theta}}=-\frac{1}{N}\sum_{i=1}^N S_i \cdot \frac{1}{\pi(\boldsymbol X_i,\widehat{\boldsymbol \alpha})}\cdot \frac{e^{\widehat{\boldsymbol \theta}'\boldsymbol Z_i}}{(1+e^{\widehat{\boldsymbol \theta}'\boldsymbol Z_i})^2}\cdot\boldsymbol Z_i \boldsymbol Z_i'\\
   & \widehat{H} = -\frac{1}{N}\sum_{i=1}^N S_{\text{ext},i}\cdot \frac{\pi(\boldsymbol X_i,\widehat{\boldsymbol \alpha})}{\pi_{\text{ext}}(\boldsymbol M)} \cdot \{1-\pi(\boldsymbol X_i,\hat{\boldsymbol \alpha})\}\cdot \boldsymbol X_i \cdot \boldsymbol X_i'\\
   & \widehat{G_{\boldsymbol \alpha}}=-\frac{1}{N}\sum_{i=1}^N S_i \cdot \frac{\{1-\pi(\boldsymbol X_i,\widehat{\boldsymbol \alpha})\}}{\pi(\boldsymbol X_i,\widehat{\boldsymbol \alpha})}\cdot\left\{D_i-\frac{e^{\widehat{\boldsymbol \theta}'\boldsymbol Z_i}}{(1+e^{\widehat{\boldsymbol \theta}'\boldsymbol Z_i})}\right\}\boldsymbol Z_i \boldsymbol X_i'\\
   & \widehat{E}_1 =\frac{1}{N}\sum_{i =1}^N S_i \cdot \left\{\frac{1}{\pi(\boldsymbol X_i,\widehat{\boldsymbol \alpha})}\right\}^2\left\{D_i-\frac{e^{\widehat{\boldsymbol \theta}'\boldsymbol Z_i}}{(1+e^{\hat{\theta}' \boldsymbol Z_i})}\right\}^2\cdot \boldsymbol Z_i \boldsymbol Z_i'
\end{align*}
\begin{align*}
    & \widehat{E}_2 =  \widehat{E}_3'=\frac{1}{N}\sum_{i=1}^N S_i \cdot \widehat{G_{\boldsymbol \alpha}}\cdot \widehat{H}^{-1} \cdot \boldsymbol X_i \cdot \frac{1}{\pi(\boldsymbol X_i,\widehat{\boldsymbol \alpha})}\cdot \left\{D_i\boldsymbol Z_i'-\frac{e^{\widehat{\boldsymbol \theta}'\boldsymbol Z_i}}{(1+e^{\widehat{\boldsymbol \theta}' \boldsymbol Z_i})} \cdot \boldsymbol Z_i'\right\}\\
    & -\frac{1}{N}\sum_{i=1}S_i\cdot S_{\text{ext},i} \cdot \widehat{G_{\boldsymbol \alpha}} \cdot \widehat{H}^{-1} \cdot \frac{1}{\pi_{\text{ext}}(\boldsymbol M_i)} \cdot \boldsymbol X_i \cdot \left\{D_i \boldsymbol Z_i'-\frac{e^{\widehat{\boldsymbol \theta}' \boldsymbol Z_i}}{(1+e^{\widehat{\boldsymbol \theta}' \boldsymbol Z_i})} \cdot \boldsymbol Z_i'\right\}\\
    & \widehat{E_4} = \frac{1}{N} \cdot \widehat{G_{\boldsymbol \alpha}} \cdot \widehat{H}^{-1}\left[\sum_{i=1}^N S_i\cdot \boldsymbol X_i \cdot \boldsymbol X_i'-2 \cdot \sum_{i=1}^ N S_i\cdot S_{\text{ext},i} \cdot \frac{\pi(\boldsymbol X_i,\widehat{\boldsymbol \alpha})}{\pi_{\text{ext}}(\boldsymbol M_i)}\boldsymbol X_i \cdot \boldsymbol X_i'+ \right.\\
    & \left. \sum_{i=1}^N S_{\text{ext},i} \cdot \left\{\frac{\pi(\boldsymbol X_i,\widehat{\boldsymbol \alpha})}{\pi_{\text{ext}}(\boldsymbol M)}\right\}^2 \cdot \boldsymbol X_i \cdot \boldsymbol X_i'\right]\cdot(\widehat{H}^{-1})'\cdot(\widehat{G_{\boldsymbol \alpha}})' \\
    & \widehat{E}=\widehat{E}_1-\widehat{E}_2-\widehat{E}_3 + \widehat{E}_4 \cdot
\end{align*}
\end{theorem}
\begin{proof}
    Under all the assumptions of Theorems \ref{thm4}, \ref{thm5} and \textit{A}8 and using ULLN and Continuous Mapping Theorem, the proof of consistency for each of the following sample quantities are exactly same as the approach in Theorem \ref{thm3}. Using the exact same steps on the joint parameters $\boldsymbol \eta$ instead of $\boldsymbol \theta$ (as in Theorem \ref{thm3}) we obtain
\begin{align*}
     & \widehat{G_{\boldsymbol \theta}}=-\frac{1}{N}\sum_{i=1}^N S_i \cdot \frac{1}{\pi(\boldsymbol X_i,\hat{\boldsymbol \alpha})}\cdot \frac{e^{\widehat{\boldsymbol \theta}'\boldsymbol Z_i}}{(1+e^{\widehat{\boldsymbol \theta}'\boldsymbol Z_i})^2}\cdot\boldsymbol Z_i \boldsymbol Z_i' \\
     & \xrightarrow{p} G_{\boldsymbol \theta^*} = \mathbb{E}\left\{-\frac{S}{\pi(\boldsymbol X,\boldsymbol\alpha^*)}\cdot \frac{e^{\boldsymbol \theta^{*'}\boldsymbol Z}}{(1+e^{\boldsymbol \theta^{*'}\boldsymbol Z})^2}.\boldsymbol Z\boldsymbol Z'\right\}\cdot
\end{align*}
Similarly we obtain 
\begin{align*}
     & \widehat{G_{\boldsymbol \alpha}}=-\frac{1}{N}\sum_{i=1}^N S_i \cdot \frac{\{1-\pi(\boldsymbol X_i,\widehat{\boldsymbol \alpha})\}}{\pi(\boldsymbol X_i,\widehat{\boldsymbol \alpha})}\cdot\left\{D_i-\frac{e^{\widehat{\boldsymbol \theta}' \boldsymbol Z_i}}{(1+e^{\widehat{\boldsymbol \theta}'\boldsymbol Z_i})}\right\}\boldsymbol Z_i \boldsymbol X_i' \\
     & \xrightarrow{p} G_{\boldsymbol \alpha} = \mathbb{E}\left[-\frac{S}{\pi(\boldsymbol X,\boldsymbol \alpha^*)}\{1-\pi(\boldsymbol X,\boldsymbol \alpha^*)\}\left\{D-\frac{e^{\boldsymbol \theta^{*'}\boldsymbol Z}}{(1+e^{\boldsymbol \theta^{*'}\boldsymbol Z})}\right\}\boldsymbol Z\boldsymbol X'\right]\cdot
\end{align*}
\begin{align*}
     & \widehat{H} = -\frac{1}{N}\sum_{i=1}^N S_{\text{ext},i}\cdot \frac{\pi(\boldsymbol X_i,\hat{\boldsymbol \alpha})}{\pi_{\text{ext}}(\boldsymbol M_i)} \cdot \{1-\pi(\boldsymbol X_i,\hat{\boldsymbol \alpha})\}\cdot \boldsymbol X_i \cdot \boldsymbol X_i'\\
     &\xrightarrow{p} H = \mathbb{E}\left[-\frac{\pi(\boldsymbol X,\boldsymbol \alpha^*)}{\pi_e(\boldsymbol M)}\cdot \{1-\pi(\boldsymbol X,\boldsymbol \alpha^*)\}\cdot S_{\text{ext}} \cdot \boldsymbol X \cdot \boldsymbol X'\right]\cdot
\end{align*}
\begin{align*}
     &  \widehat{E}_1 =\frac{1}{N}\sum_{i =1}^N S_i \cdot \left\{\frac{1}{\pi(\boldsymbol X_i,\hat{\boldsymbol \alpha})}\right\}^2\left\{D_i-\frac{e^{\widehat{\boldsymbol \theta}' \boldsymbol Z_i}}{(1+e^{\hat{\theta}'\boldsymbol Z_i})}\right\}^2\cdot \boldsymbol Z_i \boldsymbol Z_i'\\
     & \xrightarrow{p}  \mathbb{E}\left[S \cdot \left\{\frac{1}{\pi(\boldsymbol X,\boldsymbol \alpha^*)}\right\}^2\cdot \left\{D-\frac{e^{\boldsymbol \theta^{*'}\boldsymbol Z_i}}{(1+e^{\boldsymbol \theta^{*'} \boldsymbol Z})}\right\}^2\cdot \boldsymbol Z\boldsymbol Z'\right]\cdot
\end{align*}
\begin{align*}
     &  \widehat{E}_2 =  \widehat{E}_3'=\frac{1}{N}\sum_{i=1}^N S_i \cdot \widehat{G_{\boldsymbol \alpha}}\cdot \widehat{H}^{-1} \cdot \boldsymbol X_i \cdot \frac{1}{\pi(\boldsymbol X_i,\hat{\boldsymbol \alpha})}\cdot \left\{D_i\boldsymbol Z_i'-\frac{e^{\widehat{\boldsymbol \theta}'\boldsymbol Z_i}}{(1+e^{\widehat{\boldsymbol \theta}'  \boldsymbol Z_i})} \cdot \boldsymbol Z_i'\right\}\\
    & -\frac{1}{N}\sum_{i=1}S_i\cdot S_{\text{ext},i} \cdot \widehat{G_{\boldsymbol \alpha}} \cdot \widehat{H}^{-1} \cdot \frac{1}{\pi_{\text{ext}}(\boldsymbol M_i)} \cdot \boldsymbol X_i \cdot \left\{D_i \boldsymbol Z_i'-\frac{e^{\widehat{\boldsymbol \theta}'  \boldsymbol Z_i}}{(1+e^{\widehat{\boldsymbol \theta}'  \boldsymbol Z_i})} \cdot \boldsymbol Z_i'\right\}\\
    &\xrightarrow{p} \mathbb{E}\{G_{\boldsymbol \alpha^*}\cdot H^{-1}\cdot \boldsymbol h(\boldsymbol \alpha^*)\cdot  \boldsymbol g(\boldsymbol \theta^*,\boldsymbol \alpha^*)'\}\cdot
\end{align*}
\begin{align*}
    & \widehat{E_4} = \frac{1}{N} \cdot \widehat{G_{\boldsymbol \alpha}} \cdot \widehat{H}^{-1}\left[\sum_{i=1}^N S_i\cdot \boldsymbol X_i \cdot \boldsymbol X_i'-2 \cdot \sum_{i=1}^ N S_i\cdot S_{\text{ext},i} \cdot \frac{\pi(\boldsymbol X_i,\widehat{\boldsymbol \alpha})}{\pi_{\text{ext}}(\boldsymbol M_i)}\boldsymbol X_i \cdot \boldsymbol X_i'+ \right.\\
    & \left. \sum_{i=1}^N S_{\text{ext},i} \cdot \left\{\frac{\pi(\boldsymbol X_i,\widehat{\boldsymbol \alpha})}{\pi_{\text{ext}}(\boldsymbol M)}\right\}^2 \cdot \boldsymbol X_i \cdot \boldsymbol X_i'\right]\cdot(\widehat{H}^{-1})'\cdot(\widehat{G_{\boldsymbol \alpha}})' \\
    & \xrightarrow{p} \mathbb{E}\{G_{\boldsymbol \alpha^*}H^{-1}\cdot \boldsymbol h(\boldsymbol \alpha^*)\cdot \boldsymbol h(\boldsymbol \alpha^*)'\cdot (H^{-1})'\cdot (G_{\boldsymbol\alpha^*})'\} 
\end{align*}
Therefore we obtain
\begin{align*}
    \widehat{E}=\widehat{E}_1-\widehat{E}_2-\widehat{E}_3 + \widehat{E}_4\xrightarrow{p}\mathbb{E}[\{\boldsymbol g(\boldsymbol \theta^*,\boldsymbol \alpha^*)+G_{\boldsymbol \alpha^*}\cdot \boldsymbol \Psi(\boldsymbol \alpha^*)\}\{\boldsymbol g(\boldsymbol \theta^*,\boldsymbol \alpha^*)+G_{\boldsymbol \alpha^*}.\boldsymbol \Psi(\boldsymbol \alpha^*)\}']\cdot
\end{align*}
Using all the above results we obtain $\widehat{G_{\boldsymbol \theta}}^{-1}\cdot\widehat{E}\cdot(\widehat{G_{\boldsymbol \theta}}^{-1})^{'}$
\begin{align*}
    \xrightarrow{p}(G_{\boldsymbol \theta^*})^{-1}\cdot \mathbb{E}[\{\boldsymbol g(\boldsymbol \theta^*,\boldsymbol \alpha^*)+G_{\boldsymbol \alpha^*}\cdot \boldsymbol \Psi(\boldsymbol \alpha^*)\}\{\boldsymbol g(\boldsymbol \theta^*,\boldsymbol \alpha^*)+G_{\boldsymbol \alpha^*}.\boldsymbol \Psi(\boldsymbol \alpha^*)\}'].(G_{\boldsymbol \theta^*}^{-1})^{'}\cdot
\end{align*}
From Theorem \ref{thm5} using the same approach used in the last step of Theorem \ref{thm3}, we obtain that $\frac{1}{N}\cdot\widehat{G_{\boldsymbol \theta}}^{-1}\cdot\hat{E}\cdot(\widehat{G_{\boldsymbol \theta}}^{-1})^{'}$ is a consistent estimator of the asymptotic variance of $\widehat{\boldsymbol \theta}$.
\end{proof}
\subsection{Identity of Simplex Regression}\label{sec:simproof}
In this section, we prove the identity in equation \eqref{eq:eq11}. Let $S_{\text{comb}}=(S=1\text{ or }S_{\text{ext}}=1)$.
For any $(j,k)$ in $(0,1),(1,0),(1,1)$, 
\begin{align*}
      & P(S=j,S_{\text{ext}}=k|\boldsymbol X,S_{\text{comb}}=1)=\frac{ P(S=j,S_{\text{ext}}=k,\boldsymbol X|S_{\text{comb}}=1)}{P(\boldsymbol X|S_{\text{comb}}=1)}\\
      & = \frac{P(\boldsymbol X|S=j,S_{\text{ext}}=k,S_{\text{comb}}=1).P(S=j,S_{\text{ext}}=k|S_{\text{comb}}=1)}{\sum_{(a,b)\in (10,01,11)}P(\boldsymbol X|S=a,S_{\text{ext}}=b,S_{\text{comb}}=1).P(S=a,S_{\text{ext}}=b|S_{\text{comb}}=1)}\\
      & = \frac{P(\boldsymbol X|S=j,S_{\text{ext}}=k).P(S=j,S_{\text{ext}}=k)}{\sum_{(a,b)\in (10,01,11)}P(\boldsymbol X|S=a,S_{\text{ext}}=b).P(S=a,S_{\text{ext}}=b)}\cdot
\end{align*}
We define three quantities,
\begin{align*}
    & \mu_{jk}=P(\boldsymbol X|S=j,S_{\text{ext}}=k)\\
    &  \alpha_{jk}=P(S=j,S_{\text{ext}}=k)\\
    & p_{jk}=P(S=j,S_{\text{ext}}=k|\boldsymbol X,S_{\text{comb}}=1)\cdot
\end{align*}
As a result, we obtain that, 
\begin{align*}
    p_{jk}=\frac{\mu_{jk}\alpha_{jk}}{\mu_{11}\alpha_{11}+\mu_{01}\alpha_{01}+\mu_{10}\alpha_{10}}\cdot
\end{align*}
Therefore, we obtain that, 
$$\frac{\mu_{11}\alpha_{11}}{p_{11}}=\frac{\mu_{10}\alpha_{10}}{p_{10}}=\frac{\mu_{01}\alpha_{01}}{p_{01}}\cdot$$
Now, we also obtain that,
\begin{align*}
    P(\boldsymbol X|S=1)& =\sum_b P(\boldsymbol X|S=1,S_{\text{ext}}=b)P(S_{\text{ext}}=b|S=1)\\
    & =\frac{\mu_{11}\alpha_{11}+\mu_{10}\alpha_{10}}{P(S=1)}=\frac{\mu_{11}\alpha_{11}+\mu_{11}\alpha_{11}\frac{p_{10}}{p_{11}}}{P(S=1)}\\
    P(\boldsymbol X|S_{\text{ext}}=1)& =\sum_b P(\boldsymbol X|S=a,S_{\text{ext}}=1)P(S=a|S_{\text{ext}}=1)\\
    & =\frac{\mu_{11}\alpha_{11}+\mu_{10}\alpha_{10}}{P(S_{\text{ext}}=1)}=\frac{\mu_{11}\alpha_{11}+\mu_{11}\alpha_{11}\frac{p_{10}}{p_{11}}}{P(S_{\text{ext}}=1)}\cdot
\end{align*}
By definition of conditional probability 
\begin{align*}
    & P(S=1|\boldsymbol X)=\frac{P(\boldsymbol X|S=1)P(S=1)}{P(\boldsymbol X)}\\
    & P(S_{\text{ext}}=1|\boldsymbol X)=\frac{P(\boldsymbol X|S_{\text{ext}}=1)P(S_{\text{ext}}=1)}{P(\boldsymbol X)}\cdot
\end{align*}
Therefore we obtain 
\begin{align*}
    & P(S=1|\boldsymbol X)=P(S_{\text{ext}}=1|\boldsymbol X)\frac{P(\boldsymbol X|S=1)P(S=1)}{P(\boldsymbol X|S_{\text{ext}}=1)P(S_{\text{ext}}=1)}\cdot
\end{align*}
Using above expressions of $P(\boldsymbol X|S=1)$, $P(\boldsymbol X|S_{\text{ext}}=1)$ and $ P(S=1|\boldsymbol X)$, we obtain that,
$$P(S=1|\boldsymbol X)=P(S_{\text{ext}}=1|\boldsymbol X)\cdot \frac{p_{11}+p_{10}}{p_{11}+p_{01}}\cdot$$
This proves equation \eqref{eq:eq11}.

\subsection{Asymptotic Distribution and Variance Estimation of CL}\label{sec:proofcl}
\begin{theorem}
     Under assumptions \textit{A}1, \textit{A}2 \textit{A}3, \textit{A}5  and assuming the selection model is correctly specified, that is, $\pi(\boldsymbol X,\boldsymbol \alpha^*)=P(S=1|\boldsymbol X)=\pi(\boldsymbol X)$ where $\boldsymbol \alpha^*$ is the true value of $\boldsymbol\alpha$, then $\widehat{\boldsymbol \theta}$ estimated using CL is consistent for $\boldsymbol \theta^*$ as $N\rightarrow \infty$. \label{thm7}
\end{theorem}
\begin{proof}
    In this case, the estimating equation consists of both selection model and disease model parameter estimation. The two step estimating equation is given by 
\begin{align*}
     & \delta_N(\boldsymbol \alpha)=\frac{1}{N}\sum_{i=1}^N\frac{S_i\boldsymbol X_i}{\pi(\boldsymbol X_i,\boldsymbol \alpha)}-\frac{1}{N}\sum_{i=1}^ N\boldsymbol X_i\\
     & \phi_N(\boldsymbol \theta,\widehat{\boldsymbol \alpha})= \frac{1}{N}\sum_{i=1}^{N}\frac{S_i}{\pi(\boldsymbol X_i,\widehat{\boldsymbol \alpha})}\left\{D_i \boldsymbol Z_i-\frac{e^{\boldsymbol\theta'\boldsymbol Z_i}}{(1+e^{\boldsymbol\theta'\boldsymbol Z_i})}\cdot \boldsymbol Z_i\right\}\cdot
\end{align*}
From \citet{tsiatis2006semiparametric} to show  $\widehat{\boldsymbol\theta}\xrightarrow{p}\boldsymbol\theta^*$ in a two step estimation procedure, we need to proof both $\mathbb{E}(\delta_N(\boldsymbol \alpha^*))=\mathbf{0}$ and $\mathbb{E}(\phi_N(\boldsymbol \theta^*,\boldsymbol \alpha^*))=\mathbf{0}$. At first we show that $\mathbb{E}(\delta_N(\boldsymbol \alpha^*))=\mathbf{0}$.
\begin{align*}
    &\mathbb{E}\left\{\frac{1}{N}\sum_{i=1}^ N \frac{S_i\boldsymbol X_i}{\pi(\boldsymbol X_i,\boldsymbol \alpha^*)}\right\}=
    \frac{1}{N}\sum_{i=1}^ N \mathbb{E}\left\{\frac{S_i\boldsymbol X_i}{\pi(\boldsymbol X_i,\boldsymbol \alpha^*)}\right\}\\
    &=\frac{1}{N}\sum_{i=1}^ N \mathbb{E}_{\boldsymbol X_i}\left[\mathbb{E}\left\{\left.\frac{S_i\boldsymbol X_i}{\pi(\boldsymbol X_i,\boldsymbol \alpha^*)} \right|\boldsymbol X_i\right\}\right]\\
    &=\frac{1}{N}\sum_{i=1}^ N \mathbb{E}_{\boldsymbol X_i}\left[\frac{\boldsymbol X_i}{\pi(\boldsymbol X_i,\boldsymbol \alpha^*)}\cdot \mathbb{E}\left(\left. S_i\right|\boldsymbol X_i\right)\right]\\
    &=\frac{1}{N}\sum_{i=1}^ N \mathbb{E}_{\boldsymbol X_i}\left[\frac{\boldsymbol X_i}{\pi(\boldsymbol X_i,\boldsymbol \alpha^*)}\cdot \pi(\boldsymbol X_i,\boldsymbol \alpha^*)\right] = \frac{1}{N}\sum_{i=1}^ N \mathbb{E}_{\boldsymbol X_i}(\boldsymbol X_i)\cdot
\end{align*}
Using this above expression we obtain that 
\begin{align*}
    \mathbb{E}(\delta_N(\boldsymbol \alpha^*))=\mathbb{E}\left\{\frac{1}{N}\sum_{i=1}^ N \frac{S_i\boldsymbol X_i}{\pi(\boldsymbol X_i,\boldsymbol \alpha^*)}-\frac{1}{N}\sum_{i=1}^N\boldsymbol X_i\right\}=\mathbf{0}\cdot
\end{align*}
The proof to show that $\mathbb{E}(\phi_N(\boldsymbol \theta^*,\boldsymbol \alpha^*))=\mathbf{0}$ is exactly as Theorem \ref{thm1} except we use $\pi(\boldsymbol X_i,\boldsymbol \alpha^*)$ instead of $\pi(\boldsymbol X_i)$. Therefore we obtain $\widehat{\boldsymbol \theta}$ is a consistent estimator of $\boldsymbol \theta^*$.
\end{proof}
\begin{theorem}
 Under assumptions of Theorem \ref{thm7} and consistency of $\widehat{\boldsymbol \theta}$, the asymptotic distribution of $\widehat{\boldsymbol \theta}$ using CL is given by%,
\begin{equation}
\sqrt{N}(\widehat{\boldsymbol \theta}-\boldsymbol\theta^*)\xrightarrow{d}\mathcal{N}(\mathbf{0},V). 
\end{equation}
where
\begin{align*}
    & V=(G_{\boldsymbol \theta^*})^{-1}.\mathbb{E}[\{\boldsymbol g(\boldsymbol \theta^*,\boldsymbol \alpha^*)+G_{\boldsymbol \alpha^*}.\boldsymbol \Psi(\boldsymbol \alpha^*)\}\{\boldsymbol g(\boldsymbol \theta^*,\boldsymbol \alpha^*)+G_{\boldsymbol \alpha^*}.\boldsymbol \Psi(\boldsymbol \alpha^*)\}'].(G_{\boldsymbol \theta^*}^{-1})^{'}\\
    & G_{\boldsymbol \theta^*} = \mathbb{E}\left\{-\frac{S}{\pi(\boldsymbol X,\boldsymbol \alpha^*)}\cdot \frac{e^{\boldsymbol \theta^{*'}\boldsymbol Z}}{(1+e^{\boldsymbol \theta^{*'}\boldsymbol Z})^2}\cdot \boldsymbol Z\boldsymbol Z'\right\}\\
    & \boldsymbol g(\boldsymbol \theta^*,\boldsymbol \alpha^*) = \frac{S}{\pi(\boldsymbol X,\boldsymbol \alpha^*)}\left\{D \boldsymbol Z-\frac{e^{\boldsymbol \theta^{*'}\boldsymbol Z}}{(1+e^{\boldsymbol \theta^{*'}\boldsymbol Z})}\cdot \boldsymbol Z\right\}\\
    & G_{\boldsymbol \alpha^*}=\mathbb{E}\left[-\frac{S}{\pi(\boldsymbol X,\boldsymbol \alpha^*)}\cdot \{1-\pi(\boldsymbol X,\boldsymbol \alpha^*)\}\left\{D-\frac{e^{\boldsymbol \theta^{*'}\boldsymbol Z}}{(1+e^{\boldsymbol \theta^{*'}\boldsymbol Z})}\right\}\cdot \boldsymbol Z\boldsymbol X'\right]\\
     & \boldsymbol \Psi(\boldsymbol \alpha^*)= \mathbb{E}\left[\frac{S}{\pi(\boldsymbol X,\boldsymbol \alpha^*)}\cdot\{1-\pi(\boldsymbol X,\boldsymbol \alpha^*)\}\cdot\boldsymbol X\boldsymbol X'\right]^{-1}.\left\{\frac{S \boldsymbol X}{\pi(\boldsymbol X,\boldsymbol \alpha^*)}-\boldsymbol X\right\}.
\end{align*}\label{thm8}
\end{theorem}
\begin{proof}
Let $\boldsymbol h(\boldsymbol \alpha^*)=\left\{\frac{S\boldsymbol X}{\pi(\boldsymbol X,\boldsymbol \alpha^*)}-\boldsymbol X\right\}$.
From \citet{tsiatis2006semiparametric} under assumptions of Theorem \ref{thm7} and consistency of $\widehat{\boldsymbol \theta}$ for a two step Z-estimation problem, we obtain that 
$$\sqrt{N}(\widehat{\boldsymbol \theta}-\boldsymbol\theta^*)\xrightarrow{d}\mathcal{N}(0,(G_{\boldsymbol \theta^*})^{-1}.\mathbb{E}[\{\boldsymbol g(\boldsymbol \theta^*,\boldsymbol \alpha^*)+G_{\boldsymbol \alpha^*}.\boldsymbol \Psi(\boldsymbol \alpha^*)\}\{\boldsymbol g(\boldsymbol \theta^*,\boldsymbol \alpha^*)+G_{\boldsymbol \alpha^*}.\boldsymbol \Psi(\boldsymbol \alpha^*)\}'].(G_{\boldsymbol \theta^*}^{-1})^{'}\cdot$$
We will derive the expression of each of the terms in the above expression.
\begin{align*}
    & G_{\boldsymbol \theta^*}=\left.\mathbb{E}\left\{\frac{\partial g(\boldsymbol \theta,\boldsymbol \alpha^*)}{\partial \boldsymbol \theta}\right\}\right|_{\boldsymbol\theta=\boldsymbol \theta^*}\hspace{1cm} G_{\boldsymbol \alpha}=\left.G_{\boldsymbol \alpha^*}=\mathbb{E}\left\{\frac{\partial g(\boldsymbol \theta^*,\boldsymbol \alpha)}{\partial \boldsymbol \alpha}\right\}\right|_{\boldsymbol\alpha=\boldsymbol \alpha^*}\\
    & H = \left.\mathbb{E}\left\{\frac{\partial \boldsymbol h(\boldsymbol \alpha)}{\partial \boldsymbol \alpha}\right\}\right|_{\boldsymbol\alpha=\boldsymbol \alpha^*}\hspace{1.9cm} \boldsymbol \Psi(\boldsymbol \alpha^*)=-H^{-1}\boldsymbol h(\boldsymbol \alpha^*) \cdot
\end{align*}
First we calculate $G_{\boldsymbol \theta^*}$.
\subsubsection*{Calculation for $G_{\boldsymbol \theta}$}
\begin{align*}
    & \left.\frac{\partial \boldsymbol g(\boldsymbol \theta,\boldsymbol \alpha^*)}{\partial \boldsymbol \theta}\right|_{\boldsymbol\theta=\boldsymbol \theta^*}= \left.\frac{\partial}{\partial \boldsymbol \theta}\left[\frac{S}{\pi(\boldsymbol X,\boldsymbol \alpha^*)}\left\{D\boldsymbol Z-\frac{e^{\boldsymbol \theta^{*'}\boldsymbol Z}}{(1+e^{\boldsymbol \theta^{*'}\boldsymbol Z})}\cdot \boldsymbol Z\right\}\right]\right|_{\boldsymbol\theta=\boldsymbol \theta^*}\\
    & = -\frac{S}{\pi(\boldsymbol X,\boldsymbol \alpha^*)} \cdot \frac{e^{\boldsymbol \theta^{*'}\boldsymbol Z}}{(1+e^{\boldsymbol \theta^{*'}\boldsymbol Z})^2}\cdot \boldsymbol Z\boldsymbol Z' \implies G_{\boldsymbol \theta^*} = \mathbb{E}\left[-\frac{S}{\pi(\boldsymbol X,\boldsymbol \alpha^*)} \cdot \frac{e^{\boldsymbol \theta^{*'}\boldsymbol Z}}{(1+e^{\boldsymbol \theta^{*'}\boldsymbol Z})^2}\cdot \boldsymbol Z \boldsymbol Z'\right]\cdot
\end{align*}
Next we calculate $G_{\boldsymbol \alpha^*}$.
\subsubsection*{Calculation for $G_{\boldsymbol \alpha^*}$}
\begin{align*}
    & \left.\frac{\partial \boldsymbol g(\boldsymbol \theta^*,\boldsymbol \alpha)}{\partial \boldsymbol \alpha}\right|_{\boldsymbol\alpha=\boldsymbol \alpha^*}= \left.\frac{\partial}{\partial \boldsymbol \alpha}\left[\frac{S}{\pi(\boldsymbol X,\boldsymbol \alpha)}\left\{D\boldsymbol Z-\frac{e^{\boldsymbol \theta^{*'}\boldsymbol Z}}{(1+e^{\boldsymbol \theta^{*'}\boldsymbol Z})}\cdot \boldsymbol Z\right\}\right]\right|_{\boldsymbol\alpha=\boldsymbol \alpha^*}\\
    & = -\frac{S}{\pi(\boldsymbol X,\boldsymbol \alpha^*)^2}\cdot\left\{D\boldsymbol Z-\frac{e^{\boldsymbol \theta^{*'}\boldsymbol Z}}{(1+e^{\boldsymbol \theta^{*'}\boldsymbol Z})}\cdot \boldsymbol Z\right\}\cdot \left.\frac{\partial \pi(\boldsymbol X,\boldsymbol \alpha)}{\partial \boldsymbol \alpha}\right|_{\boldsymbol\alpha=\boldsymbol \alpha^*}\\
    & =-\frac{S}{\pi(\boldsymbol X,\boldsymbol \alpha^*)}\{1-\pi(\boldsymbol X,\boldsymbol \alpha^*)\}\left\{D-\frac{e^{\boldsymbol \theta^{*'}\boldsymbol Z}}{(1+e^{\boldsymbol \theta^{*'}\boldsymbol Z})}\right\}\boldsymbol Z\boldsymbol X'
\end{align*}
Therefore we obtain
\begin{align*}
     G_{\boldsymbol \alpha^*} = \mathbb{E}\left[-\frac{S}{\pi(\boldsymbol X,\boldsymbol \alpha^*)}\{1-\pi(\boldsymbol X,\boldsymbol \alpha^*)\}\left\{D-\frac{e^{\boldsymbol \theta^{*'}\boldsymbol Z}}{(1+e^{\boldsymbol \theta^{*'}\boldsymbol Z})}\right\}\boldsymbol Z\boldsymbol X'\right] \cdot
\end{align*}
Next we calculate $\Psi(\boldsymbol \alpha^*)$.
\subsubsection*{Calculation for $\boldsymbol\Psi(\boldsymbol \alpha^*)$}
\begin{align*}
    & \left.\frac{ \partial \boldsymbol h(\boldsymbol \alpha)}{\partial \boldsymbol \alpha}\right|_{\boldsymbol \alpha=\boldsymbol \alpha^*}= \left.\frac{ \partial}{\partial \boldsymbol \alpha}\left\{\frac{S\boldsymbol X}{\pi(\boldsymbol X,\boldsymbol \alpha)}-\boldsymbol X\right\}\right|_{\boldsymbol \alpha=\boldsymbol \alpha^*}=-\frac{S\boldsymbol X}{\pi(\boldsymbol X,\boldsymbol \alpha^*)^2}\cdot \left.\frac{ \partial}{\partial \alpha}\pi(\boldsymbol X,\boldsymbol \alpha)\right|_{\boldsymbol \alpha=\boldsymbol \alpha^*}\\
    & = -\frac{S\boldsymbol X\boldsymbol X'}{\pi(\boldsymbol X,\boldsymbol \alpha^*)^2}\cdot \pi(\boldsymbol X,\boldsymbol \alpha^*)\{1- \pi(\boldsymbol X,\boldsymbol \alpha^*\}\cdot
\end{align*}
This implies
\begin{align*}
   & H = \mathbb{E}\left[-\frac{S\boldsymbol X\boldsymbol X'}{\pi(\boldsymbol X,\boldsymbol \alpha^*)}\cdot \{1- \pi(\boldsymbol X,\boldsymbol \alpha^*\}\right]\cdot
\end{align*}
Therefore we obtain 
\begin{align*}
    & \boldsymbol \Psi(\boldsymbol \alpha^*)= \mathbb{E}\left[\frac{S\boldsymbol X\boldsymbol X'}{\pi(\boldsymbol X,\boldsymbol \alpha^*)}\cdot \{1- \pi(\boldsymbol X,\boldsymbol \alpha^*\}\right]^{-1}\cdot \left\{\frac{S\boldsymbol X}{\pi(\boldsymbol X,\boldsymbol \alpha^*)}-\boldsymbol X\right\}\cdot
\end{align*}
This gives the asymptotic distribution of $\widehat{\boldsymbol \theta}$ for PL.
\end{proof}
\noindent
Next we derive a consistent estimator of asymptotic variance of $\widehat{\boldsymbol\theta}$. Apart from assumptions of Theorem \ref{thm7} and \ref{thm8} we make the following additional assumption to derive a consistent estimator of asymptotic variance of $\widehat{\boldsymbol\theta}$.\\

\noindent
Let $\eta=(\boldsymbol \theta,\boldsymbol \alpha)$ and $\mathcal{N}=\Theta \times \Theta_{\boldsymbol \alpha}$.\\

\noindent
\textit{A}9.  Each of the following expectations are finite.
\begin{align*}
    & \mathbb{E}[\sup_{\boldsymbol \eta \in \mathcal{N}}(G_{\boldsymbol \theta})]<\infty \hspace{0.8cm} 
    \mathbb{E}[\sup_{\boldsymbol \eta \in \mathcal{N}}(G_{\boldsymbol \alpha})]<\infty \hspace{0.8cm} \mathbb{E}[\sup_{\boldsymbol \alpha \in \Theta_{\boldsymbol \alpha}}(H)]<\infty\\
    & \mathbb{E}\left(\sup_{\boldsymbol \eta \in \mathcal{N}}\left[S \cdot \left\{\frac{1}{\pi(\boldsymbol X,\boldsymbol \alpha)}\right\}^2\left\{D-\frac{e^{\boldsymbol \theta' \boldsymbol Z}}{(1+e^{\boldsymbol \theta'\boldsymbol Z})}\right\}^2\cdot \boldsymbol Z \boldsymbol Z'\right]\right)<\infty\\
    & \mathbb{E}\left(\sup_{\boldsymbol \eta \in \mathcal{N}}\left[\frac{S}{\pi(\boldsymbol X,\boldsymbol \alpha)} \cdot \left\{\frac{S}{\pi(\boldsymbol X,\boldsymbol \alpha)}-1\right\} \left\{D-\frac{e^{\boldsymbol \theta'\boldsymbol Z}}{(1+e^{\boldsymbol \theta' \boldsymbol Z})}\right\}\cdot \boldsymbol X\boldsymbol Z'\right]\right)<\infty\\
    & \mathbb{E}\left(\sup_{\boldsymbol \alpha \in \Theta_{\boldsymbol \alpha}}\left[\frac{S}{\pi(\boldsymbol X,\boldsymbol \alpha)^2}\cdot \{1-\pi(\boldsymbol X,\boldsymbol \alpha)\}\cdot \boldsymbol X \boldsymbol X' \right]\right) <\infty \cdot
\end{align*}
\begin{theorem}
    Under all the assumptions of Theorems \ref{thm7}, \ref{thm8} and \textit{A}9, $\frac{1}{N} \cdot \widehat{G_{\boldsymbol \theta}}^{-1}\cdot \widehat{E}\cdot (\widehat{G_{\boldsymbol \theta}}^{-1})^{'}$ is a consistent estimator of the variance of $\widehat{\boldsymbol \theta}$ for CL where
\begin{align*}
   & \widehat{G_{\boldsymbol \theta}}=-\frac{1}{N}\sum_{i=1}^N S_i \cdot \frac{1}{\pi(\boldsymbol X_i,\widehat{\boldsymbol \alpha})}\cdot \frac{e^{\widehat{\boldsymbol \theta}'\boldsymbol Z_i}}{(1+e^{\widehat{\boldsymbol \theta}'\boldsymbol Z_i})^2}\cdot\boldsymbol Z_i \boldsymbol Z_i'\\ 
   &\widehat{H} = -\frac{1}{N}\sum_{i=1}^N \frac{S_i}{\pi(\boldsymbol X_i,\widehat{\boldsymbol \alpha})} \cdot \{1-\pi(\boldsymbol X_i,\widehat{\boldsymbol \alpha})\}\cdot \boldsymbol X_i \cdot \boldsymbol X_i'\\
   & \widehat{G_{\boldsymbol \alpha}}=-\frac{1}{N}\sum_{i=1}^N S_i \cdot \frac{\{1-\pi(\boldsymbol X_i,\widehat{\boldsymbol \alpha})\}}{\pi(\boldsymbol X_i,\widehat{\boldsymbol \alpha})}\cdot\left\{D_i-\frac{e^{\widehat{\boldsymbol \theta}'\boldsymbol Z_i}}{(1+e^{\widehat{\boldsymbol \theta}'\boldsymbol Z_i})}\right\}\boldsymbol Z_i \boldsymbol X_i'\\
   & \widehat{E}_1 =\frac{1}{N}\sum_{i =1}^N S_i \cdot \left\{\frac{1}{\pi(\boldsymbol X_i,\widehat{\boldsymbol \alpha})}\right\}^2\left\{D_i-\frac{e^{\widehat{\boldsymbol \theta}'\boldsymbol Z_i}}{(1+e^{\widehat{\boldsymbol\theta}'\boldsymbol Z_i})}\right\}^2\cdot \boldsymbol Z_i \boldsymbol Z_i'\\
    & \widehat{E}_2 =  \widehat{E}_3'=\frac{1}{N}\sum_{i=1}^N \frac{S_i}{\pi(\boldsymbol X_i,\widehat{\boldsymbol \alpha})} \cdot \widehat{G_{\boldsymbol \alpha}}\cdot \widehat{H}^{-1} \cdot \left\{\frac{S_i}{\pi(\boldsymbol X_i,\widehat{\boldsymbol \alpha})}-1\right\} \left\{D_i-\frac{e^{\widehat{\boldsymbol \theta}'\boldsymbol Z_i}}{(1+e^{\widehat{\boldsymbol \theta}' \boldsymbol Z_i})}\right\}\cdot \boldsymbol X_i\boldsymbol Z_i'\\
    & \widehat{E_4} = \frac{1}{N} \sum_{i=1}^N\widehat{G_{\boldsymbol \alpha}} \cdot \widehat{H}^{-1}\left[\frac{S_i}{\pi(\boldsymbol X_i,\widehat{\boldsymbol \alpha})^2}\cdot \{1-\pi(\boldsymbol X_i,\widehat{\boldsymbol \alpha})\}\right]\cdot \boldsymbol X_i \cdot \boldsymbol X_i' \cdot (\widehat{H}^{-1})'\cdot(\widehat{G_{\boldsymbol \alpha}})'\\
    &\widehat{E}=\widehat{E}_1-\widehat{E}_2-\widehat{E}_3 + \widehat{E}_4\cdot
\end{align*}\label{thm9}
\end{theorem}
\begin{proof}
    Under all the assumptions of Theorems \ref{thm7}, \ref{thm8} and \textit{A}9 using ULLN and Continuous Mapping Theorem, the proof of consistency for each of the following sample quantities are exactly same as the approach in Theorem \ref{thm3}. Using the exact same steps on the joint parameters $\boldsymbol \eta$ instead of $\boldsymbol \theta$ (as in Theorem \ref{thm3}) we obtain
    \begin{align*}
     & \widehat{G_{\boldsymbol \theta}}=-\frac{1}{N}\sum_{i=1}^N S_i \cdot \frac{1}{\pi(\boldsymbol X_i,\hat{\boldsymbol \alpha})}\cdot \frac{e^{\widehat{\boldsymbol \theta}'\boldsymbol Z_i}}{(1+e^{\widehat{\boldsymbol \theta}'\boldsymbol Z_i})^2}\cdot\boldsymbol Z_i \boldsymbol Z_i' \\
     & \xrightarrow{p} G_{\theta} = \mathbb{E}\left\{-\frac{S}{\pi(\boldsymbol X,\boldsymbol\alpha^*)}\cdot \frac{e^{\boldsymbol \theta^{*'}\boldsymbol Z}}{(1+e^{\boldsymbol \theta^{*'}\boldsymbol Z})^2}.\boldsymbol Z\boldsymbol Z'\right\}\cdot
\end{align*}
Similarly we obtain
\begin{align*}
     & \widehat{G_{\boldsymbol \alpha}}=-\frac{1}{N}\sum_{i=1}^N S_i \cdot \frac{\{1-\pi(\boldsymbol X_i,\widehat{\boldsymbol \alpha})\}}{\pi(\boldsymbol X_i,\widehat{\boldsymbol \alpha})}\cdot\left\{D_i-\frac{e^{\widehat{\boldsymbol \theta}' \boldsymbol Z_i}}{(1+e^{\widehat{\boldsymbol \theta}'\boldsymbol Z_i})}\right\}\boldsymbol Z_i \boldsymbol X_i' \\
     & \xrightarrow{p} G_{\boldsymbol \alpha} = \mathbb{E}\left[-\frac{S}{\pi(\boldsymbol X,\boldsymbol \alpha^*)}\{1-\pi(\boldsymbol X,\boldsymbol \alpha^*)\}\left\{D-\frac{e^{\boldsymbol \theta^{*'}\boldsymbol Z}}{(1+e^{\boldsymbol \theta^{*'}\boldsymbol Z})}\right\}\boldsymbol Z\boldsymbol X'\right]\cdot
\end{align*}
\begin{align*}
     & \widehat{H}= -\frac{1}{N}\sum_{i=1}^N \frac{S_i}{\pi(\boldsymbol X_i,\hat{\boldsymbol \alpha})} \cdot \{1-\pi(\boldsymbol X_i,\hat{\boldsymbol \alpha})\}\cdot \boldsymbol X_i \cdot \boldsymbol X_i'\\
     & \xrightarrow{p} H = \mathbb{E}\left[\frac{S}{\pi(\boldsymbol X,\boldsymbol \alpha^*)}\cdot\{1-\pi(\boldsymbol X,\boldsymbol \alpha^*)\}\cdot\boldsymbol X\boldsymbol X'\right]\cdot
\end{align*}
\begin{align*}
     &  \widehat{E}_1 =\frac{1}{N}\sum_{i =1}^N S_i \cdot \left\{\frac{1}{\pi(\boldsymbol X_i,\hat{\boldsymbol \alpha})}\right\}^2\left\{D_i-\frac{e^{\widehat{\boldsymbol \theta}' \boldsymbol Z_i}}{(1+e^{\hat{\theta}'\boldsymbol Z_i})}\right\}^2\cdot \boldsymbol Z_i \boldsymbol Z_i'\\
     & \xrightarrow{p}  \mathbb{E}\left[S \cdot \left\{\frac{1}{\pi(\boldsymbol X,\boldsymbol \alpha^*)}\right\}^2\cdot \left\{D-\frac{e^{\boldsymbol \theta^{*'}\boldsymbol Z_i}}{(1+e^{\boldsymbol \theta^{*'}\boldsymbol Z})}\right\}^2\cdot \boldsymbol Z\boldsymbol Z'\right]\cdot
\end{align*}
\begin{align*}
     & \widehat{E}_2 =  \widehat{E}_3'=\frac{1}{N}\sum_{i=1}^N \frac{S_i}{\pi(\boldsymbol X_i,\widehat{\boldsymbol \alpha})} \cdot \widehat{G_{\boldsymbol \alpha}}\cdot \widehat{H}^{-1} \cdot \left\{\frac{S_i}{\pi(\boldsymbol X_i,\widehat{\boldsymbol \alpha})}-1\right\} \left\{D_i-\frac{e^{\widehat{\boldsymbol \theta}'\boldsymbol Z_i}}{(1+e^{\widehat{\boldsymbol \theta}'\boldsymbol Z_i})}\right\}\cdot \boldsymbol X_i\boldsymbol Z_i'\\
    &\xrightarrow{p} \mathbb{E}\{G_{\boldsymbol \alpha^*}\cdot H^{-1}\cdot \boldsymbol h(\boldsymbol \alpha^*)\cdot \boldsymbol g(\boldsymbol \theta^*,\boldsymbol \alpha^*)'\}\cdot
\end{align*}
\begin{align*}
    & \widehat{E_4} = \frac{1}{N} \cdot \widehat{G_{\boldsymbol \alpha}} \cdot \widehat{H}^{-1}\left[.\sum_{i=1}^N \frac{S_i}{\pi(\boldsymbol X_i,\widehat{\boldsymbol \alpha})^2}\cdot \boldsymbol X_i \cdot \boldsymbol X_i'\right.\\
    & \left. -2 \cdot \sum_{i=1}^ N\frac{S_i}{(\boldsymbol X_i,\widehat{\boldsymbol \alpha})}\cdot \boldsymbol X_i \cdot \boldsymbol X_i'+ \frac{S_i}{\pi(\boldsymbol X_i,\widehat{\boldsymbol \alpha})}\cdot \boldsymbol X_i \cdot \boldsymbol X_i'\right]\cdot(\widehat{H}^{-1})'\cdot(\widehat{G_{\boldsymbol \alpha}})' \\
    & = \frac{1}{N} \sum_{i=1}^N\widehat{G_{\boldsymbol \alpha}} \cdot \widehat{H}^{-1}\left[\frac{S_i}{\pi(\boldsymbol X_i,\hat{\boldsymbol \alpha})^2}\cdot \{1-\pi(\boldsymbol X_i,\widehat{\boldsymbol \alpha})\}\right]\cdot \boldsymbol X_i \cdot \boldsymbol X_i' \cdot (\widehat{H}^{-1})'\cdot(\widehat{G_{\boldsymbol \alpha}})'\\
    & \xrightarrow{p} E_4 =  \mathbb{E}\{\widehat{G_{\boldsymbol \alpha^*}}\cdot \widehat{H}^{-1}\cdot \boldsymbol h(\boldsymbol \alpha^*)\cdot \boldsymbol h(\boldsymbol \alpha^*)'\cdot (\widehat{H}^{-1})'\cdot (\widehat{G_{\boldsymbol \alpha^*}})'\}
\end{align*}
Therefore we obtain
\begin{align*}
    \widehat{E}=\widehat{E}_1-\widehat{E}_2-\widehat{E}_3 + \widehat{E}_4\xrightarrow{p} \mathbb{E}[\{\boldsymbol g(\boldsymbol \theta^*,\boldsymbol \alpha^*)+G_{\boldsymbol \alpha^*}\cdot \boldsymbol \Psi(\boldsymbol \alpha^*)\}\{\boldsymbol g(\boldsymbol \theta^*,\boldsymbol \alpha^*)+G_{\boldsymbol \alpha^*}.\boldsymbol \Psi(\boldsymbol \alpha^*)\}']\cdot
\end{align*}
From Theorem \ref{thm8} using the same approach used in the last step of Theorem \ref{thm3}, we obtain that $\frac{1}{N}\cdot\widehat{G_{\boldsymbol \theta}}^{-1}\cdot\widehat{E}\cdot(\widehat{G_{\boldsymbol \theta}}^{-1})^{'}$ is a consistent estimator of the asymptotic variance of $\widehat{\boldsymbol \theta}$ using CL.
\end{proof}
\subsection{Criteria for coarsening variables for PS method}\label{sec:coarse}
For PS method, the selection probability of $(D,Z_2',W^{'})$ are available where both $Z_2'$ and $W^{'}$ are the coarsened versions of $Z_2$ and $W$. For any continuous random variable say, $L$, a coarsened version $L'$ is defined as, 
\[   
L' = 
     \begin{cases}
      0 &\text{if}\hspace{0.2cm} L<\text{Cutoff}_1\\
       1 &\text{if}\hspace{0.2cm}<\text{Cutoff}_1<=L<=<\text{Cutoff}_2 \\
       \vdots & \vdots\\
      K-1 &\text{if}\hspace{0.2cm} L>\text{Cutoff}_K\\
     \end{cases}
\]
\doublespacing
In this simulation for both $Z_2$ and $W$, we chose $K=2$ and $\text{Cutoff}_1=\epsilon_{0.15}$, $\text{Cutoff}_2=\epsilon_{0.85}$, where $\epsilon_{0.15}$ and $\epsilon_{0.85}$ are the 15 and 85 percentile quartiles for both the continuous random variables respectively.

\subsection{Real Data Cancer Weights} \label{sec:realcancerweights}
As described in Section \ref{sec:realex}, we first estimated the weights for PL and SR without using cancer $w_o$ as a selection variable. Then we used equation \eqref{eq:eqcwt} to modify the weights. The proof of equation \eqref{eq:eqcwt} is presented in this section. At first we note that $1/w_0=P(S=1|\boldsymbol W,Z_2)$. Using Bayes theorem we obtain
\begin{align*}
    P(S=1|D,\boldsymbol W,Z_2)&=\frac{P(D|S=1,\boldsymbol W,Z_2)P(\boldsymbol W,Z_2,S=1)}{P(D|\boldsymbol W,Z_2)P(\boldsymbol W,Z_2)}\\
    &= \frac{P(D|S=1,\boldsymbol W,Z_2)P(S=1|\boldsymbol W,Z_2)P(\boldsymbol W,Z_2)}{P(D|\boldsymbol W,Z_2)P(\boldsymbol W,Z_2)}\\
    &= \frac{P(D|S=1,\boldsymbol W,Z_2)}{P(D|\boldsymbol W,Z_2)}\cdot P(S=1|\boldsymbol W,Z_2)\\
    &= \frac{P(D|S=1,\boldsymbol W,Z_2)}{P(D|\boldsymbol W,Z_2)}\cdot \frac{1}{w_0}\cdot
\end{align*}

\subsection{Details of different subsampling strategies} \label{sec:subdet}
$$\text{logit}(P(S_{\text{sub}}=1|\text{cancer}, \text{sex}, \text{diabetes}))=\mu_0+\mu_1\cdot\text{cancer} + \mu_2\cdot\text{diabetes} +\mu_3 \cdot \text{sex}\cdot$$
We selected $\text{diabetes}$ to represent a part of $\boldsymbol W$ which has a significant association with cancer. The four different subsampling strategies are given by
\begin{itemize}
    \item \textbf{Random Sampling:} $S_{\text{sub}}$ does not depend on any variables. In this case $\mu_0=-0.5$, $\mu_1=0$, $\mu_2=0$ and $\mu_3=0$. 
    \item \textbf{Only cancer:} $S_{\text{sub}}$ only depend on cancer. In this case $\mu_0=-1$, $\mu_1=1$, $\mu_2=0$ and $\mu_3=0$. 
    \item \textbf{Cancer and sex:} $S_{\text{sub}}$ depends on cancer and sex. In this case $\mu_0=-0.5$, $\mu_1=1$, $\mu_2=0$ and $\mu_3=-1$. 
    \item \textbf{Cancer, sex and diabetes:} $S_{\text{sub}}$ depends on cancer, sex and diabetes. In this case $\mu_0=-1$, $\mu_1=1$, $\mu_2=1$ and $\mu_3=-1$. 
\end{itemize}
\subsection*{Supplementary Figures}
\begin{figure}[H]
    \centering
    \includegraphics[scale=0.65]{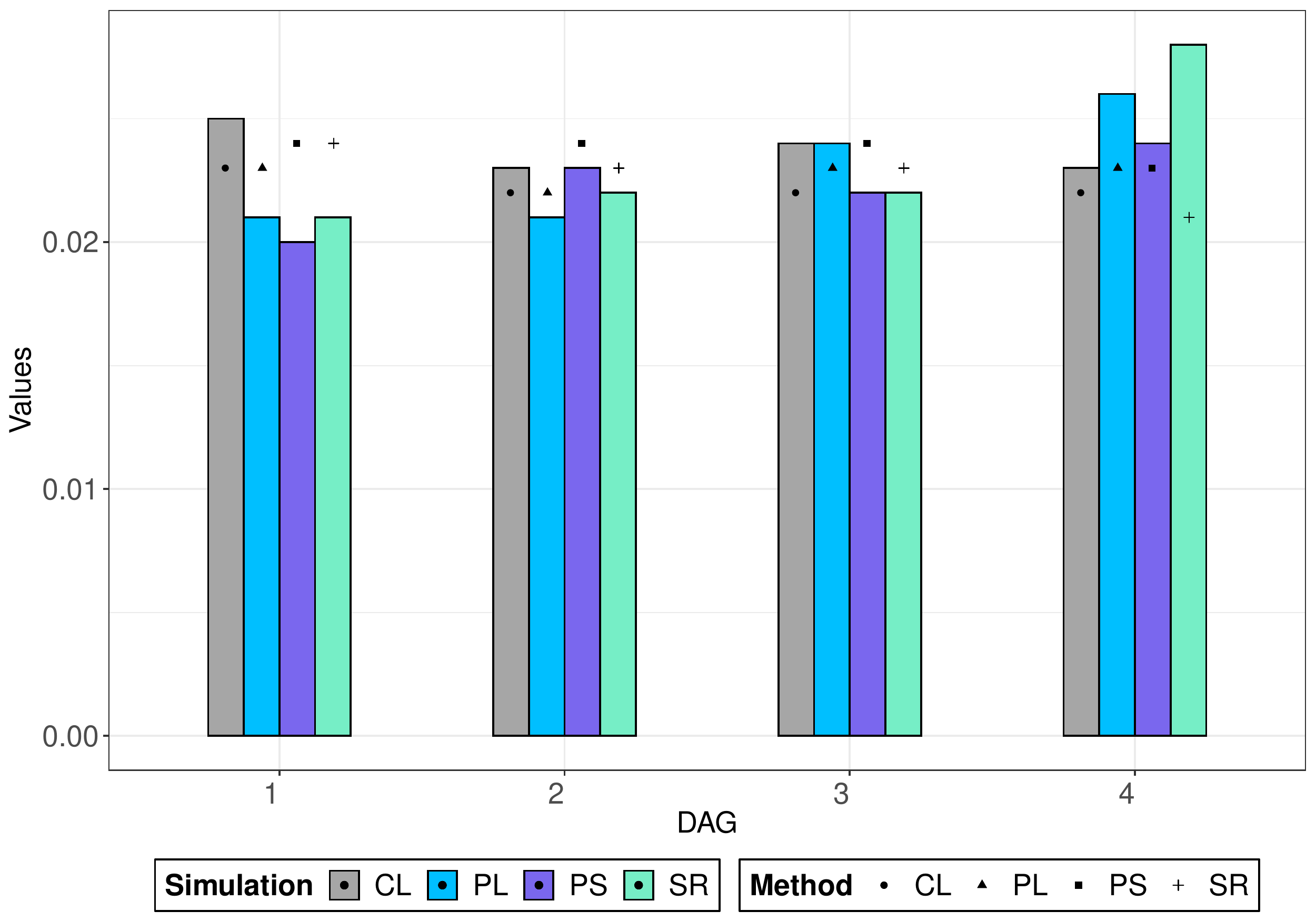}
    \caption{Comparison between empirical monte carlo variance and proposed variance estimators of $\widehat{\theta_2}$ using the four weighted methods  under the four DAGs in Setup 1. The colored bars denote the empirical variances and the dots represent the estimated ones. SR : Simplex regression, PL : Pseudolikelihood, PS : Post Stratification and CL : Calibration.}\label{fig:Var}
\end{figure}
\newpage
\begin{figure}[H]
    \centering
    \includegraphics[width=\linewidth]{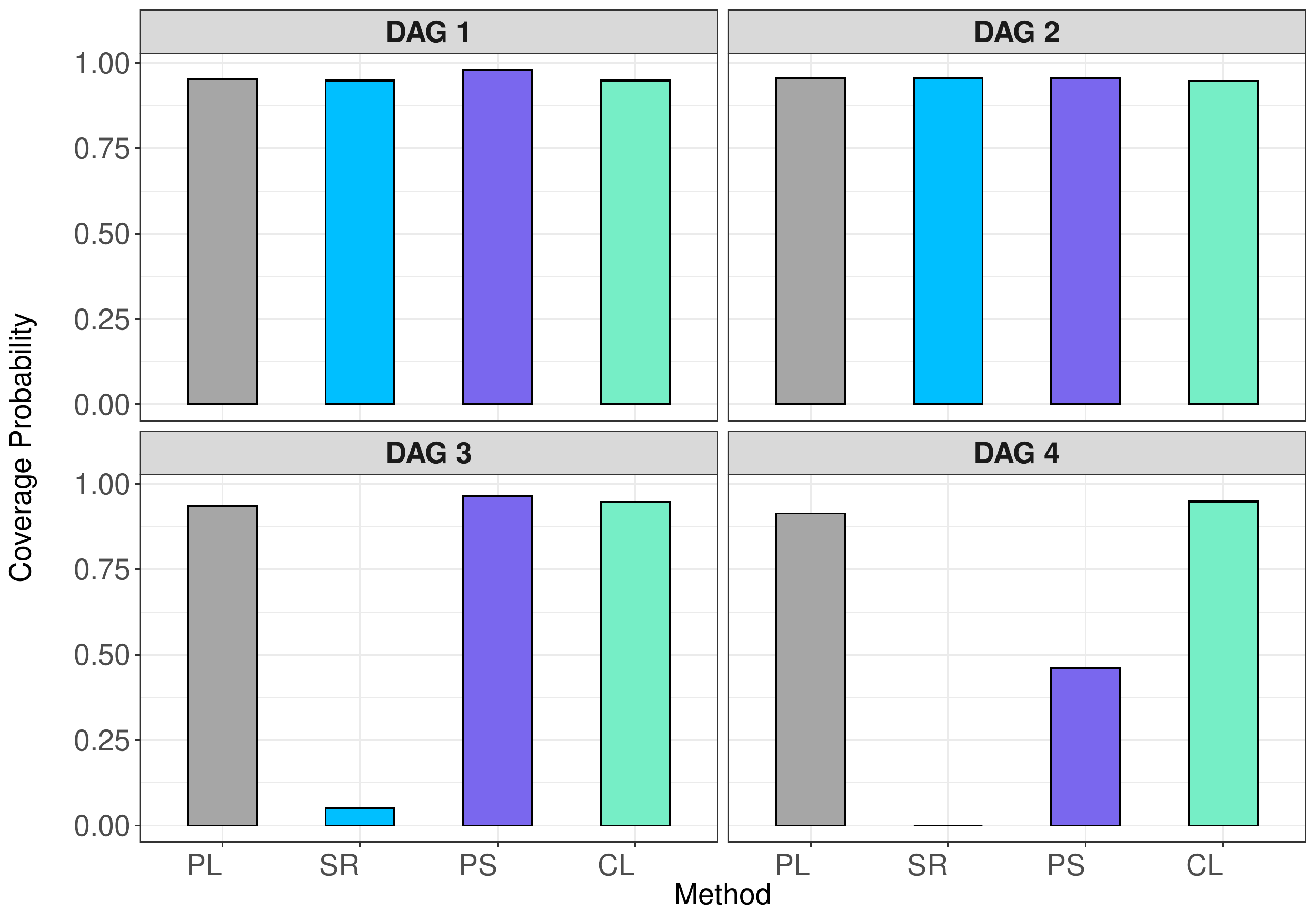}
    \caption{The coverage probabilities for the four weighted methods using the proposed variance estimators of $\widehat{\theta_2}$ under the four DAGs in Setup 1. SR : Simplex regression, PL : Pseudolikelihood, PS : Post Stratification and CL : Calibration.}\label{fig:cp}
\end{figure}
% \usepackage{multirow}

\end{document}